%% file: LPHS.tex
\documentclass[11pt]{article}

\usepackage{fullpage}

\usepackage[titletoc,title]{appendix}
\usepackage{mathtools}
\usepackage{amsfonts}
\usepackage{graphicx}
\usepackage{hyperref}
\usepackage{cleveref}
\usepackage{amsmath}
\usepackage{amssymb}
\usepackage{framed}
\usepackage{amsthm}
\usepackage{bbm}
\usepackage{verbatim}
\usepackage[linesnumbered,ruled]{algorithm2e}
\usepackage[usenames,dvipsnames,svgnames,table]{xcolor}

\newcommand{\ynote}[1]{{\color{blue} Yuval: #1}}

\newtheorem{theorem}{Theorem}[section]
\newtheorem{proposition}[theorem]{Proposition}
\newtheorem{lemma}[theorem]{Lemma}
\newtheorem{claim}[theorem]{Claim}
\newtheorem{corollary}[theorem]{Corollary}

\newtheorem{definition}[theorem]{Definition}

\theoremstyle{remark}\newtheorem*{remark}{\textbf{Remark}}
\renewcommand{\paragraph}[1]{\vspace{.1in}\noindent{\bf #1}}

\newtheorem{notation}[theorem]{Notation}

\newtheorem*{proposition*}{Proposition}

\newtheorem{construction}[theorem]{Construction}

\usepackage{enumitem}
\newlist{casenv}{enumerate}{4}
\setlist[casenv]{leftmargin=*,align=left,widest={iiii}}
\setlist[casenv,1]{label={{\itshape\ \casename} \arabic*.},ref=\arabic*}
\setlist[casenv,2]{label={{\itshape\ \casename} \roman*.},ref=\roman*}
\setlist[casenv,3]{label={{\itshape\ \casename\ \alph*.}},ref=\alph*}
\setlist[casenv,4]{label={{\itshape\ \casename} \arabic*.},ref=\arabic*}
\providecommand{\casename}{Case}

\crefformat{footnote}{#2\footnotemark[#1]#3}

\newcommand{\mrm}[1]{\mathrm{#1}}

\newcommand{\bem}[0]{\begin{eqnarray*}}
\newcommand{\enm}[0]{\end{eqnarray*}}
\newcommand{\anderbrace}[2]{
	\if\relax\detokenize{#2}\relax
		\sbox0{$\underbrace{#1}_{}$}
		\mathrel{\mathmakebox[\wd0]{#1}}
	\else
		\mathrel{\underbrace{#1}_{\mathclap{#2}}}
	\fi}

\newcommand{\integers}{\mathbb{Z}}
\newcommand{\pintegers}{\mathbb{N}}

\newcommand{\reals}{\mathbb{R}}

\newcommand{\spm}{\left\{ -1,1\right\}  }
\newcommand{\zo}{\left\{  0,1\right\}  }

\newcommand{\sm}{\setminus}

\newcommand{\es}{\emptyset}
\newcommand{\ddd}{\doteqdot}

\newcommand{\one}{\mathbbm{1}}

\newcommand{\eps}{\epsilon}

\newcommand{\pr}{\Pr}
\newcommand{\given}[2]{#1\,\middle|\, #2}

\DeclareMathOperator*{\be}{\mathbb{E}}
\DeclareMathOperator*{\var}{\mrm{Var}}
\DeclareMathOperator*{\cov}{\mrm{Cov}}

\DeclareMathOperator{\E}{E}
\DeclareMathOperator{\HH}{H}
\DeclareMathOperator{\I}{I}

\newcommand{\wh}[1]{\widehat{#1}}

\newcommand{\lbr}{\left(}
\newcommand{\rbr}{\right)}
\newcommand{\lbs}{\left[}
\newcommand{\rbs}{\right]}

\newcommand{\lba}{\left|}
\newcommand{\rba}{\right|}

\newcommand{\reddit}[1]{\begin{color}{purple}{#1}\end{color}}
\definecolor{ohadcolor}{RGB}{0, 100, 250}

\newcommand{\func}[3]{{#1}\colon {#2} \to {#3}}

\newcommand{\restrict}[2]{{
  \left.\kern-\nulldelimiterspace 
  #1 
  \vphantom{\big|} 
  \right|_{#2} 
}}

\newcommand{\len}{b}

\newcommand{\calH}{\mathcal{H}}

\newcommand{\Good}{$\alpha$-Good}
\newcommand{\good}{$\alpha$-good}
\newcommand{\sfGood}{{\sf Good}^{\alpha}_n}

\newcommand{\Goodw}{$(\alpha,W)$-Good}
\newcommand{\goodw}{$(\alpha,W)$-good}
\newcommand{\sfGoodw}{{\sf Good}^{\alpha,W}_n}
\newcommand{\sfGoodwp}{{\sf Good}^{\alpha,W}_{n'}}
\newcommand{\sfGoodwpp}{{\sf Good}^{\alpha/2,W}_{n'+r}}

\newcommand{\HShare}{{\sf Share}}
\newcommand{\HEval}{{\sf Eval}}
\newcommand{\G}{G}	
\newcommand{\cD}{\mathcal{D}}

\renewenvironment{proof}[1][]
    {\noindent
       \ifx&#1&{\it Proof.}
       \else{\it Proof ({#1}).}
       \fi}{\hfill $\blacksquare$}

\title{Locality-Preserving Hashing for Shifts with \\Connections to Cryptography\thanks{This is a full version of~\cite{LPHSconf}.}}

\author{Elette Boyle\thanks{IDC Herzliya, Israel and NTT Research, USA. \tt{elette.boyle@idc.ac.il}}
  \and Itai Dinur\thanks{Ben-Gurion University, Be'er Sheva, Israel. \tt{dinuri@cs.bgu.ac.il}}
  \and Niv Gilboa\thanks{Ben-Gurion University, Be'er Sheva, Israel. \tt{gilboan@bgu.ac.il}}
  \and Yuval Ishai \thanks{Technion, Haifa Israel. \tt{yuvali@cs.technion.ac.il}}
  \and Nathan Keller\thanks{Bar-Ilan University, Ramat Gan, Israel. \tt{nathan.keller27@gmail.com}} 
  \and Ohad Klein\thanks{Bar-Ilan University, Ramat Gan, Israel. \tt{ohadkel@gmail.com}}}

\begin{document}

\maketitle

\begin{abstract}
Can we sense our location in an unfamiliar environment by taking a sublinear-size sample of our surroundings? Can we efficiently encrypt a message that only someone physically close to us can decrypt? To solve this kind of problems, we introduce and study a new type of hash functions for finding shifts in sublinear time. A function $h:\{0,1\}^n\to \integers_n$ is a $(d,\delta)$ {\em locality-preserving hash function for shifts} (LPHS) if: (1) $h$ can be computed by (adaptively) querying $d$ bits of its input, and (2) $\pr \lbs h(x) \neq h(x \ll 1) + 1 \rbs \leq \delta$, where $x$ is random and $\ll 1$ denotes a cyclic shift by one bit to the left. 
We make the following contributions.
\begin{itemize}
\item {\bf Near-optimal LPHS via Distributed Discrete Log.}  We establish a general two-way connection between LPHS and algorithms for {\em distributed discrete logarithm} in the generic group model. Using such an algorithm of Dinur et al.\ (Crypto 2018), we get  LPHS with near-optimal error of $\delta=\tilde O(1/d^2)$. 
This  gives an unusual example for the usefulness of group-based cryptography in a post-quantum world. 
We extend the positive result to non-cyclic and worst-case variants of LPHS. 

  
\item {\bf Multidimensional LPHS.}  We obtain positive and negative results for a multidimensional extension of LPHS, making progress towards an optimal 2-dimensional LPHS.

\item {\bf Applications.} We demonstrate the usefulness of LPHS by presenting cryptographic and algorithmic applications.  In particular, we apply multidimensional LPHS to obtain an efficient ``packed'' implementation of  {\em homomorphic secret sharing} and a sublinear-time implementation of {\em location-sensitive encryption} whose decryption requires a significantly overlapping view. 

\end{itemize}
\end{abstract}
\thispagestyle{empty}



\newpage
\tableofcontents
\newpage


\input{Introduction}

\input{Preliminaries}

\input{Reductions} 
\input{LPHS_DDL}
\input{2D-LPHS-compressed-ICALP}


\input{WorstCase}

\input{Applications}

\paragraph{Acknowledgements.} We thank Piotr Indyk, Leo Reyzin, David Woodruff, and anonymous reviewers for helpful pointers and suggestions.

E.\ Boyle was supported by AFOSR Award FA9550-21-1-0046, ERC Project HSS (852952), and a Google Research Scholar Award.
I.\ Dinur was supported by ISF grant 1903/20 and ERC starting grant 757731 (LightCrypt).
N.\ Gilboa was supported by ISF grant 2951/20, ERC grant 876110, and a grant by the BGU Cyber Center.
Y.\ Ishai was supported by ERC Project NTSC (742754), ISF grant 2774/20, and BSF grant 2018393.
N.\ Keller was supported by ERC starting grant 757731 (LightCrypt) and by the BIU Center for Research in Applied Cryptography and Cyber Security in conjunction with the Israel National Cyber Bureau in the Prime Minister’s Office.
O.\ Klein was supported by the Clore Scholarship Programme.

\bibliographystyle{plainurl}
\bibliography{LPHS}

%
%
%
%
%
%
%
%
%
%
%
%
%
%

\appendix
\section*{Appendix}

\input{DKK-Results}

\end{document}

%% file: Introduction.tex

\section{Introduction}

A {\em locality-preserving hash function}~\cite{LinialS96,IndykMRV97}  is a distance-respecting mapping from a complex input space to a simpler output space.
Inspired by recent results in cryptography, we study a new kind of locality-preserving hash functions that map strings to integers while respecting the {\em shift} distance
between pairs of input strings with high probability. A distinctive feature of these hash functions is that they can be computed in {\em sublinear time} with low error probability.

\paragraph{\bf Why shifts? Why sublinear?}
Our hash functions for shifts can be thought of as {\em sublinear-time location sensors} that measure a relative position in an unfamiliar environment by taking a sublinear-size sample of the surroundings. This can apply in a variety of settings. For instance, ``surroundings'' may refer to a local view of an unexplored territory, a long string such as a DNA sequence, an external signal such as a GPS synchronization sequence, a digital document such as big pdf file or a virtual world, or a huge mathematical object such as a cryptographic group. See~\cite{AndoniIKH13,ohlsson2014compressive}  for applications of shift finding to GPS synchronization, image alignment,  motion estimation, and more.\footnote{While previous related works study a {\em noise tolerant} variant of shift distance, which arises naturally in the applications they consider, in this work we focus on the simpler noiseless case. Beyond theoretical interest, the simpler notion is motivated by applications. For instance, a local view of a digital document or a mathematical object is noiseless. The noisy case is studied in a follow-up work~\cite{DKK-followup}, which obtains nearly tight bounds on the (sublinear) amount of random noise that can be tolerated.
}
We will discuss additional cryptographic and algorithmic applications in Section~\ref{sec-introapps} below.
We are motivated by scenarios in which the local view contains an enormous amount of relevant information that cannot be naively sub-sampled or compressed. This calls for {\em sublinear-time} solutions.  

\paragraph{\bf Simple shift-finding solutions.} To motivate the new primitive,
consider the following simple shift-finding problem. An $n$-bit string $x$ is picked uniformly at random, and then cyclically shifted by $s$ bits the left, for some $0\le s < n$. Let $y$ be the resulting string. For instance, $x,y$ may be obtained by measuring the same periodic signal at different phases. We write $y=x\ll s$. The
shift-finding problem is to find the shift amount $s$ given $x$ and $y$. 

In a centralized setting, where $x$ and $y$ are both given as inputs, it is easy to solve the problem  in sublinear time (with small error probability), querying only $\tilde O(n^{1/2})$ bits of the input, by matching substrings of $x$ of length $\ell=O(\log n)$ starting at positions $1,2,\ldots, \sqrt{n}$ with length-$\ell$ substrings of $y$ whose starting position is a multiple of $\sqrt{n}$. (This is a simplified version of a noise-resilient algorithm from~\cite{AndoniIKH13}.) This algorithm is nearly optimal, since any shift-finding algorithm
for an unbounded shift amount $s$
should read $\Omega(\sqrt{n})$ bits of the input~\cite{BatuEKMRRS03}.

In a distributed setting, a natural goal is to design a {\em sketching} algorithm that compresses a single input into a short sketch, such that given the sketches of $x$ and $y$ one can recover $s$ with high probability. Note that the previous centralized algorithm does imply such a sublinear-size sketch, but only with $\tilde O(\sqrt{n})$ output size, which is far from optimal. Instead, one could use the following classical approach~\cite{Cover}:
let the sketch of $x$ be an integer $0\le z_x<n$ that minimizes $x\ll z_x$ (viewed as an $n$-bit integer), and similarly for $y$. It can be easily seen that $s=z_x-z_y \mod n$ whenever the minimum is uniquely defined.

The logarithmic sketch size of this simple solution is clearly optimal. Moreover, it realizes something even stronger than sketching: a hash function $h:\{0,1\}^n\to\integers_n$ that respects cyclic shifts in the sense that for a random input $x$,  we have $h(x)=h(x\ll s)+s$ except with small probability. That is, shifting the input by $s$ changes the output by $s$ in the same direction. This is useful for applications. For instance, given $t$ hashes $z_i=h(y_i)$, where $y_i=x\ll s_i$ for $i=1,\ldots,t$, one can easily compute in time $\tilde O(t)$ the relative offsets of all $y_i$.


The main downside of the above hashing-based solution compared to the centralized algorithm is its linear running time. A natural question is whether one can enjoy the best of both worlds:
\begin{quote}
{\em Can we combine the sublinear running time of the centralized algorithm with the optimal sketch size and locality-sensing features of the hashing-based solution?}
\end{quote}

\subsection{Our Contribution}

We initiate a  study of hashing-based solutions to the shift-finding problem. We capture such solutions via the following notion of locality-preserving hash function for shifts.

\begin{definition}
A function $h:\{0,1\}^n\to \integers_n$ is a $(d,\delta)$ {\em locality-preserving hash function for shifts} (LPHS) if: (1) $h$ can be computed by (adaptively) querying $d$ bits of its input, and (2) $\pr \lbs h(x) \neq h(x \ll 1) + 1 \rbs \leq \delta$, where $x$ is random and $\ll 1$ denotes a cyclic shift by one bit to the left.
\end{definition}
Note that, by a union bound, an LPHS as above satisfies $\pr \lbs h(x) \neq h(x \ll s) + s \rbs \leq s\cdot \delta$ for any shift amount $0\le s<n$. Thus, an LPHS has a better accuracy guarantee for smaller shifts. Intuitively, an LPHS can be thought of as a sublinear-time computable location identifier that suffices (with high probability) for determining the exact relative location with respect to adjacent identifiers.


\paragraph{Other LPHS flavors.} The above notion of LPHS addresses the basic shift-finding problem as discussed above, but is limited in several important ways: it only considers cyclic shifts and 1-dimensional inputs,
and it only guarantees average-case correctness for uniformly random inputs. To address these limitations, we additionally consider other flavors of the basic LPHS notion defined above that are more suitable for applications. These include a  {\em non-cyclic} variant, where instead of $x\ll 1$ we remove the leftmost bit of $x$ and add a random bit on the right; a {\em $k$-dimensional} variant, where the input is a $k$-dimensional matrix and the output is in $\integers_n^k$;
and a {\em worst-case} variant where the quantification is over an arbitrary $x$ that is ``far from periodic'' and the probability is over the choice of $h$. (The latter variant better corresponds to the typical notion of a randomized hash function.) The applications we present crucially depend on these extensions.

\subsubsection {Near-Optimal LPHS via Distributed Discrete Log}

We establish a general two-way connection between LPHS and algorithms for the {\em distributed discrete logarithm} (DDL) problem~\cite{BGI16}. Before explaining this connection, we start with relevant background.

The traditional discrete logarithm (DL) problem is parameterized by a cyclic group $\mathbb G$ of order $n$ with a generator $g$, where $n$ is typically a large prime. The challenge is to recover a random $u\in\integers_n$ from $g^u$. Many cryptographic applications
rely on the conjectured intractability of the DL problem in special types of groups, including subgroups of $\integers^*_p$ and certain families of elliptic curves.

The DDL problem is a distributed variant of the DL problem that was recently introduced in the context of group-based homomorphic secret sharing~\cite{BGI16}. In DDL there are two parties, where the first party's input is $g^u$ for a random $u$, and the second party's input is $g^{u+s}$ where $s\in\{0,1\}$ (more generally, $s$ can be a small integer). The goal is for each party to locally output an integer, such that the difference between the two outputs is $s$. One can assume without loss of generality that the two parties run the same algorithm.

Note that a DL algorithm can be used to perfectly solve the DDL problem. However, this is computationally infeasible in a cryptographically hard group, where $n$ is enormous. Instead, a DDL algorithm uses a bounded running time (typically polylogarithmic in the group order $n$) to obtain the correct difference except with error probability $\delta$. For instance,
the initial solution proposed in~\cite{BGI16} uses a pseudorandom function to mark each group element as ``distinguished'' with probability $\delta$, and makes each party, on input $v$, output the smallest $z\ge 0$ such that $v\cdot g^z$ is distinguished. The (expected) running time of this algorithm is roughly $1/\delta$, and the error probability is $\delta$ (corresponding to the case where $s=1$ and $g^u$ is distinguished).

The DDL problem can be related to the LPHS problem (over a non-binary alphabet) by associating each party's DDL input $v$ with an LPHS input consisting of the sequence of group elements $x=(v,  gv, g^2v,\ldots, g^{n-1}v)$. Indeed, multiplication of the DDL input by $g$ corresponds to a cyclic shift of $x$ by one symbol to the left. We formalize this intuition by proving a general two-way relation between LPHS and DDL algorithms in the {\em generic group model}~\cite{S97}, where group elements are assigned random labels and the algorithm is only given oracle access to the group operation.\footnote{Specifically, in Section~\ref{sec:LPHS-DDL} we prove that any LPHS gives a DDL algorithm (in the generic group model) with similar parameters, while any DDL algorithm gives an LPHS with a negligible cost in error probability assuming $d = O(n^{1/4})$ and $n$ is prime.}

The applications we derive from the above connection give an unusual example for the usefulness of results on group-based cryptography in a post-quantum world. Indeed, all traditional applications of group-based cryptography are subject to quantum polynomial-time attacks using Shor's algorithm, and are thus useless in a post-quantum world. If scalable quantum computers become a reality, cryptosystems that are ``quantum broken'' will become obsolete. In contrast, sublinear-time classical algorithms will still be meaningful even in a post-quantum world.

\paragraph{\bf LPHS constructions.} The simple DDL algorithm from~\cite{BGI16} corresponds to a $(d,\delta)$-LPHS where $\delta=\tilde O(1/d)$. Another simple LPHS construction with similar parameters, implicit in a DDL algorithm from~\cite{BGI17}, makes a simple use of MinHash~\cite{BroderCFM00}: let $h(x)$ output the index $i$, $1\le i \le d$, that minimizes the value of a MinHash applied to a polylogarithmic-length substring of $x$ starting from $x_i$.

It is tempting to conjecture that the above simple LPHS constructions are near-optimal, in the sense that $\delta=o(1/d)$ is impossible. It turns out, however, that a quadratic improvement can be obtained from a recent optimal DDL algorithm due to Dinur et al.~\cite{DKK18}. Their \emph{Iterated Random Walk (IRW)} algorithm, whose self-contained description appears in
Appendix~\ref{sec:dkk}
is based on a carefully chosen sequence of random walks in the group. It can be viewed as a non-trivial extension of Pollard's classical ``kangaroo'' DL algorithm~\cite{Pollard78}, which runs in time $\tilde O(\sqrt{n})$ and has low space complexity. Applying the LPHS vs.\ DDL connection to the positive and negative results on DDL from~\cite{DKK18}, we get the following theorem.

\begin{theorem} [Near-optimal LPHS]\label{th-main}
There exist $(d,\delta)$-LPHS with: (1) $\delta= \tilde{O}(d^{-2})$ for $d\le\sqrt{n}$,
and (2) $\delta = n^{-\omega(1)}$ for $d = \tilde{O}(\sqrt{n})$. Furthermore, both ``$\delta= \tilde{O}(d^{-2})$'' in (1) and ``$d = \tilde{O}(n^{1/2})$'' in (2) are optimal up to polylogarithmic factors.
\end{theorem}

Interestingly, any sublinear-time LPHS must inherently make adaptive queries to its input.
Adaptive queries are unusual in the context of sublinear metric algorithms, but were previously used in sublinear algorithms for approximating edit distance~\cite{Saha14,ChakrabortyGK16,GoldenbergKS19}. A random walk technique was recently
used in~\cite{KociumakaS20} to obtain a sublinear-time embedding of edit distance to Hamming distance.

\paragraph{\bf Additional variants.} We prove similar bounds for the {\em worst-case} and {\em non-cyclic variants} of LPHS, which are motivated by the applications we discuss in Section~\ref{sec-introapps}. The result for the worst-case variant is obtained via a general reduction, and inevitably excludes a small set of inputs that are close to being periodic. The result for the non-cyclic case does not follow generically from the cyclic case (except when $d<n^{1/3}$), and requires a special analysis of the IRW algorithm~\cite{DKK18}. Also, in this case only  (1) holds, since the error probability of a non-cyclic LPHS must satisfy $\delta=\Omega(1/n)$ regardless of $d$. In fact, whereas $d = \sqrt{n}$ is the hardest case for Theorem~\ref{th-main} in the non-cyclic case (in that (1) for $d = \sqrt{n}$ easily implies (1) for smaller $d$), it is the easiest for the cyclic case (in that it is implied by the simpler algorithm of Pollard~\cite{Pollard78}).

In the context of {\em sketching} for shifts, the above results imply solutions that simultaneously achieve near-optimal sketch size of at most polylog$(n)$, near-optimal running time of $\tilde O(\sqrt{n})$, and negligible error probability, for both cyclic and non-cyclic shifts, and for arbitrary ``far-from-periodic'' inputs.

\subsubsection{Multidimensional LPHS}

Viewing LPHS as a location identifier, it is natural to consider a generalization to two dimensions and beyond. Indeed, a 2-dimensional (non-cyclic) LPHS can be useful for aligning or sequencing local views of a big 2-dimensional (digital or physical) object. A $k$-dimensional LPHS maps a $k$-dimensional matrix (with entries indexed by $\integers_n^k$) into a vector in $\integers_n^k$ so that (cyclically) shifting the input matrix by 1 in axis $i$ changes the output vector by the unit vector $e_i$, except with $\delta$ error probability. As before, this guarantees recovering an arbitrary shift vector with error probability that scales with the $\ell_1$ norm of the shift.

\paragraph{Upper bounds.} In the 2-dimensional case, the algorithm can be viewed as allowing two non-communicating parties, who are given points $(x,y)$ and $(x+\alpha,y+\beta)$ in the same random for unknown $\alpha,\beta \in \{0,1\}$, to maximize the probability of synchronizing at the same point, where only $d$ queries are allowed. A straightforward approach is to use a MinHash algorithm in which the parties take the minimal  hash value computed on values of a $d^{1/2} \times d^{1/2}$ matrix of elements beginning at the location of each party, resulting in a $(d,\delta)$-LPHS with $\delta = \tilde{O}(d^{-1/2})$. A better error bound of $\delta = \tilde{O}(d^{-2/3})$ can be obtained by combining the application of MinHash on one axis with the application of the aforementioned optimal IRW algorithm on the other axis.

We present three improved algorithms in Section~\ref{sec:2dim}. The simplest of those, with a bound of $\delta = \tilde{O}(d^{-4/5})$, is obtained by applying IRW on both axes.
A natural idea is to first synchronize on the column; then, synchronizing on the row is easy, using the 1-dimensional IRW algorithm. To synchronize on the column, the parties perform the 1-dimensional IRW algorithm with $d/d'$ `horizontal' steps, where the information used to determine each step is distilled from the column in which the current point resides by using an IRW algorithm with $d'$ `vertical' steps. The analysis in Section~\ref{sec:2dim} (Lemma~\ref{lem:2D45}) shows that the bound $\delta = \tilde{O}(d^{-4/5})$ is obtained for the parameter $d'=d^{3/5}$.

Our main upper bound is obtained by a more complex algorithm, which -- perhaps, surprisingly -- does not rely on the optimal 1-dimensional IRW algorithm at all. We prove:
\begin{theorem} \label{thm:2dupper}
For $n = \tilde{\Omega}(d)$,
there is a 2-dimensional
$(d,\delta)$-LPHS
with $\delta = \tilde{O}(d^{-7/8})$. There is also a non-cyclic 2-dimensional LPHS with the same parameters.
\end{theorem}
The algorithm we use to prove Theorem~\ref{thm:2dupper} consists of three stages. After each stage the parties either converge to the same location, in which case they stay synchronized to the end of the algorithm, or the two walks are within a bounded distance from each other. Stage 1 begins with a distance of at most $1$ on each axis, and is a straightforward application of the 2-dimensional MinHash-based algorithm. Stage 2 begins with a distance of at most $\sqrt{d}$ on each axis and uses an asymmetric deterministic walk that consists of $\sqrt{d}$ horizontal steps of size $\sim d^{1/4}$, where each step is pseudo-randomly determined by information distilled from a vertical walk of length $\sqrt{d}$ and step sizes $\sim d^{1/4}$. Stage 3 begins with a distance of at most $d^{3/4}$ on each axis and uses a different deterministic walk. This time, the horizontal steps are of size $1$, while the vertical steps are of size about $d^{3/8}$, and unlike all other steps, can be negative. The analysis of the algorithm relies on martingale techniques.

Finally, we present another 2-dimensional LPHS algorithm, which seems harder to analyze, but for which we conjecture that the error rate is at most $\tilde{O}(d^{-1})$. This bound is essentially the best one can hope for given the lower bound discussed below. The idea behind this algorithm is to not treat the axes separately but rather to perform a series of deterministic walks over $\mathbb{Z}^2$, with step sizes of about $d^{1/4}, d^{3/8}, d^{7/16},\ldots,\sqrt{d}/2$. Our experiments suggest that the error rate of this algorithm is indeed $\tilde{O}(d^{-1})$. However, the analysis (and especially deterministic resolution of cycles in the random walk) is quite involved, and settling our conjecture is left open for future work. Nevertheless, the heuristic algorithm can be used in cryptographic applications (such as packed homomorphic secret sharing which is described next) without compromising their security. Moreover, the worst-case scenario in which the error is larger than predicted by our experiments can be easily detected by applications.

\paragraph{Lower bound.}
We complement our positive results by proving the following lower bound, extending in a nontrivial way the lower-bound for 1-dimensional LPHS obtained from~\cite{DKK18} via the DDL connection.
\begin{theorem}\label{thm:kdlower}
For
$n=\Omega(d^{2/k})$,
any $k$-dimensional $(d,\delta)$-LPHS satisfies $\delta = \Omega(d^{-2/k})$.
\end{theorem}
The intuition is related to the birthday bound. A $k$-dimensional box with edge length $d^{2/k}$ contains $d^2$ points.
If the $k$-dimensional shift is uniform within this box, we expect the $d$ queries of the two parties not to
intersect with constant probability, implying that $\delta = \Omega(1)$.
Given this, the proof for
smaller shifts follows by a union bound.
While the intuition is simple, is it not clear how to directly apply the birthday bound, and the formal proof
is based on an argument involving Minkowski sums and differences of sets in $\mathbb{R}^k$.


\subsubsection{Applications}
\label{sec-introapps}
We present several cryptographic and algorithmic applications that motivate different variants of LPHS, exploiting both the functionality and the sublinearity feature. All of the applications can benefit from our 2-dimensional LPHS constructions, and most require the non-cyclic, worst-case variant. See 
Section~\ref{sec:applications}
for a taxonomy of the LPHS variants required by different applications.

\paragraph{Packed homomorphic secret sharing.}
We demonstrate a cryptographic application of \emph{$k$-dimensional} LPHS in trading computation for communication in group-based homomorphic secret sharing (HSS). In a nutshell, we use LPHS to further improve the succinctness of the most succinct approach for simple ``homomorphic'' computations on encrypted data, by packing 2 or more plaintexts into a single ciphertext. Compared to competing approaches (see, e.g.,~\cite{AkaviaSWY19,Paillier1} for recent examples), group-based packed HSS can provide much better succinctness and client efficiency. The key technical idea is to use a non-cyclic $k$-dimensional LPHS for implementing a $k$-dimensional variant of DDL, where $k$ independent group generators are used for encoding $k$ small integers by a single group element, and where multiplication by each generator is viewed as a (non-cyclic) shift in the corresponding direction. This $k$-dimensional generalization of DDL can be potentially useful for other  recent applications of DDL that are unrelated to HSS~\cite{DottlingGIMMO19,trapdoor2,BrakerskiKM20}.  See 
Section~\ref{sec-applicationhss}
for a detailed discussion of this cryptographic application of LPHS along with the relevant background.

\paragraph{Location-sensitive encryption.} We apply LPHS to obtain a sublinear-time solution for {\em location-sensitive encryption} (LSE), allowing one to generate a public ciphertext that can only be decrypted by someone in their (physical or virtual) neighborhood. Here proximity is defined as having significantly overlapping views, and security should be guaranteed as long as a non-overlapping view is sufficiently unpredictable.  The above goal can be reduced to realizing a sublinear-time computable fuzzy extractor~\cite{DORS}\footnote{There are two differences from the  standard notion of fuzzy extractors: the ``distance'' is not a strict metric, and the notion of unpredictability needs to ensure that the source is far from periodic with high probability. }  for shift distance.
Obtaining such fuzzy extractors from LPHS constructions requires an understanding of their behavior on entropic sources. It turns out that even for high-entropy sources, an LPHS provides no unpredictability guarantees. We get around this problem by defining a hash function that combines the output of an LPHS with a local function of the source. Using this approach, we obtain a sublinear-time LSE whose security holds for a broad class of mildly unpredictable sources.
See Section~\ref{sec-lse}
for the LSE application of LPHS.

\paragraph{Algorithmic applications.} As discussed above, algorithmic applications of LPHS follow from the vast literature on sketching, locality-preserving and locality-sensitive hashing, and metric embeddings. Indeed, our different LPHS flavors can be roughly viewed as probabilistic isometric embeddings of certain shift metrics into a Euclidean space. Thus, for example, an LSH for the same shift metric can potentially follow by concatenating the LPHS with an LSH from the literature. However, some care should be taken in applying this high-level approach. One issue is the average-case nature of LPHS, which makes the failure probability input-dependent. We get around this via a worst-case to average-case reduction that restricts the input space to ``non-pathological'' inputs that are far from periodic. Another issue is that LPHS provides no explicit guarantees for inputs that are too far apart. We get around this by using the fact that an LPHS must have a well-spread output distribution on a random input. As representative examples, we demonstrate how LPHS can be applied in the contexts of communication complexity and LSH-based near-neighbor data structures for shifts.
The algorithmic applications of LPHS are discussed in 
Section~\ref{sec-algorithmic}.

\paragraph{\bf Open questions.}
Our work leaves several open questions.
The main question, on which we make partial progress, is obtaining optimal parameters for  $k$-dimensional LPHS.
Other questions concern the optimality of the LPHS-based approach to sketching.
A negative result from~\cite{DKK18}, which can be used to rule out sublinear-time LPHS with non-adaptive queries, in fact holds even for sketching.
Do LPHS-based sketches also provide an optimal tradeoff between sketch size and error probability?



\paragraph{\bf Organization.}
In Section~\ref{sec:prelims} we introduce necessary preliminaries and notation, including the definition of Locality-Preserving Hash functions for Shifts (LPHS), and simple properties of and relations between LPHS variants. 
In Section~\ref{sec:LPHS-DDL}, we present the general two-way connection between LPHS and algorithms for distributed discrete logarithm in the generic group model. 
In Section~\ref{sec:2dim} we provide our results on multidimensional LPHS. 
Section~\ref{sec:applications} contains applications of LPHS. And, for completeness, in Appendix~\ref{app:dkk}, we provide LPHS results based on the Iterative Random Walk algorithms of Dinur, Keller, and Klein~\cite{DKK18}.

%% file: Preliminaries.tex
\newcommand{\Z}{\mathbb{Z}}

\section{Preliminaries}
\label{sec:prelims}

We denote by $\mathbb{Z}_n$ the additive group of integers modulo $n$.
We will typically consider strings of length $n$ over an alphabet $\Sigma_b = \zo^{b}$, indexing string entries by $i\in\Z_n$.
We will use the notation $x^{(b)}$ when we want to make the alphabet size explicit. When the alphabet is binary or when $b$ is clear from the context, we will typically omit the superscript and use the notation $x$.
For $x^{(b)}\in\Sigma_b^n$ we denote by $x^{(b)}[i]$ the $i$'th symbol of $x^{(b)}$, for $i\in\Z_n$.

We denote by $x^{(b)} \ll r$ the cyclic rotation of $x^{(b)}$ by $r$ symbols to the left, namely the string $y^{(b)}$ defined by $y^{(b)}[i]=x^{(b)}[i+r]$ with addition modulo $n$. We will also consider a {\em non-cyclic shift}, denoted by $x^{(b)} \lll r$, where the $r$ leftmost symbols of $x$ are chopped and $r$ {\em random} symbols are added on the right. Note that unlike the cyclic shift operator, which is deterministic, the non-cyclic version is randomized.
We use $\Delta(x,y)$ to denote the Hamming distance between $x$ and $y$, namely the number of symbols $i$ in which $x[i]$ and $y[i]$ differ.

We use the notation $x^{(2,b)}$ to denote a 2-dimensional string (i.e., matrix) over alphabet $\Sigma_b$ and denote by $x^{(2,b)}[i,j]$ its $(i,j)$ entry. We denote by $x^{(2,b)} \ll (r_1,r_2)$ the cyclic rotation of $y^{(2,b)}$ by $r_1$ symbols to the left on the first axis and $r_2$ symbols to the left on the second axis. That is, $y=x^{(2,b)} \ll (r_1,r_2)$ is defined by $y[i_1,i_2]=x^{(2,b)}[i_1+r_1, i_2+r_2]$, where addition is modulo $n$. We will also consider the natural $k$-dimensional generalization $x^{(k,b)} \ll (r_1,r_2,\ldots,r_k)$ and its non-cyclic variant $x^{(k,b)} \lll (r_1,r_2,\ldots,r_k)$.


\subsection{Locality-Preserving Hash Functions for Shifts}


We now define our main notion of LPHS and some of its useful variants.

\begin{definition}[LPHS: main variants] \label{def-lphs}
Let $\func{h}{\Sigma_b^n}{\mathbb{Z}_n}$ be a function. We say that $h$ is a (cyclic) {\em $(d,\delta)$-LPHS} if $h$ can be computed by making $d$ {\em adaptive} queries (of the form $x[i]$) to an input $x\in\Sigma_b^n$ and moreover
$\pr_{x\in_R\Sigma_b^n} \lbs h(x) \neq h(x \ll 1) + 1 \rbs \leq \delta$.
\medskip

We will consider the following modifiers (that can be combined in a natural way):
\begin{itemize}
\item {\bf Non-cyclic} LPHS: replace $\mathbb{Z}_n$ by $\mathbb{Z}$ and $x \ll 1$ by
$x \lll 1$;
\item {\bf $k$-dimensional} LPHS: let $x$ be a random $k$-dimensional string  $x\in\Sigma_b^{\Z^k_n}$ and $\func{h}{\Sigma_b^{\Z^k_n}}{\mathbb{Z}_n^k}$. We require that $\pr_{x} \lbs h(x) \neq h(x \ll e_i) + e_i \rbs \leq \delta$ for every {\em unit vector} $e_i\in\Z_n^k$.
\end{itemize}
We will sometimes make more parameters explicit in the notation. For instance, an $(n,b,d,\delta)$-LPHS is
a
$(d,\delta)$-LPHS $\func{h}{\Sigma_b^n}{\mathbb{Z}_n}$.
\end{definition}

\begin{remark}[On computational complexity]
A $(d,\delta)$-LPHS $\func{h}{\Sigma_b^n}{\mathbb{Z}_n}$ can be viewed as a {\em depth-$d$ decision tree} over $n$ input variables taking values from the alphabet $\Sigma_b$. In all of our positive results, $h$ is {\em semi-explicit} in the sense that it can be realized by a {\em randomized} polynomial-time algorithm having oracle access to the input $x$. (In fact, our algorithms can be implemented in probabilistic  $\tilde O(d)$ time.) Here the same randomness for $h$ is used in the two invocations $h(x)$ and $h(x \ll 1)$. Alternatively, our positive results imply a deterministic $h$ in a non-uniform setting. Our negative results apply to the existence of $h$ with the given parameters, irrespective of the computational complexity of generating it.
\end{remark}

\begin{remark}[Worst-case vs.\ average-case LPHS] Our default notions of LPHS assume a uniformly random input $x$. While this suffices for some  applications, a worst-case notion of LPHS is more desirable for most applications. Since shift detection is impossible for highly periodic inputs (such as the all-0 string), or even for approximately periodic in the context of sublinear-time algorithms, the notion of worst-case LPHS is restricted to a set of ``typical'' inputs that are far from being periodic. Our notion of ``typical'' is very broad and arguably captures essentially all naturally occurring inputs in our motivating applications. In
Section~\ref{sec-worstcase} we present a simple reduction of this worst-case flavor of LPHS to our default notion of LPHS for random inputs. This applies both to the cyclic and non-cyclic variants. The reduction only incurs a polylogarithmic loss in the parameters. Note that, unlike our main notion of LPHS, here it is inherent that the function $h$ be randomized. A useful related byproduct of the worst-case variant is that the failure events of two independently chosen $h_1$ and $h_2$ are independent. This is useful for algorithmic applications of LPHS.
\end{remark}

For some applications, we will be interested in the following additional LPHS variants.
\begin{definition}[LPHS: additional variants] \label{def-morelphs}
We consider the following additional variants of the main notion of LPHS from Definition~\ref{def-lphs}.
\begin{itemize}
\item {\bf Shift-bounded} LPHS with shift bound $R$: requires that for every $1\le r\le R$, we have
\[\pr_{x\in_R\Sigma_b^n} \lbs h(x) \neq h(x \ll r) + r \rbs \leq \delta,\]
and similarly for the non-cyclic case.
\item {\bf Las Vegas} LPHS: allow $h$ to output $\bot$ with probability $\le \delta$, and require that $h$ never fail in the event that neither of its two invocations outputs $\bot$.
\end{itemize}
\end{definition}

A generic way of obtaining a Las Vegas LPHS $h'$ from an LPHS $h$ is to invoke $h$ on both $x$ and $x'=x\ll 1$ and output $\bot$ if $h(x)\neq h(x')+1$. However, an extension of this to a {\em shift-bounded} LPHS is inefficient, since it requires invoking $h$ on $x\ll r$ for every $0\le r\le R$. In Appendix~\ref{sec:LasVegas} (Lemma~\ref{lem-lv})
we show that an optimal shift-bounded LPHS admits a Las Vegas variant with better parameters.

%% file: Reductions.tex

\subsection{Simple LPHS Properties and Reductions}
\label{sec:relations}

In this section we prove simple properties of LPHS and reductions between different LPHS variants that will be useful in what follows.

The following lemma is immediate.
\begin{lemma}[Conversion between non-cyclic and cyclic LPHS] \label{lem:convert}
Any $k$-dimensional non-cyclic $(n,b,d,\delta)$-LPHS that queries the last symbol $(n-1) \cdot e_i\in\Z_n^k$ in every dimension with probability 0 gives a cyclic $(n,b,d,\,\delta)$-LPHS and vice versa.
\end{lemma}

The following lemma shows the usefulness of LPHS for detecting arbitrary shift amounts.
\begin{lemma}[Bigger shifts]\label{lem:shift}
For any positive integer $r$, an $(n,b,d,\delta)$-LPHS $h$ satisfies
$$\pr_{x^{(b)}} \lbs h(x^{(b)}) \neq h(x^{(b)} \ll r) + r \rbs \leq r \cdot \delta.$$ The same holds for non-cyclic LPHS, where $x^{(b)} \ll r$ is replaced by the string $x^{(b)} \lll r$ obtained from $x^{(b)}$ by chopping the $r$ left-most symbols and adding $r$ uniformly random symbols to the right. In the $k$-dimensional case, we have
$$\pr_{x^{(k,b)}} \lbs h(x^{(k,b)}) \neq h(x^{(k,b)} \ll (r_1,\ldots,r_k)) + (r_1,\ldots,r_k) \rbs \leq \left(\sum_{i=1}^k r_i\right) \cdot \delta.$$
\end{lemma}
\begin{proof} We give the proof for the cyclic 1-dimensional case for  simplicity. The proof for the non-cyclic and $k$-dimensional cases is similar.
\begin{align*}
\pr_{x^{(b)}} \lbs h(x^{(b)}) \neq h(x^{(b)} \ll r) + r \rbs \leq \\
\pr_{x^{(b)}}
\lbs \exists i \in [0,r-1] : h(x^{(b)} \ll i) \neq h(x^{(b)} \ll i+1) + 1 \rbs \leq \\
\sum_{i=0}^{r-1} \pr_{x^{(b)}}
 \lbs h(x^{(b)}) \neq h(x^{(b)} \ll 1) + 1 \rbs \leq r \cdot \delta.
\end{align*}
\end{proof}
\medskip

As a simple corollary, we get a lower bound on the error probability of non-cyclic LPHS.
\begin{claim}[Error bound for non-cyclic LPHS]\label{cl:nc}
For any non-cyclic $(n,b,d,\delta)$-LPHS we have $\delta\ge 1/(2n)$.
\end{claim}
\begin{proof} Suppose towards contradiction that $\delta<1/(2n)$. Applying Lemma~\ref{lem:shift} with $r=n$ we get that for two random and independent strings $x,x'\in_R\Sigma_b^n$ we have $\Pr_{x,x'}[h(x)=h(x')+n]>1/2$, where addition is taken over the integers. By symmetry, we also have $\Pr_{x,x'}[h(x')=h(x)+n]>1/2$. Since the two events cannot co-occur, we get a contradiction.
\end{proof}

The above lower bound does {\em not} apply to cyclic LPHS. Indeed, when $d>n^{1/2}$ one can use discrete logarithm techniques to get cyclic LPHS in which $\delta$ is negligible in $n$. The proof of Claim~\ref{cl:nc} extends naturally to the $k$-dimensional case, where we have $\delta\ge 1/(2kn)$.
\medskip

The following lemma shows a simple way to trade alphabet size for input length.
\begin{lemma}[Bigger alphabet]\label{lem:bigger-alphabet}
For any positive integer $r$ such that $n/r$ is an integer, if there exists an $(n,b,d,\delta)$-LPHS, then there exists an $(n/r,b \cdot r,d,r \cdot \delta)$-LPHS.
\end{lemma}
\begin{proof}
On input $x^{(b \cdot r)}$, the $(n,b \cdot r,d,\delta \cdot k)$-LPHS (denoted by $h_1$) runs the $(n,b,d,\delta)$-LPHS (denoted by $h_2$) on $x^{(b)}$, and outputs $\lfloor h_2(x^{b})/r \rfloor$. Clearly, each query of $h_2$ to $x^{(b)}$ can be answered by a single query of $h_2$ to $x^{(b \cdot r)}$. Based on Lemma~\ref{lem:shift},
\begin{align*}
\pr_{x^{(b \cdot r)}} \lbs h_1(x^{(b \cdot r)}) \neq h_1(x^{(b \cdot r)} \ll 1) + 1 \rbs \leq
\pr_{x^{(b)}} \lbs h_2(x^{(b)}) \neq h_2(x^{(b)} \ll r) + r \rbs \leq r \cdot \delta.
\end{align*}
\end{proof}

The following lemma can be used to decrease the alphabet size of LPHS. In particular, it is useful for obtaining near-optimal LPHS over a {\em binary} alphabet from LPHS over a bigger alphabet $\Sigma_b$, such as the one derived directly from~\cite{DKK18}.
\begin{lemma}[Smaller alphabet]\label{lem:smaller-alphabet}
For any positive integers $r,k$, if there exists an $(n,b \cdot r,d,\delta)$-LPHS, then there exists an $(n,b,d \cdot k,\delta + 4 \cdot d^2/2^{b \cdot k})$-LPHS.
\end{lemma}
\begin{proof}
Given a function $f: \zo^{b \cdot k} \rightarrow \zo^{b \cdot r}$ selected uniformly at random,
on input $x^{(b)}$, the $(n,b,d \cdot k,\delta + 4 \cdot d^2/2^{b \cdot k})$-LPHS (denoted by $h_1$) runs the $(n,b \cdot r,d,\delta)$-LPHS (denoted by $h_2$) and answers query $i$ by querying a $k$-tuple of elements and applying $f(x^{(b)}[i], \ldots , x^{(b)}[i+k-1])$. Finally, it outputs the same value as $h_2$.

We define the bad event $\mathcal{E}$, where two $k$-tuple queries of $h_1(x^{(b)})$ or $h_1(x^{(b)} \ll 1)$ for $i \neq j$ satisfy $(x^{(b)}[i], \ldots , x^{(b)}[i+k-1]) = (x^{(b)}[j], \ldots , x^{(b)}[j+k-1])$.
Conditioned on $\neg \mathcal{E}$,
\begin{align*}
\pr_{x^{(b)}} \lbs h_1(x^{(b)}) \neq h_1(x^{(b)} \ll 1) + 1 \rbs =
\pr_{x^{(b \cdot r)}} \lbs h_2(x^{(b \cdot r)}) \neq h_2(x^{(b \cdot r)} \ll 1) + 1 \rbs \leq \delta.
\end{align*}
For any $i \neq j$, the probability that $(x^{(b)}[i], \ldots , x^{(b)}[i+k-1]) = (x^{(b)}[j], \ldots , x^{(b)}[j+k-1])$ is $2^{-b \cdot k}$, and taking a union bound on all $\binom{2d}{2}$ query pairs of $h_1(x^{(b)})$ and $h_1(x^{(b)} \ll 1)$ gives $\pr[\mathcal{E}] < 4 \cdot d^2/2^{b \cdot k}$.

Note that this requires both executions of $h_1(x^{(b)})$ and $h_1(x^{(b)} \ll 1)$ to have shared randomness. In a non-uniform setting, we can use an averaging argument to fix this randomness and obtain a deterministic $h$. To make $h$ semi-explicit, pick $f$ from an efficient family of $n$-wise independent hash functions.
\end{proof}


Next, we observe that any cyclic LPHS implies a non-cyclic LPHS with a higher error probability $\delta$. Concretely, the error probability grows by $d/n$. For $d=n^{1/3}$, this can yield a non-cyclic LPHS with (near-optimal) $\delta=\tilde O(n^{-2/3})$. But for $n^{1/3}<d<n^{1/2}$ this generic construction is subsumed by the direct construction from
Section~\ref{sec:iterated}
that can achieve the near-optimal parameters $d=n^{1/2}$ and $\delta=\tilde O(1/n)$.
\begin{lemma}[From cyclic to non-cyclic]
\label{lem:non-cyclic}
If there exists a {\em cyclic} $(n,b,d,\delta)$-LPHS, then there exists a {\em non-cyclic} $(n,b,d,\delta+d/n)$-LPHS.
\end{lemma}
\begin{proof}
Let $h$ be a cyclic $(n,b,d,\delta)$-LPHS.  We define a randomized non-cyclic LPHS $h'_\rho(x)=h(x\ll \rho)$ where $\rho\in_R\Z_n$ is a uniformly random cyclic shift. Since the difference between $x\ll 1$ and the corresponding string $x'$ for non-cyclic shift is restricted to the $(n-1)$-th position, after cyclically shifting by $\rho$ the difference will be restricted to a single random position $\rho^*=(n-1-\rho)\mod n$. Since $h$ (adaptively) queries $d$ symbols of its input, the probability of querying the differing position $\rho^*$ is bounded by $d/n$, from which the lemma follows.
\end{proof}

Finally, several of our results will rely on the following ``smoothness'' feature of LPHS.

\begin{lemma}\label{lem-smoothness}
Let $m > 1$ be a natural number. Then any $(n,b,d,\delta)$-LPHS satisfies
\[
\Pr_x[h(x) \bmod m = a] \leq m^{-1} + \delta + O(1/n)
\]
for every $a \in \mathbb{Z}_m$.
\end{lemma}
\begin{proof}
We assume that $b = 1$. The proof for $b> 1$ is similar.

Consider the set $S$ of inputs $x$ for which the minimal $r$ such that $x = x \ll r$ is $n$. The set $S$ contains a fraction of $1 - o(1/n)$ of the inputs $x \in \{0,1\}^n$. Suppose first that $m$ divides $n$. Partition the set $S$ into sequences of length $m$ of the form $x,x \ll 1,\ldots,x \ll (m-1)$ and consider the executions $h(x) \bmod m,\ldots,h(x \ll (m-1)) \bmod m$. Note that if $a$ appears $t > 1$ times in an execution sequence, then $h$ must err on at least $t-1$ inputs in this sequence. Consequently,
\[
\Pr_x[h(x) \bmod m = a] \leq m^{-1} + \delta + o(1/n).
\]
If $m$ does not divide $n$, then the set $S$ cannot be accurately partitioned as above.
Consequently, each set of cyclic shifts of any input $x$ (of size $n$) in $S$ may contain at most one additional input $y$ such that $h(y) \bmod m = a$, contributing to the additional $O(1/n)$ factor.
\end{proof}


%% file: LPHS_DDL.tex
\section{LPHS and Distributed Discrete Log}
\label{sec:LPHS-DDL}

In this section we introduce the Generic Group Model (GGM) and Distributed Discrete Logarithm (DDL) problem, and present a general two-way relation between (cyclic) LPHS and generic-group algorithms for DDL.
We discuss and define the GGM and DDL problem in Section~\ref{sec:DDL}, and we provide the correspondence with LPHS in Section~\ref{sec-ddl}.

\subsection{Generic Group Model and Distributed Discrete Log}
\label{sec:DDL}

Our notion of LPHS is closely related to variants of the discrete logarithm problem: given a group generator $g\in \mathrm{G}$ and a group element $g^v$, find $v$. More concretely, we will be interested in algorithms for problems related to discrete logarithm in the so-called {\em generic group model} (GGM). The GGM assigns random labels to group elements and treats the group operation as an oracle.  We formalize this below.

Let $n$ be a positive integer parameter (corresponding to the group order) and $b \geq 3\log n$ an integer (representation length of group elements). The GGM setting can be described as a game, where at the beginning, a string $x^{(b)} \in \Sigma_b^n$ is chosen uniformly at random.\footnote{We require $b \geq 3\log n$ to ensure that for each $i \neq j$, $x[i] \neq x[j]$ (except with
$\le 1/n$
probability).} In the discrete log problem for $\mathbb{Z}_n$, a value $v \in \mathbb{Z}_n$ is chosen uniformly at random. A generic algorithm $A$ for the discrete log problem in $\mathbb{Z}_n$ is a probabilistic algorithm that issues $d$ (adaptive) queries of the form $(i,j) \in \mathbb{Z}_n \times \mathbb{Z}_n$. The answer to query $(i,j)$ is $x^{(b)}[\ell_v(i,j)]$, where $\ell: \mathbb{Z}_n \times \mathbb{Z}_n \rightarrow \mathbb{Z}_n$ is the affine query evaluation function defined by $\ell_v(i,j) = i \cdot v + j$ (where arithmetic operations are performed modulo $n$). Using the group notation, the query $(i,j)$ corresponds to group element $(g^v)^i\cdot g^j$.

The algorithm $A$ succeeds to solve the discrete log problem if $A^{\mathrm{G}}(x^{(b)},v) = v$,\footnote{We use the notation $A^{\mathrm{G}}(x^{(b)},v)$ to indicate that $A$ is a generic algorithm with no direct access to the parameters $x^{(b)},v$.} and its success probability is taken over the uniform choices of $x^{(b)}$ and $v$ (and possibly additional randomness of its own coin-tosses).

The flavor of GGM we use in this paper is similar
to the one of Shoup~\cite{S97}. Besides differences in notations, there are two additional technical differences which are generally minor. First, in~\cite{S97}, strings are uniformly assigned to elements of $\mathbb{Z}_n$ without replacement, whereas in our model, we assign strings with replacement. However, a collision $x[i] = x[j]$ for some pair $(i,j)$ is possible with probability $\le 1/n$, which is negligible in our context.
Second, in~\cite{S97} queries of $A$ are limited to linear combinations with coefficients of $\pm 1$ to previously queried elements (where the initial queried elements consist of $g$ and $g^v$). We note that any query $(i,j)$ can be issued in Shoup's original GGM after at most $O(\log n)$ queries (using the double-and-add algorithm). Therefore, although our model is slightly stronger, any algorithm in our model can be simulated by an algorithm in the model of~\cite{S97} by increasing the query complexity by a multiplicative factor of $O(\log n)$.

The following success probability upper bound was proved in~\cite{S97}.

\begin{theorem}[\cite{S97}, Theorem 1 (adapted)] \label{thm:DL}
For a generic discrete log algorithm $A$ with $d$ queries and prime $n$, we have
$\pr_{x^{(b)},v}[A^{\mathrm{G}}(x^{(b)},v) = v] = O(d^2/n)$.
\end{theorem}
Although our model is slightly different than the one of~\cite{S97}, this result holds in our model as well (by a straightforward adaptation of the proof of~\cite{S97}). The assumption that $n$ is prime ensures that $\mathbb{Z}_n$ does not contain any non-trivial subgroup. It is necessary in general, since for composite $n$, the Pohlig-Hellman algorithm~\cite{PohligH78} breaks the discrete log problem into smaller problems in subgroups of $\mathbb{Z}_n$, beating the bound of Theorem~\ref{thm:DL}.

We now define a restricted class of GGM algorithms that better correspond to LPHS.
\begin{definition}
A GGM algorithm $A$ is called query-restricted if it only issues queries of the form $(i,j) \in \mathbb{Z}_n \times \mathbb{Z}_n$ with $i = 1$.
\end{definition}
Thus, $A$ is restricted to query group elements with a known shift $j$ from $v$, analogously to the way an LPHS algorithm queries elements at a known offset.
Query-restricted algorithms cannot exploit the subgroup structure of composite groups, and thus Theorem~\ref{thm:DL} holds for them regardless of whether $n$ is prime. For a similar reason, the factorization of $n$ will not play any role in our results on LPHS.



LPHS is closely related to query-restricted GGM algorithms for a variant of discrete log called {\em distributed discrete log} (DDL)~\cite{BGI16,DKK18} that we describe next.
The syntax is identical to that of discrete log. However, the goal here is different: rather than output $v$ when the (implicit) input is $v$, the goal here is to maintain the {\em difference} between the outputs on $v$ and $v+1$,
except with error probability $\delta$. More formally:
\begin{definition}
A GGM algorithm $A$ is an $(n,b,d,\delta)$-DDL algorithm if it makes $d$  (potentially adaptive)  queries to $x^{(b)}$ and
$\pr_{x^{(b)},v}[A^{\mathrm{G}}(x^{(b)},v) - A^{\mathrm{G}}(x^{(b)},v+1) \neq 1] \leq \delta$.
\end{definition}

\begin{remark}
The original definition of DDL in~\cite{BGI16} involves two parties $A$ and $B$
that may potentially run two different algorithms. The parties are placed within an unknown distance $r \in \{-1,0,1\}$ from each other and their goal is to minimize the error probability defined as
$\pr_{x^{(b)},v}[A^{\mathrm{G}}(x^{(b)},v) - B^{\mathrm{G}}(x^{(b)},v+r) \neq r].$
However, it was shown in~\cite[Lemma 9]{DKK18}, that if both parties use $A$'s algorithm, then the multiplicative loss in error probability is bounded by a constant. Hence, the above restricted definition of DDL is essentially equivalent to the original one of~\cite{BGI16}.
\end{remark}

While LPHS is only directly related to query-restricted GGM algorithms for DDL, Lemma~\ref{lem:qr} asserts that when $n$ is prime and $d$ is sufficiently small compared to $n$, any unrestricted GGM algorithm for DDL can be converted to a query-restricted one at a negligible cost in error probability. This gives rise to Corollary~\ref{corr:conv} that establishes a reduction which converts any unrestricted GGM algorithm for DDL to an LPHS with a negligible cost in error probability.

%


\subsection{Reductions Between LPHS Variants and DDL}
\label{sec-ddl}

In this section we show that a query-restricted DDL algorithm in the GGM is equivalent to our basic notion of (cyclic) LPHS, and then describe consequences of this equivalence.  We then show that for most parameters, any DDL algorithm can be converted into a query-restricted one at a negligible cost.

We use this correspondence to derive asymptotically tight bounds on the parameters of LPHS, using the Iterative Random Walk algorithm from~\cite{DKK18}.

\paragraph{Equivalence of LPHS and query-restricted DDL.}
The equivalence in the query-restricted model is formally captured by the following two-way relation.
\begin{lemma}\label{lem:equiv}
There exist reductions that convert an $(n,b,d,\delta)$-LPHS to a query-restricted $(n,b,d,\delta)$-DDLA and vice versa.

\end{lemma}
\begin{proof}
Given access to an $(n,b,d,\delta)$-LPHS denoted by $h$, we construct a query-restricted $(n,b,d,\delta)$-DDLA, denoted by $A$ as follows. We run $h$ and translate query $j$ into query $(1,j)$ for $A^{\mathrm{G}}(x^{(b)},v)$. We then feed the answer $x^{(b)}[v + j]$ to $h$. Finally, we output the same value as $h$. Since $x^{(b)}[v + j] = (x^{(b)} \ll v)[j]$, we have $A^{\mathrm{G}}(x^{(b)},v) = h(x^{(b)} \ll v)$, where $x^{(b)} \ll v$ is a uniform string. Therefore,
\begin{align*}
\pr_{x^{(b)},v}[A^{\mathrm{G}}(x^{(b)},v) - A^{\mathrm{G}}(x^{(b)},v+1) \neq 1] = \\
\pr_{x^{(b)},v}[h(x^{(b)} \ll v) - h(x^{(b)} \ll v + 1) \neq 1] = \\
\pr_{x^{(b)}}[h(x^{(b)}) - h(x^{(b)} \ll 1) \neq 1] = \delta.
\end{align*}
In a similar way, a query-restricted $(n,b,d,\delta)$-DDLA can be used to construct a $(n,b,d,\delta)$-LPHS.
\end{proof}

The DDL algorithm based on the Iterated Random Walk (IRW) from~\cite{DKK18} 
is query-restricted. Therefore, combining Lemma~\ref{lem:equiv} with the parameters of IRW,
we get the positive result below for cyclic LPHS. The result for non-cyclic LPHS follows from
the fact that the random walk makes queries within an interval of size bounded by $O(d^2)$, hence if $n =\Omega(d^2)$ is large enough, the LPHS gives both cyclic and non-cyclic LPHS with the same parameters.
\begin{theorem}[LPHS Upper Bounds]
\label{thm:1dupper}
For $n = \Omega(d^2)$ and $b \ge 3 \log n$ there is an $(n,b,d,\delta)$-LPHS such that
   $ \delta =  O(1/d^2)$. Moreover, for $n = \Omega(d^2)$ and any $b \ge 1$ there is an $(n,b,d,\delta)$-LPHS with
   $ \delta =  \tilde{O}(1/d^2)$. There are also non-cyclic LPHS with the same parameters.
\end{theorem}

\medskip We can similarly convert the main negative result for DDLA from~\cite[Theorem~5]{DKK18} to a nearly tight lower bound on the error probability of LPHS.
\begin{theorem}[LPHS Lower Bound]
\label{thm:1dlower}
For $n = \Omega(d^2)$, any (cyclic or non-cyclic) $(n,1,d,\delta)$-LPHS satisfies $\delta \geq \Omega(1/d^2)$.
\end{theorem}


\paragraph{From GGM to query-restricted GGM.} 
In this section we show that, when $n$ is prime and $d$ is sufficiently small compared to $n$, any DDL algorithm can be converted into a query-restricted one with similar parameters.

\begin{lemma}\label{lem:qr}
For $b \geq 3 \log b$, there exists a reduction that converts any $(n,b,d,\delta)$-DDLA, for prime $n$, to a query-restricted $(n,b,d,\delta + O(d^2/n))$-DDLA.
\end{lemma}
If $d = O(n^{1/4})$, then since $\delta = \Omega(d^{-2})$ by
Theorem~\ref{thm:1dlower},
we have $\delta + O(d^2/n) = \delta + O(d^{-2}) = O(\delta)$. Hence the reduction is almost without loss.

\begin{proof} (sketch) Given black-box access to an $(n,b,d,\delta)$-DDLA denoted by $A$, we construct a query-restricted $(n,b,d,\delta + O(d^2/n))$-DDLA, denoted by $B$. We will define a query mapping $map : \mathbb{Z}_n \times \mathbb{Z}_n \rightarrow \mathbb{Z}_n$ below and run $A$, while translating query $(i,j)$ into query $(1,map(i,j))$ for $B^{\mathrm{G}}(x^{(b)},v)$ for which the answer $x^{(b)}[v + map(i,j)]$ is fed back into $A$. Finally, we output the same value as $A$.

Our goal in defining $map$ is to simulate the joint distribution of query answers for $A^{\mathrm{G}}(y^{(b)},v)$ and $A^{\mathrm{G}}(y^{(b)},v+1)$ (for uniform $y^{(b)},v$), while making only restricted queries. The simulation will be perfect unless some bad event (which happens with probability $O(d^2/n)$) occurs. This will assure that the error probability of $h$ is bounded by $\delta + O(d^2/n)$.

The query answers of $A^{\mathrm{G}}(y^{(b)},v)$ and $A^{\mathrm{G}}(y^{(b)},v+1)$ are uniform in $\Sigma_b$, unless they query the same element, namely, $\ell_{v+a_0}(i,j) = \ell_{v+a_1}(i',j')$ for $i,j,i',j' \in \mathbb{Z}_n$ and $a_0,a_1 \in \{0,1\}$. We thus require that the mapping preserves equality, namely,
\begin{align}\label{eq:map}
\ell_{v+a_0}(i,j) = \ell_{v+a_1}(i',j') \Leftrightarrow \ell_{v+a_0}(1,map(i,j)) = \ell_{v+a_1}(1,map(i',j')).
\end{align}
Considering the left-hand side, if $i \neq i'$, then the discrete log $v$ can be computed (e.g., for $a_0 = a_1 = 0$, $v = (j' - j) \cdot (i - i')^{-1}$). We refer to this first bad event as a collision, and can bound its probability by $O(d^2/n)$ using Theorem~\ref{thm:DL}. Hence, the analysis will assume that if $\ell_{v+ a_0}(i,j) = \ell_{v+a_1}(i',j')$, then $i = i'$.

The second bad event is that $\ell_{v+a_0}(1,map(i,j)) = \ell_{v+a_1}(1,map(i',j'))$, but $\ell_{v+a_0}(i,j) \neq \ell_{v+a_1}(i',j')$. The probability of this event will be bounded by $O(d^2/n)$ below.

In order to define $map$,\footnote{This mapping was defined in a slightly different context in~\cite{DKK18}.} assume that we fix elements $D_0,D_1,\ldots,D_{n-1}$, where $D_i \in \mathbb{Z}_n$. We set
\[
    map(i,j) =
\begin{cases}
    D_i + j \cdot i^{-1}, & \text{if } i \neq 0\\
    D_i + j, & \text{otherwise},
\end{cases}
\]
where $i^{-1}$ is the multiplicative inverse of $i$ modulo $n$.

Observe that for $i \neq 0$, if $\ell_v(i,j) = \ell_{v+1}(i,j')$, then $j = j' + i$ implying that
$$\ell_v(1,map(i,j)) = v + D_i + j \cdot i^{-1} = v + D_i + (j' + i)\cdot i^{-1} = v + 1 + D_i + j' \cdot i^{-1} = \ell_{v+1}(1,map(i,j')),$$
and the same equality holds for $i=0$. Hence, assuming no collision, the first (right) part of Equation~(\ref{eq:map}) holds for $a_0 \neq a_1$ (it also holds trivially for $a_0 = a_1$).

For the second part, observe that if $\ell_{v+a_0}(1,map(i,j)) = \ell_{v+a_1}(1,map(i',j')$ for $i,i' \neq 0$, then $v + a_0 + D_i + j \cdot i^{-1} = v + a_1 + D_{i'} + j' \cdot (i')^{-1}$. If $i = i'$, we get $i \cdot a_0 + j = i \cdot a_1 + j'$. Hence,
$$\ell_{v+a_0}(i,j) = i \cdot (v + a_0) + j = i \cdot (v + a_1) + j' = \ell_{v+a_1}(i,j'),$$ as required (a similar equality holds when $i = i' =0$). When $i \neq i'$, we get a second bad event and we choose each $D_i \in \mathbb{Z}_n$ uniformly at random to bound its probability (for any specific $a_0,i,j,a_1,i',j'$) by $1/n$. Summing over all $\binom{2d}{2}$ query pairs of $B^{\mathrm{G}}(x^{(b)},v)$ and  $B^{\mathrm{G}}(x^{(b)},v+1)$, we bound the probability of the second bad event by $O(d^2/n)$.

Note that this requires both executions of $B^{\mathrm{G}}(x^{(b)},v)$ and $B^{\mathrm{G}}(x^{(b)},v+1)$ to have shared randomness. However, this requirement can be removed since our model is non-uniform, and $D_i \in \mathbb{Z}_n$ can be fixed by a standard averaging argument.
\end{proof}

Combined with Lemma~\ref{lem:equiv}, we obtain the following corollary.
\begin{corollary}
\label{corr:conv}
For $b \geq 3 \log b$, there exists a reduction that converts any $(n,b,d,\delta)$-DDLA, for prime $n$, to an $(n,b,d,\delta + O(d^2/n))$-LPHS.
\end{corollary}

%% file: 2D-LPHS-compressed-ICALP.tex

\section{Multidimensional LPHS}\label{sec:2dim}

In this section we study the $k$--dimensional generalization of LPHS, focusing mainly on the case $k=2$ (2D-LPHS).
First, in Section~\ref{sec:sub:upper}, we consider the upper bound side.
We start with simple constructions achieving error $\delta=\tilde{O}(d^{-1/2})$ (Section~\ref{sec-simple}) and $\delta=\tilde{O}(d^{-4/5})$ (Section~\ref{sec-lesssimple}). The latter makes a black-box use of the 1-dimensional IRW algorithm of~\cite{DKK18}.
More concretely, for $n = \Omega(d^{6/5})$, any $b \geq 3 \log n$ and any $(r_1,r_2) \in \{0,1\} \times \{0,1\}$, we have
\begin{equation}\label{eq:2Dreq}
\pr_{x \in_R \Sigma_b^{\integers_{n} \times \integers_{n}}}\lbs h(x)-h(x \ll (r_1,r_2)) \neq (r_1,r_2) \rbs = O(d^{-4/5}).
\end{equation}
Since our construction only makes queries in a limited box of dimensions $O(d^{6/5}) \times O(d^{4/5})$ while $n = \Omega(d^{6/5})$, it gives both a cyclic and a non-cyclic LPHS with the same parameters.

This proves a weak version of Theorem~\ref{thm:2dupper}. In Section~\ref{app:2D78} we
 obtain the improved upper bound of Theorem~\ref{thm:2dupper}
by presenting a more intricate algorithm that achieves error rate of $\delta=\tilde O(d^{-7/8})$.
Finally, in Section~\ref{app:2Dopt} we present a heuristic algorithm that we {\em conjecture} to achieve the near-optimal error probability of $\delta=\tilde{O}(d^{-1})$.
This conjecture is supported by experimental evidence.

In Section~\ref{sec:sub:lower} we study limitations of $k$-dimensional LPHS. We prove Theorem~\ref{thm:kdlower}, which for $k=2$ implies that the error probability of a 2D-LPHS must satisfy $\delta=\tilde{\Omega}(d^{-1})$.



\subsection{Upper bounds on 2D-LPHS algorithms}
\label{sec:sub:upper}

We present and analyze 2D-LPHS algorithms achieving error $\delta = \tilde O(d^{-1/2}), \tilde O(d^{-4/5}), \tilde O(d^{-7/8})$, respectively, as well as a heuristic algorithm conjectured to achieve optimal error $\delta = \tilde O(d^{-1})$.


\subsubsection{A simple 2D-LPHS with error rate $\delta = O(d^{-1/2})$}
\label{sec-simple}

We begin by describing a very simple 2D-LPHS algorithm called \textsc{Min-Hash}.

\medskip
\noindent
\textbf{Notation} All integer operations in algorithms within this section are assumed to be floored. For example, we write $x/y$ for $\lfloor x/y \rfloor$ and $\sqrt{n}$ for $\lfloor \sqrt{n} \rfloor$.
\begin{algorithm}
	\Begin{
		\Return $\mathrm{arg\,min}_{i,j\in [0, \sqrt{d}]} \{x[i,j]\}$
	}
	\caption{\textsc{Min-Hash}$(x \in \Sigma_{b}^{\integers_{n}^2},d \in \pintegers)$}
	\label{alg:2Dstage1}
\end{algorithm}

The following lemma captures the performance of \textsc{Min-Hash}.
\begin{lemma}\label{lem:min-hash-pr} For $n = \Omega(d^{1/2})$ and any $b \geq 3 \log n$ and $(r_1,r_2) \in \{0,1\} \times \{0,1\}$,
	\begin{equation}\label{eq:min-hashErr}
	\pr\lbs \textsc{Min-Hash}(x, d)-\textsc{Min-Hash}(x \ll (r_1, r_2) , d) \neq (r_1, r_2) \rbs = O(d^{-1/2}).
	\end{equation}
\end{lemma}

\begin{proof} 
	Notice that no matter what the values of $r_1, r_2 \in \{0, 1\}$ are, both applications of $\textsc{Min-Hash}$ on $x$ and $y$ query the values $G = \{x[i,j]\}_{i,j \in [1, \sqrt{d}]}$, together with some other $2\sqrt{d}-1$ values. Hence $\textsc{Min-Hash}(x, d)$ and $\textsc{Min-Hash}(y, d)$ together read at most $4\sqrt{d}$ values outside of $G$. Under the uniformity assumption of $x$, the probability the minimum of all the symbols read by the two applications of $\textsc{Min-Hash}$ is not in $G$, is bounded by $4\sqrt{d}/d = O(1/\sqrt{d})$. Hence, assuming that the minimum $x[i_0, j_0]$ is in $G$, and that this minimum is unique, we have
	\begin{equation*}
	\begin{gathered}
	\textsc{Min-Hash}(x, d) = (i_0, j_0).\\
	\textsc{Min-Hash}(y, d) = \mathrm{arg\,min}_{i,j\in [0, \sqrt{d}]} \{x[i+r_1,j+r_2)\} = (i_0-r_1, j_0-r_2),
	\end{gathered}
	\end{equation*}
	in which case, $\textsc{Min-Hash}(x, d)-\textsc{Min-Hash}(y, d) = (r_1, r_2)$.
	
	The only thing left for the proof is showing that with a very high probability, the minimum of Algorithm~\ref{alg:2Dstage1} is uniquely attained. It can be easily verified (e.g., by induction on $d$) that the probability that the minimum is not unique, is upper bounded by $d/2^b$. Under the assumption $b \geq 3/2\lg(d)$, this probability is dominated by the $O(1/\sqrt{d})$ error in~\eqref{eq:min-hashErr}.
\end{proof}

\subsubsection{An IRW-based 2D-LPHS with error rate $\delta = O(d^{-4/5})$}
\label{sec-lesssimple}

In this subsection we demonstrate how an 1D-LPHS may be used in order to construct a 2D-LPHS with $\delta = O(d^{-4/5})$. Let us recall the functionality of an optimal 1D-LPHS (see
Theorem~\ref{thm:1dupper}).

\medskip\noindent
\textbf{Optimal 1D-LPHS, rephrased.}
	For $b \geq 3 \log n$, there exists an algorithm $\textsc{optimal1D}\colon \Sigma_b^{\integers_n} \times \integers \to \integers_n$ with the following properties. If $n = \Omega(d^2)$, then
	\[
	\pr_{x\sim \Sigma_b^{\integers_n}}[\textsc{optimal1D}(x, d) - \textsc{optimal1D}( x \ll 1, d) \neq 1] < O(1/d^{2}).
	\]
	
\medskip

The 2D-LPHS is described in the \textsc{Recursive-Hash} algorithm (Algorithm~\ref{alg:2D45}). The algorithm works in two stages, first returning a column $i_1$ and then a row $j_1$. Namely, the parties first try to synchronize on a column via a random walk along their input rows, and then try to synchronize on a row via a random walk along this common column.

Specifically, the column $i_1$ is located by a walk along the fixed input row. The walk uses the \textsc{optimal1D} algorithm with $d/d'-1$ queries (for a parameter $d'$). A query in this algorithm on any column $i_0$ is answered by the \textsc{rec1D} algorithm, which executes \textsc{optimal1D} with $d'$ queries on column $i_0$.

After the column $i_1$ is returned, the \textsc{Recursive-Hash} algorithm runs the \textsc{optimal1D} algorithm along this column, starting from the input row to return the output row $j_1$.


\begin{algorithm}
	\Begin{
		Define $u \in \Sigma_b ^{\integers_n}$ by $u[j] \gets z[i_{0},j]$\\
		$j_0 \gets \textsc{optimal1D}(u, d')$\\
		\Return $z[i_0,j_0 + 10 d'^{2}]$
	}
	\caption{\textsc{rec1D}$(z \in \Sigma_b^{\integers_{n}^2},d' \in \pintegers, i_{0} \in \integers_{n})$}
\end{algorithm}
\begin{algorithm}
	\Begin{
		$i_1 \gets \textsc{optimal1D}(i \mapsto \textsc{rec1D}(z,d^{3/5}-1,i), d^{2/5})$\\
		$j_1 \gets \textsc{optimal1D}(j \mapsto z[i_1,j], d^{2/5})$\\
		\Return $(i_{1},j_{1})$
	}
	\caption{\textsc{Recursive-Hash}$(z \in \Sigma_b^{\integers_{n}^2},d \in \pintegers)$}
	\label{alg:2D45}
\end{algorithm}
\begin{lemma}\label{lem:2D45}
	For $n = \Omega(d^{6/5})$ and any $b \geq 3 \log n$ and $(r_1,r_2) \in \{0,1\} \times \{0,1\}$,
	\begin{equation*}\label{eq:Rec2DErr}
	\pr\lbs \textsc{Rec2D}(x, d)-\textsc{Rec2D}(x \ll (r_1,r_2), d) \neq (r_1, r_2) \rbs =  O(d^{-4/5}).
	\end{equation*}
\end{lemma}

\begin{proof} 
	First, notice $\textsc{Recursive-Hash}(x, d)$ makes at most $d$ queries to $x$. This is because all but $d^{2/5}$ queries are made inside $\textsc{rec1D}$. This procedure is called at most $d^{2/5}$ times, and queries $u$ at most $d^{3/5}-1$ times, which in turn makes exactly one query to $x$.
	
	Second, denoting $y = x \ll (r_1,r_2)$, notice that since $r_1 \in \{0, 1\}$, we have from 
	Theorem \ref{thm:1dupper},
	\[
	\pr \lbs \textsc{rec1D}(x,d^{3/5}-1,i+r_1) \neq \textsc{rec1D}(y,d^{3/5}-1,i) \rbs \leq O(d^{-6/5}).
	\]
	Since $\textsc{Recursive-Hash}$ makes at most $d^{2/5}$ calls to $\textsc{rec1D}$, then except for probability $O(d^{-6/5}\cdot d^{2/5})=O(d^{-4/5})$, we have $\textsc{rec1D}(x,d^{3/5},i+r_1) = \textsc{rec1D}(y,d^{3/5},i)$ for all $i$'s such that both sides were computed in $\textsc{Recursive-Hash}(x,d)$ and $\textsc{Recursive-Hash}(y,d)$, respectively. If the output of $\textsc{rec1D}(x,d^{3/5},i)$ would be uniformly distributed in $\Sigma_b$ (and independent of all other queries),
Theorem~\ref{thm:1dupper}
would imply $i_{1}^{x}+r_1=i_{1}^{y}$ except for probability $O(d^{-4/5})$, where $i_{1}^{z}$ is the value of $i_{1}$ computed during $\textsc{Recursive-Hash}(z,d)$, and $z \in \{x,y\}$. Similarly, $j_{1}$ is synchronized except for probability $O(d^{-4/5})$. In conclusion, the total error probability of the algorithm is $O(d^{-4/5})$. Finally, $x[i_0, j_0+10d'^{2}]$ is indeed uniformly distributed in $\Sigma_b$ independently of all other queries, as \textsc{optimal1D} does not query that value.
\end{proof}

\subsubsection{3-Stage-Hash: a 2D-LPHS algorithm with $\delta = \tilde O(d^{-7/8})$}\label{app:2D78}

In this subsection we prove Theorem~\ref{thm:2dupper}, by presenting the algorithm \textsc{3-Stage-Hash} which achieves error rate of $\delta=\tilde{O}(d^{-7/8})$.
\textsc{3-Stage-Hash} is composed of 3 stages. The first is \textsc{Min-Hash} and we refer to the other two as \textsc{stage2} and \textsc{stage3}.

\begin{algorithm}
	\Begin{
		$(i_0, j_0) \gets \textsc{Min-Hash}(z,\ d/3)$\\
		$(i_1, j_1) \gets \textsc{stage2}(z,\ d/3,\ i_0 + 2\sqrt{d},\ j_0+2\sqrt{d})$\\
		$(i_2, j_2) \gets \textsc{stage3}(z,\ d/3,\ i_1 + 2d^{3/4},\ j_1+2d^{3/4})$\\
		\Return $(i_2, j_2)$
	}
	\caption{\textsc{3-Stage-Hash}$(z \in \Sigma_b^{\integers_n^2},d \in \pintegers)$}
	\label{alg:2D3stage}
\end{algorithm}

\begin{algorithm}[t]
	\Begin{
		$L \gets \sqrt[4]{d'}$\\
		Let $\psi_{1}, \psi_{2} \colon \Sigma_b \to \integers_{L}$ be $\psi_{1}(t)=t\ \mathrm{mod}\,L$ and $\psi_{2}(t)=\lfloor t/L\rfloor\ \mathrm{mod}\,L$\\
		$P, P' \gets \es, \es$\\
		$i \gets i_0$\\
		\For{$s=1$ to $\sqrt{d'}$}{
			$j \gets j_0$\\
			\For{$t=1$ to $\sqrt{d'}$}{
				$P' \gets P' \cup \{(i,j)\}$\\
				$j \gets j + 1 + \psi_{1}(z[i, j])$
			}
			$i \gets i + 1 + \psi_{2}(\min_{p'\in P'}\{z[p']\})$\\
			$P, P' \gets P \cup P', \es$
		}
		\Return $\mathrm{arg\,min}_{p\in P}\{z[p]\}$
	}
	\caption{\textsc{stage2}$(z \in \Sigma_b^{\integers_n^2}, d' \in \pintegers, i_0 \in \integers_{n}, j_0 \in \integers_{n})$}
	\label{alg:stage2}
\end{algorithm}
	
\begin{algorithm}
	\Begin{
		Let $\psi \colon \Sigma_{b} \to \{-d'^{3/8}, \ldots, d'^{3/8}\}$ be $\psi(t) = (t\ \mathrm{mod}\, (2d'^{3/8}+1))-d'^{3/8}$\\
		$P \gets \es$\\
		$(i,j) \gets (i_0, j_0)$\\
		\For{$s=1$ to $d'$}{
			$P \gets P \cup \{(i,j)\}$\\
			$(i,j) \gets (i+1, j + \psi(z[i,j]))$
		}
		\Return $\mathrm{arg\,min}_{p\in P}\{z[p]\}$
	}
	\caption{\textsc{stage3}$(z \in \Sigma_b^{\integers_n^2}, d' \in \pintegers, i_0 \in \integers_{n}, j_0 \in \integers_{n})$}
\end{algorithm}

%


\begin{lemma}\label{lem:2D78-err}
Let $2^b\geq d^4$, $n \geq \Omega(d)$, and $r_1,r_2 \in \{0,1\}$. Then,
	\begin{equation}\label{eq:2D78}
	\pr\lbs \textsc{3-Stage-Hash}(x, d)-\textsc{3-Stage-Hash}(y, d) \neq (r_1, r_2) \rbs < O(d^{-7/8}),
	\end{equation}

\end{lemma}
In order to prove Lemma~\ref{lem:2D78-err}, we need the following facts,
which we prove later.
\begin{lemma}\label{lem:geoms}
	Let $S_1, S_2, \ldots$ be a sequence of i.i.d. geometric random variables: $S_i \sim \mrm{Geom}(p)$. If $K$ is the minimal integer with $\sum_{k=1}^{K} S_k \geq r$,	then $\be[K] \leq rp+1$.
\end{lemma}

\begin{lemma}\label{lem:sparse_walks}
	Let $I_0, I_1, I_2, \ldots$ be a random walk with $I_{k+1}-I_{k}$ being i.i.d. variables distributed as the difference of two independent random variables uniformly distributed in $\{0,1,\ldots, m\}$, with $m \in \pintegers$. If $T$ is the minimal time with $I_T = 0$, then $\be[\min(T, r)] \leq O\left( (m + |I_0|/m) \sqrt{r}\right)$.
\end{lemma}

Given these facts, we turn to Lemma~\ref{lem:2D78-err}.

\begin{proof}[of Lemma~\ref{lem:2D78-err}, sketch]
	Let $(i_k^z, j_k^z)$ be the values computed as $(i_k, j_k)$ at $\textsc{3-Stage-Hash}(z, d)$, for $k \in \{0,1,2\}$ and $z \in \{x,y\}$. We say that $(i_k^x, j_k^x)$ and $(i_k^y, j_k^y)$ are synchronized if $(i_k^x, j_k^x) - (i_k^y, j_k^y) = (r_1, r_2)$. Observe that if $(i_k^x, j_k^x)$ and $(i_k^y, j_k^y)$ are synchronized, then so are $(i_{k+1}^x, j_{k+1}^x)$ and $(i_{k+1}^y, j_{k+1}^y)$. This is because each stage $k+1$ of $\textsc{3-Stage-Hash}(z,d)$ deterministically depends on values queried from $z$ with offset $(i_k^z, j_k^z)$, so that evaluations keep being aligned. Overall,
	\begin{align}
	\begin{split}
		\delta \ddd &\pr_{x,y,r_1,r_2}[\textsc{3-Stage-Hash}(x, d) - \textsc{3-Stage-Hash}(y, d) \neq (r_1,r_2)]
		\\
		&\leq
		\pr[(i_0^x, j_0^x) - (i_0^y, j_0^y) \neq (r_1, r_2)] \cdot
		\\
		&\cdot \prod_{k=1}^{2}\pr\left[\given{(i_k^x, j_k^x) - (i_k^y, j_k^y) \neq (r_1, r_2)}{(i_{k-1}^x, j_{k-1}^x) - (i_{k-1}^y, j_{k-1}^y) \neq (r_1,r_2)}\right].
	\end{split}
	\end{align}
	Hence, to bound $\delta$ it is sufficient to verify the following three claims:
	\begin{enumerate}
	\item $\pr[(i_0^x, j_0^x) - (i_0^y, j_0^y) \neq (r_1, r_2)]\leq O(1/\sqrt{d})$.
	\item $\pr\left[\given{(i_1^x, j_1^x) - (i_1^y, j_1^y) \neq (r_1, r_2)}{(i_0^x, j_0^x) - (i_0^y, j_0^y) \neq (r_1, r_2)}\right] \leq O(1/\sqrt[4]{d})$.
	\item $\pr\left[\given{(i_2^x, j_2^x) - (i_2^y, j_2^y) \neq (r_1, r_2)}{(i_1^x, j_1^x) - (i_1^y, j_1^y) \neq (r_1, r_2)}\right] \leq O(1/\sqrt[8]{d})$.
	\end{enumerate}
	\medskip\noindent \textbf{\textit{Claim 1)}} Follows from Lemma~\ref{lem:min-hash-pr}.
	
	\medskip\noindent \textbf{\textit{Claim 2)}}	
	Since $\textsc{Min-Hash}(z,d)$ scans a $\sqrt{d} \times \sqrt{d}$ area and outputs a point inside it, we are guaranteed that $|i_0^x-i_0^y| \leq \sqrt{d}+1$ and $|j_0^x-j_0^y| \leq \sqrt{d} + 1$. Because $\textsc{stage2}$ is fed with the output point of $\textsc{Min-Hash}$ shifted by $2\sqrt{d}$ in each axis, its queries do not overlap these of $\textsc{Min-Hash}$, and its performance is independent of the $\textsc{Min-Hash}$ phase. Moreover, $\textsc{stage2}$ can be modeled as a random walk on the $i$ axis, whose steps are integers uniformly distributed in $[1, \sqrt{d'}]$ (with $d'=d/3$, as in $\textsc{stage2}$), which are determined by some random walk on the $j$ axis. Denote by $I_1^x, \ldots, I_{\sqrt{d'}}^x$ and $I_1^y, \ldots, I_{\sqrt{d'}}^y$ the sequences of $i$'s observed by $\textsc{stage2}$ applied on $x$ and on $y$.
	
	In order to compute the probability that the outputs of $\textsc{stage2}(x, d')$ and $\textsc{stage2}(y, d')$ are not synchronized, it is sufficient (following the proof of Lemma~\ref{lem:min-hash-pr}) to count the number of queries that the two processes make, which are not shared. These queries can be classified into two categories: queries with non-shared $i$, and queries with shared $i$ and non-shared $j$. Our goal is to show each class contains on average $O(d'^{3/4})$ such queries, implying that the probability of the outputs not being synchronized is $O(d'^{3/4}/d')$ (similarly to Lemma~\ref{lem:min-hash-pr}).
	We start by reasoning about the first class of queries, and then proceed to the second.
	
	Let $U_1$ denote the total number of $i$-steps until $i^x$ and $i^y$ are synchronized (i.e. $U_1 = k+k'$ when $k,k'$ are minimal with $I_k^x-I_{k'}^y=r_1$). Up to this point, the two $\textsc{stage2}$ applications act independently, as their queries do not overlap. Using~\cite[Lemmas 3,5]{DKK18} with $b \leq \sqrt{d}+1$ and $L=\sqrt[4]{d'}$ we see $\be[U_1] \leq O(d^{1/4})$. Next, we note that once $I_k^x-I_{k'}^y=r_1$, it is likely that $I_{k+1}^x-I_{k'+1}^y=r_1$. Specifically, we will show
	\[
		\pr \left[ \given{I_{k+1}^x-I_{k'+1}^y \neq r_1}{I_{k}^x-I_{k'}^y=r_1} \right] \leq O(1/\sqrt{d}).
	\]
	Assuming this, the two walks make $U_1$ unsynchronized steps, then $S_1$ synchronized steps with $S_1 \sim \mrm{Geom}(O(1/\sqrt{d}))$ distributed geometrically. The walks then make another $U_2$ unsynchronized steps, with~\cite[Lemma 5]{DKK18} yielding $\be[U_2] \leq O(d^{1/4})$, followed by $S_2$ synchronized steps with $S_2 \sim \mrm{Geom}(O(1/\sqrt{d}))$. The process continues this way until one of the walks has completed its $\sqrt{d'}$ steps. Using Lemma~\ref{lem:geoms}, the expected number of such phases of synchronization-unsychronization is $\leq O(\sqrt{d'} \cdot \sqrt{1/d} + 1) = O(1)$. Combining this with the fact that $\be[U_k] = O(d^{1/4})$, we deduce that the expected number of unsynchronized $i$-steps is $O(d^{1/4})$. Each such step involves $\sqrt{d'}$ $j$-steps, so the total number of queries with non-shared $i$ is $O(d^{3/4})$.
	
	Regarding the steps with shared $i$ and non-shared $j$, random-walk arguments similar to the above argument imply that since on each shared $i$, the two $j$-walks start with distance $O(\sqrt{d})$, and have steps of size $\Theta(d^{1/4})$, they are expected to meet after $O(d^{1/4})$ queries (follows from~\cite[Lemmas 3,5]{DKK18}). Since there are $O(\sqrt{d})$ $i$-steps, the total number of non-shared such queries is $O(d^{3/4})$ as well.
	
	\medskip\noindent \textbf{\textit{Claim 3)}}
	Similarly to the previous claim, the queries made by stages before $\textsc{stage3}$ are confined to a square area of size $(d^{3/4}+3\sqrt{d})\times (d^{3/4}+3\sqrt{d})$, and since the input is shifted by $2d^{3/4}$, the queries of $\textsc{stage3}$ do not overlap previous stages. Note that the queries made by $\textsc{stage3}$ are confined to a $2d \times 2d$ square except for a negligibly small error prbobability ($\exp(-\Omega(d^{1/4}))$) due to Hoeffding's inequality, and since $n \geq \Omega(d)$ (recall $x,y \in \Sigma_b^{\integers_{n}^2}$), the queries of the different stages do not overlap (with high probability) even though the index space of $x,y$ is cyclic.

	It is clear that once the two walks of $\textsc{stage3}$ on $x$ and on $y$ are synchronized, they remain synchronized. Let $T$ be the total number of steps until the two walks share a point. There are at most $\min(2T, d')$ steps which are not shared, and the failure probability is $\leq \min(2T, d') / d'$, similarly to the proof of Lemma~\ref{lem:min-hash-pr}. Thus, it is sufficient to verify $\be[\min(T, d)] \leq d^{7/8}$. Clearly, after $|i_1^x-i_1^y|$ steps, the walks are being synchronized with respect to the $i$-axis. Let $J$ denote the random variable measuring their distance on the $j$-axis, once the walks first share this same $i$. Since each of the advances of $j$ are independent of the other steps,
	\[
		\be[J^2] = |j_1^x - j_1^y|^2 + \sum_{t=0}^{|i_1^x-i_1^y|} \be[S_t^2],
	\]
	where $S_t$ is the jump on the $j$-axis on the $t$-step of the runner-up walk. In particular $\be[S_t^2] \leq ((d/3)^{3/8})^2 \leq d^{3/4}$. Since $|i_1^x-i_1^y| \leq 2d^{3/4}$, we overall deduce $\be[J^2] \leq O(d^{3/2})$.
	
	From this point on, the walks of $\textsc{stage3}(x, d')$ and $\textsc{stage3}(y, d')$ keep being aligned with respect to the $i$-axis. Once they meet on the $j$-axis, they will remain synchronized. The distance on the $j$-axis between the walks can be modeled as a one dimensional random walk, starting at $J$, and having independent steps whose lengths are a difference of two independent variables uniformly distributed in $\{0, 1, \ldots, 2d'^{3/8}\}$. Once this difference walk hits $0$, the walks keep being synchronized. Lemma~\ref{lem:sparse_walks} then immediately yields
	\[
		\be[\min(T, d')] \leq |i_1^x-i_1^y| + O\left( (d'^{3/8}+\be|J|/d'^{3/8})\sqrt{d'} \right).
	\]
	Substituting $\be[J]^2 \leq \be[J^2] = O(d^{3/2})$, we obtain $\be[\min(T, d')] \leq d^{7/8}$, as required.
\end{proof}

We now fill in the proofs of the above-stated facts.

\begin{proof}[of Lemma~\ref{lem:geoms}]
	Since each $S_i$ counts the number of Bernoulli-$p$ variables until success, $K$ distributes as $1+$ the number of successful Bernoulli-$p$ variables, out of $r$. This interpretation immediately gives $\be[K] = p(r-1)+1$.
\end{proof}

\begin{proof}[of Lemma~\ref{lem:sparse_walks}]
	Let $T_0, T_1, T_2, \ldots$ be the sequence of times $t\geq 0$ with $|I_t| \leq \frac{m+1}{2}$ (in increasing order).
	Observe that for all $i$, the event $I_{T_i+1}=0$ happens with probability $\geq \frac{1}{2(m+1)}$, even when conditioning on the trajectory $\{I_t\}_{t \leq T_i}$ up to time $T_i$.
	Let $K$ be minimal with $I_{T_K + 1} = 0$. Using the (probabilistic) chain rule this observation means that
	\[
		\pr [K \geq k ] \leq (1 - 1/(2m+2))^k.
	\]
	When combined with
	\[
		\be[\min(T_K, r)] \leq \be[\min(T_0, r)] + \sum_{k=1}^{\infty} \be\left[\one_{\{K \geq k-1\}}\cdot \min(T_k-T_{k-1}, r)\right],
	\]
	we deduce
	\begin{equation}\label{eq:good}
		\be[\min(T_K, r)] \leq \be[\min(T_0, r)] + \sum_{k=1}^{\infty} \left(\frac{2m+1}{2m+2}\right)^{k-1} \be\left[\given{\min(T_k - T_{k-1}, r)}{K \geq k-1}\right].
	\end{equation}
	Clearly, upper bounding $\be[\min(T_K, r)]$ is relevant, as if $T$ is the minimal time with $I_T = 0$, then $T \leq T_K+1$, and in particular, $\min(T, r) \leq \min(T_K, r) + 1$.
	We claim the following:
	\begin{enumerate}
	\item $\be[\min(T_0, r)] \leq O(1 + |I_0|\sqrt{r}/m )$.
	\item For all $i$, and all values of $\{I_t\}_{t \leq T_i}$, $\be\left[\given{\min(T_{i+1} - T_{i}, r)}{I_0, I_1,\ldots, I_{T_i}}\right] \leq O(\sqrt{r})$.
	\end{enumerate}
	These claims together with~\eqref{eq:good} and $\min(T, r) \leq \min(T_K, r) + 1$ give
	\[
		\be[\min(T, r)] \leq O(1 + |I_0|\sqrt{r}/m) + \sum_{k=1}^{\infty} \left(\frac{2m+1}{2m+2}\right)^{k-1} O(\sqrt{r}) \leq O(|I_0|\sqrt{r}/m + m\sqrt{r}),
	\]
	as required.
	Let us verify the above claims.
	
	\medskip\noindent \textbf{\textit{Claim 2)}} This is a specialization of Claim~1 to the walk $I_{T_i+1}, I_{T_i+2}, \ldots$, satisfying $|I_{T_i+1}| \leq \frac{3m+1}{2}$.
	
	\medskip\noindent \textbf{\textit{Claim 1)}}	
	Without loss of generality assume $I_0 \geq 0$ (due to symmetry).
	Write $L = m\sqrt{r}$.
	Instead of the stopping time $\min(T_0, r)$, consider the stopping time $T'$ which is the minimal (time) $t\geq 0$ with $I_{t} \leq (m+1)/2$ or $I_{t} > L$. It is standard to show that $\pr[T' > k]$ decreases exponentially fast with $k$ (albeit with deficient constants), and so all quantities presented in the proof will turn out to be finite (in particular, $\be[T']$).
	
	Since the definition of $T_0$ is similar to that of $T'$, except that the latter allows to stop also when $I_{t} > L$, we may upper bound $\be[\min(T_0, r)]$ by $\be[T'] + r\pr[I_{T'} > L]$, i.e., we compensate by $r$ in all cases when $T'$ is not identical to $T_0$.
	
	Since $I_{k+1}-I_{k}$ is a symmetric random variable, and is independent of $I_0, \ldots, I_k$, the sequence $I_0, I_1, \ldots$ constitutes a martingale. In particular, by the \emph{Optional stopping theorem}, $\be[I_{T'}] = I_0$. This implies, together with that $|I_{k+1}-I_{k}| \leq m$
	\[
		I_0 = \be[I_{T'}] \geq -(m+1)/2 \cdot \pr[I_{T'} \leq (m+1)/2] + L \pr[I_{T'} > L] \geq L \pr[I_{T'} > L] - m,
	\]
	and in particular, $\pr[I_{T'} > L] \leq (m+I_0)/L$. Recall that our goal of proving \emph{claim 1} is non-trivial only when $I_0 > (m+1)/2$, and so we may assume $\pr[I_{T'} > L] \leq O(I_0/L)$. In particular,
	\[
		\be[\min(T_0, r)] \leq \be[T'] + r\cdot O(I_0/L) \leq \be[T'] + O(I_0\sqrt{r}/m).
	\]
	Thus the claim is implied from $\be[T'] \leq O(I_0 \sqrt{r}/m)$ which we now prove.
	
	Write $s = \be[(I_{k+1}-I_k)^2] = (m^2+2m)/6$, and consider the sequence of random variables
	\[
		\left( I_k^2 - s k \right)_{k=0}^{\infty}.
	\]
	We claim it is a martingale. Indeed:
	\begin{equation*}
	\begin{aligned}
		\be\left[ \given{I_{k+1}^2 - s (k+1)}{I_0, \ldots, I_k} \right]
		& = \be\left[ \given{I_{k}^2 + 2I_k (I_{k+1}-I_k) + ((I_{k+1}-I_k)^2-s) - s k}{I_0, \ldots, I_k} \right]	\\
		& = I_{k}^2 - s k \\
	\end{aligned}
	\end{equation*}
	The first equality uses $(a+b)^2 = a^2+2ab+b^2$, and the second uses the fact that $I_{k+1}-I_k$ is a symmetric random variable independent of $I_0, \ldots, I_k$, having variance $s$.
	The \emph{Optional stopping theorem} thus implies $I_0^2 = \be[I_{T'}^2 - s T']$, yielding $\be[T'] \leq \be[I_{T'}^2] / s$. To bound $\be[I_{T'}^2]$, we recall that $I_{T'}$ has absolute value $\leq (m+1)/2$ with probability $\leq 1$ (trivially), and is between $L$ and $L+m$ with probability $\leq O(I_0/L)$. Hence $\be[I_{T'}^2] \leq m^2 + (L+m)^2 O(I_0 / L)$. Since $L \geq m$, we deduce $\be[I_{T'}^2] \leq O(s + I_0 L)$. Overall,
	\[
	\be[T'] \leq O(1 + I_0 L / s) \leq O(1 + I_0\sqrt{r}/m),
	\]
	and the proof is complete.
\end{proof}


\subsubsection{Conjectured optimal algorithm}\label{app:2Dopt}
The 2D-LPHS algorithms presented in the previous sections have the property of not treating both axes symmetrically. For example, $\textsc{Recursive-Hash}$ iterates over several $i_{0}$'s, and for each of them it makes many queries of the form $x[i_0, j]$ for different $j$'s. Except for not being aesthetic, this asymmetry has other disadvantages. For example, it is not obvious how to generalize these algorithms to higher dimensions. More importantly, these algorithms (that we considered) do not have optimal dependence of $\delta$ on $d$.

We conjecture that the following symmetric algorithm (\textsc{Random-Walk-Hash}) has the optimal performance of $\delta = \widetilde{O}(1/d)$. However, we were not able to rigorously analyze it.


\begin{algorithm}
Let $\psi_{1}, \psi_{2}\colon \Sigma_b \to \{-L, \ldots, L\}$ be independent random functions \\
	\Begin{
		$P \gets \mathrm{list}()$ \\
		\For{$s \gets 0\ldots d-1$}{
			$P[s] \gets (i,j)$\\
			$v \gets z[i,j]$\\
			$(i,j) \gets (i + \psi_{1}(v), j + \psi_{2}(v))$\\
			\If{$(i,j) \in P$} {
				Let $t$ be the only index satisfying $P[t]=(i,j)$\\
				$k \gets \mathrm{arg\,min}_{u\in [t,s]} \{z[P[u]]\}$\\
				$(i,j)\gets P[k]$\\
				\While{$(i,j) \in P$} {
					$j \gets j+1$
				}
			}
		}
		\Return $\mathrm{arg\,min}_{(i',j')\in P} \{z[i',j']\}$
	}
	\caption{\textsc{rw-stage}$(z \in \Sigma_b ^{\integers_n^2}, d \in \pintegers, L \in \pintegers, i \in \integers_{n}, j \in \integers_{n})$}
\end{algorithm}
\begin{algorithm}\label{alg:2Doptimal}
	\Begin{
		$I \gets \lg \lg (d)$\\
		$d' \gets d / I$\\
		$(i_0, j_0) \gets \textsc{Min-Hash}(z,\ d')$\\
		$(i_1, j_1) \gets \textsc{rw-stage}(z,\ d',\ d'^{1/4},\ i_0,\ j_0)$\\
		$(i_2, j_2) \gets \textsc{rw-stage}(z,\ d',\ d'^{3/8},\ i_1,\ j_1)$\\
		$(i_3, j_3) \gets \textsc{rw-stage}(z,\ d',\ d'^{7/16},\ i_2,\ j_2)$\\
		$\vdots \qquad \vdots\qquad \vdots\qquad \vdots\qquad \vdots\qquad \vdots\qquad \vdots\qquad \vdots\qquad \vdots$\\
		$(i_I, j_I) \gets \textsc{rw-stage}(z,\ d',\ \sqrt{d'/2},\ i_{I-1},\ j_{I-1})$\\
		\Return $(i_I, j_I)$
	}
	\caption{\textsc{Random-Walk-Hash}$(z \in \Sigma_b ^{\integers_n^2},d \in \pintegers)$}
\end{algorithm}

\paragraph{Heuristic performance.}
Here we heuristically describe why we expect the algorithm \textsc{Random-Walk-Hash} to achieve $\delta = \widetilde{O}(1/d)$. 

The main heuristic assumption we make is that each $\textsc{rw-stage}(x, d, L, i, j)$ can be modeled by a random walk on $\integers^{2}$, starting at $(i,j)$ and having independent steps which are uniformly distributed on each axis as $\sim U(-L, L)$. We further assume that once the two walks of $\textsc{rw-stage}(x)$ and $\textsc{rw-stage}(y)$ are synchronized (collided), they remain synchronized. Moreover, we recall that the output location of a $\textsc{rw-stage}(x)$ is the point visited in this walk having the minimal $x$-value.

\begin{remark} These assumptions are not precise mainly because we need the steps to be both deterministic and independent (with respect to the input's randomness) of the previous steps. In practice we cannot guarantee independence, since the random walk occasionally runs into loops. We try to break these in a canonical way, which complicates the analysis. If the algorithm would make monotone queries along (at least) one axis (as the one-dimensional algorithm), then it would avoid loops and its analysis would be much simpler. Unfortunately, we do not know how to design such an algorithm with similar performance.
\end{remark}

Based on the heuristic assumptions above, an analysis of \textsc{Random-Walk-Hash} would follow from the following two claims:
\begin{itemize}
	\item Let $T$ be (a random variable measuring) the meeting time of two random walks on $\integers^{2}$, starting at distance $D$ (in $L_{1}$ norm), and making steps uniformly distributed in $\{-L, -L+1, \ldots , L\}$ on both axis. Then $\be[\min(d',T)] \leq O(\sqrt{d'}(L + D/L))$.\footnote{The parameters are chosen so that the parties meet within an expected number of $O(\sqrt{d'}(L + D/L))$ steps, while they do not meet within $d'$ steps with probability $O((L + D/L)/\sqrt{d'})$.}
	We say that walks $A,B$ `meet' in time $t$ if $t$ is minimal so that $\exists i,j \leq t$ with $\mathrm{location}_{i}(A)=\mathrm{location}_{j}(B)$.
	\item The expected (Manhattan) distance between the start and final point of a 2D-walk with $d'$ steps uniformly distributed in $\{-L, -L+1, \ldots , L\}$ on both axis, is $O(L\cdot \sqrt{d'})$.
\end{itemize}
Let us analyze the first few stages of \textsc{Random-Walk-Hash} using these claims (which we do not prove here).

Just after the first stage, which is \textsc{Min-Hash}, the walks of $\textsc{Random-Walk-Hash}(x)$ and
$\textsc{Random-}$ $\textsc{Walk-Hash}(y)$ are synchronized except for probability $O(1/\sqrt{d'})$ (Lemma~\ref{lem:min-hash-pr}), and in case of this failure event, the distance of the two walks has expected value $O(\sqrt{d'})$.

At the second stage, $\textsc{rw-stage}(x,d',d'^{1/4})$, the initial distance between the walks is $D=O(d'^{1/2})$ and $L=d'^{1/4}$. Using the above claims, and a Markov inequality, the random walks would synchronize except for probability $O(D/L+L)\sqrt{d'}/d' = O(d'^{-1/4})$, and in case of failure, the expected distance is $O(d'^{3/4})$.

At the third stage, $\textsc{rw-stage}(x,d',d'^{3/8})$, $D = O(d'^{3/4}), L=d'^{3/8}$ and the failure probability becomes $O(d'^{3/8})/\sqrt{d'}=O(d'^{-1/8})$, and the distance upon failure is $=O(d'^{7/8})$.

Continuing this heuristic to later stages we get that the total failure probability, which is the product of failure probabilities of all the stages, is $2^{O(I)}/d'=\widetilde{O}(d^{-1})$.

\subsection{Lower bounds on 2D-LPHS algorithms}
\label{sec:sub:lower}


In this section we prove Theorem~\ref{thm:kdlower} for $k=2$ (i.e., 2D-LPHS). The proof for any other value of $k > 1 $ is similar. In particular, we show that any 2D-LPHS algorithm satisfies $\delta \geq \Omega(1/d)$, given $n \geq 2d$.

\begin{lemma}[Bigger shifts]\label{lem:2Dextension}
	Let $h$ be a 2D $(n,b,d,\delta)$-LPHS. Let $r_{1}', r_{2}' \in \pintegers$, and $y'=x' \ll (r_1', r_2')$ where $x' \in \Sigma_b^{\integers_{n}^2}$ is a uniformly random string. Then,
	\[
	\pr\lbs h(x', d) - h(y', d)\neq (r_1', r_2') \rbs \leq \max\{r_1', r_2'\}\delta.
	\]
\end{lemma}
\begin{proof}
	Write $I = \max\{r_1', r_2'\}$ and set $r_{1}'^{(i)} = \min(r_{1}', i), r_{2}'^{(i)} = \min(r_{2}', i)$.
	Note $\forall i \colon (r_{1}'^{(i+1)}-r_{1}'^{(i)}), (r_{2}'^{(i+1)}-r_{2}'^{(i)}) \in \{0, 1\}$. Define the strings $x'_{i} \in \Sigma_b^{\integers_{n}^2}$ by $x'_{i}=x' \ll (r_1'^{(i)}, r_2'^{(i)})$ for $i=0,\ldots, I$. Since each $x'_{i}$ is a uniformly random function (because $x'$ is), we may use~\eqref{eq:2Dreq} to deduce
	\[
		\delta_i \doteqdot \pr \lbr A(x'_{i}, d) - A(x'_{i+1}, d) \neq (r_1'^{(i+1)}-r_1^{(i)}, r_{2}'^{(i+1)}-r_{2}'^{(i)}) \rbr \leq \delta.
	\]
	Notice $x'=x'_{0}$ and $y'=x'_{I}$. Using a union-bound argument, we conclude
	\[
	\pr\lbs A(x', d) - A(y', d) \neq (r_{1}', r_{2}')\rbs \leq \sum_{i=0}^{I-1} \delta_i \leq I \delta
	\]
	as required.
\end{proof}

The following lemma implies Theorem~\ref{thm:kdlower} for $k=2$.
\begin{lemma}\label{lem:2DLB}
For $n > 2d$ and any $b > 0$, every 2D (cyclic or non-cyclic) $(n,b,d,\delta)$-LPHS satisfies $\delta \geq 1/(3d)$.
\end{lemma}

\begin{proof}
	Consider the set of queries $P$ made to a uniformly random string $x \in \Sigma_b^{\integers_{n}^2}$ by an $(n,b,d,\delta)$-LPHS $h(x, d)$. That is, $P$ is a random variable whose values are sets of sizes $\leq d$ of $(i,j)$ pairs.
	
	Let $x', y' \in \Sigma_b^{\integers_n^2}$ be two uniformly random strings related by $y' = x' \ll (r_1', r_2')$ for independent uniform variables $r_1', r_2' \sim \{0,1,\ldots, 2d\}$.
	Using Lemma~\ref{lem:2Dextension},
	\begin{equation}\label{eq:aaa}
	\begin{split}
	\pr[h(x',d)-h(y',d) \neq (r_1', r_2')] & = \be_{r_1', r_2'} \lbs \pr\lbs \given{h(y',d)-h(x',d) \neq (r_1', r_2')}{r_{1}', r_{2}'} \rbs \rbs \\
	& \leq \be_{r_{1}', r_{2}'} \lbs \max(r_{1}', r_{2}')\delta \rbs \leq 2d\delta.
	\end{split}
	\end{equation}
	
	Let $P_{x'}, P_{y'}$ be copies of $P$ which are the sets of queries issued by $h(x', d), h(y', d)$, respectively. Generally, the random variables $P_{x'}, P_{y'}$ are dependent. However, we are going to see they are only slightly dependent. Indeed, suppose $P_{1}, P_{2}$ are two values of $P$. We claim that given specific values of $r_{1}', r_{2}'$ (call these $r_1'', r_2''$) so that $P_{1}$ and (the Minkowski sum) $P_{2} + \{(r_1'', r_2'')\}$ are disjoint, we have
	\begin{equation}\label{eq:lbdis}
	\pr\lbs \given{P_{x'} = P_{1} \wedge P_{y'} = P_{2}}{(r_{1}',r_{2}')=(r_{1}'', r_{2}'')} \rbs = \pr\lbs P = P_{1}\rbs \cdot \pr\lbs P = P_{2}\rbs.
	\end{equation}
	This is because the event $P_{x'} = P_{1}$ depends only on $\restrict{x'}{P_{1}}$ (i.e. lies inside the $\sigma$-algebra generated by $\restrict{x'}{P_{1}}$)\footnote{If the LPHS is probabilistic, then we should add the algorithm's randomness into the $\sigma$-algebra. Since this randomness is independent of all other random variables, the proof applies verbatim.}, which is independent of $\restrict{y'}{P_{2}}$, given that $P_{1}$ and (the Minkowski sum) $P_{2} + \{(r_1'', r_2'')\}$ are disjoint, as different entries of $x'$ are independent.
	Consider the `conditional' random variable
	\[
		X_{r_{1}'', r_{2}'', P_{1}, P_{2}} \ddd \lbs \given{h(x',d)-h(y',d)}{r_1'=r_1'', r_2'=r_2'', P_{x'}=P_1,P_{y'}=P_2} \rbs,
	\]
	which is defined on the part of the probability space in which $r_1'=r_1'', r_2'=r_2'', P_{x'}=P_1,P_{y'}=P_2$.
	We use the law of total probability to compute
	\begin{equation*}
	\begin{split}
	\pr[h(x',d)&-h(y',d) = (r_1', r_2')]\\
	&= \be_{r_1', r_2', P_{x'}, P_{y'}} \lbs \pr\lbs \given{h(x',d)-h(y',d) = (r_1', r_2')}{r_1', r_2',P_{x'}, P_{y'}} \rbs \rbs \\
	&= \be_{r_1', r_2'} \lbs \sum_{P_1, P_2} \pr\lbs X_{r_1',r_2',P_1,P_2} = (r_1',r_2')\rbs \cdot \pr\lbs \given{P_{x'}=P_1,P_{y'}=P_2}{r_1', r_2'}\rbs \rbs \\
	& \stackrel{(a)}{\leq}
	\underbrace{\be_{r_1', r_2'} \lbs \sum_{\substack{P_1, P_2 \colon\\ ((r_1', r_2') \in P_{1} - P_{2})}} \pr\lbs \given{P_{x'}=P_1,P_{y'}=P_2}{r_1', r_2'}\rbs \rbs}_{Q_{1}} + \\
	&+\underbrace{\be_{r_1', r_2'} \lbs \sum_{\substack{P_1, P_2 \colon\\ ((r_1', r_2') \notin P_{1} - P_{2})}} \pr\lbs X_{r_1',r_2',P_1,P_2} = (r_1',r_2')\rbs \cdot \pr [P=P_1] \pr[P=P_2] \rbs}_{Q_{2}} \\
	& \stackrel{(b)}{\leq} \frac{d^{2}}{(2d+1)^2} + \frac{1}{(2d+1)^2} = \frac{d^2+1}{(2d+1)^2}
	,
	\end{split}
	\end{equation*}
	where $P_{1} - P_{2}$ is a Minkowski difference.
	Inequality $(a)$ follows from~\eqref{eq:lbdis} and the fact that probabilities are upper bounded by $1$. Inequality $(b)$ is the key argument. To bound $Q_{1}$ (by $d^2/(2d+1)^2$) we use
	\[
	Q_{1} + \underbrace{\be_{r_1', r_2'} \lbs \sum_{\substack{P_1, P_2 \colon\\ ((r_1', r_2') \notin P_{1} - P_{2})}} \pr [P=P_1] \pr[P=P_2] \rbs}_{Q_{3}} = 1.
	\]
	Exchanging summation order and using $\lba P_{1} - P_{2} \rba \leq |P_{1}|\cdot |P_{2}| \leq d^2$, which holds since $h(x,d)$ makes at most $d$ queries, we see that $Q_{3} \geq 1-d^2/(2d+1)^2$. Notice we use here $n > 2d$. This proves $Q_{1} \leq d^2 / (2d+1)^2$.
	
	To bound $Q_{2}$, we note that every $(P_1, P_2)$ contributes at most
	\[
	\pr[P=P_1]\pr[P=P_2]/(2d+1)^2
	\]
	to $Q_{2}$. To see this, observe that the distribution of the random variable $X_{r_1'',r_2'',P_1,P_2}$ does not depend on the particular value of $r_1'', r_2''$ (given that $P_1 + (r_1'', r_2'')$ is disjoint from $P_2$),
	and this random variable always attains a single value (that is, a random variable $X$ and a set $E$ always satisfy $\sum_{e \in E} \pr[X=e]\leq 1$).
	Hence
	\[
		Q_{2} \leq \sum_{P_1, P_2} \frac{\pr[P = P_1] \pr[P = P_2]}{(2d+1)^2} = \frac{1}{(2d+1)^2}.
	\]
	
	Overall, we deduce $2d\delta + (d^2+1)/(2d+1)^2 \geq 1$, implying $\delta \geq 3/(8d)$.
\end{proof}

\begin{remark}[Extending Lemma~\ref{lem:2DLB}  to higher dimensions]
	The proof of Lemma~\ref{lem:2DLB} for $k=2$ readily extends to a lower bound on the error probability of any $k$-dimensional LPHS with $k>2$ (the case $k=1$ follows from the lower bound in~\cite{DKK18} and our generic model equivalence with LPHS). Concretely, for a $k$-dimensional LPHS we have $\delta \geq 1/(3d^{2/k})$ whenever $n>(2d)^{2/k}$, implying Theorem~\ref{thm:kdlower}.

	The extension to a general dimension $k$ requires the following modifications. First, we use a distance-extension lemma, analogous to Lemma~\ref{lem:2Dextension}, in a way similar to~\eqref{eq:aaa}. This step reduces our task to showing that no algorithm can synchronize on random inputs $x,y \in \Sigma_b^{\integers_n^k}$ with probability higher than, say $1/2$, where $y$ is a random $k$-dimensional shift of $x$ by about $(2 d)^{2/k}$ in every axis.

	Then, we observe that in the event that the LPHS applied on $x$ and $y$ queries disjoint input cells (we think of $x$ and $y$ as inlaid in a common landscape), synchronization is unlikely, as expressed by the bound on $Q_2$ in the proof of Lemma~\ref{lem:2DLB}. Hence, the synchronization probability is dominated by the probability that the LPHS queries a shared input cell. To bound this latter probability we use a birthday-paradox argument similar to the bound on $Q_1$ in the proof of Lemma~\ref{lem:2DLB}: there are at least $((2d)^{2/k})^k = 4d^2$ possible shifts, while there are only $d \times d$ pairs of queries that may collide -- any of the $d$ queries made to $x$ may collide with any of the $d$ queries made to $y$. It follows that there is a probability of at most $1/4$ to have a shared query, concluding the argument.
\end{remark}

%% file: WorstCase.tex

\section{LPHS for Worst-Case Inputs} 	
\label{sec-worstcase}


The basic definition of LPHS provides guarantees for {\em random} input strings. This directly aligns with some applications, where the strings to synchronize can be chosen in such manner (e.g., when broadcasting a random synchronization string). In many other application settings, however, we may wish to perform locality-sensitive hashing on inputs whose structure is not uniform, e.g.\ location coordination given substrings of an existing DNA string or image. Here, it is not sufficient to provide small error probability over a random input; rather, the input will be fixed, and we will wish to obtain small error over a random choice of hash function.  In this section, we address this notion of locality-preserving hashing for shifts (LPHS) for {\em worst-case inputs}.

The main result in this section is showing how to use an underlying LPHS for average-case inputs to obtain LPHS for worst-case inputs that do not exhibit too much regularity (in which case one cannot hope to synchronize shifts effectively).
Our approach first embeds the input space endowed with shift metric into an intermediate space over a larger alphabet also endowed with shift metric, with the promise of \emph{all symbols being distinct}. This can be viewed as an analog of approaches taken, e.g.\ by~\cite{CharikarK06}, for the case of edit distance and so-called \emph{Ulam distance} over permutations; however, for our relaxation to shift metric as opposed to edit distance, we can accommodate weaker restrictions on input strings.

In the following subsections, we address LPHS for worst-case inputs in the settings of cyclic and non-cyclic shifts, respectively. For simplicity, we restrict attention to the case of alphabet $\zo$; that is, inputs $x \in \zo^n$.

\subsection{Cyclic LPHS for Worst-Case Inputs}
\label{sec:worst}

We begin with the somewhat simpler setting of cyclic shifts. To ensure that our worst-case definition extends from a single shift to multiple shifts, it will be convenient to define the following closure operator.
\begin{definition}
\label{def:clos}[Closure under cyclic shifts]
For a set of inputs $X \subseteq \zo^n$, we let
	\[ {\hat X}=\{ x\ll i\,:\, x\in X,\, 0\le i<n \}.\]	
\end{definition}

\begin{definition}
\label{def:worst}[Worst-case cyclic LPHS]
Let $X \subseteq \zo^n$ be a set of inputs. A family $\calH$ of hash functions $h: \zo^n \to \Z_n$ is an {\em $(n,d,\delta)$-cyclic LPHS for worst-case inputs in $X$} if each $h$ makes $d$ (adaptive) queries (of the form $x[i]$),
for every $x \in \hat X$, it holds that
	\[ \Pr_{h \in_R \calH} [ h(x) \neq h( x \ll 1) + 1 ] \le \delta. \]	
\end{definition}

Unlike the case of LPHS for random inputs, in the worst-case setting the choice of hash function must necessarily be randomized in order to achieve low per-input error.

A first observation is that one cannot hope to synchronize shifts of inputs $x$ that exhibit too much regularity: for example, the all-0 string. We formalize this requirement as {\em \good ness}, such that any cyclic shift of $x$ differs in at least an $\alpha$-fraction of positions from $x$. In what follows, $\Delta(x,y)$ denotes the Hamming distance of strings $x,y \in \zo^n$.


\begin{definition}[\good~inputs]  \label{def:good-input}
An input $x \in \zo^n$ is said to be {\em \good} if for every $i \in [n]$
it holds that $\Delta(x \ll i, x) \ge \alpha n$. We denote the set of all \good~inputs in $\zo^n$ by $\sfGood$.
\end{definition}

Note that $\sfGood$ is closed under shifts.  The above goodness requirement is necessary, in the sense that even if it is only violated by a single shift $0<i<n$, it requires a notable deterioration of the LPHS parameters. For example, if $x$ is a random string with period $n/2$ (i.e., $x = (x \ll n/2)$), then for any choice of hash function $h$ necessarily $h(x) = h(x \ll n/2)$, meaning one cannot achieve $\delta$ better than $2/n$.

\begin{remark}[Biased random inputs]
As an example application, inputs that occur as the result of \emph{biased} random sampling satisfy the above \good ness with high probability, for $\alpha$ that is a function of the bias.

Namely, consider the distribution $\cD_\beta$ over $\zo^n$ where each bit $x_i$ is selected as biased i.i.d.\ bits, equalling 1 with (constant) probability $0 < \beta < 1$. Then if $n$ is prime (or, more generally, without too many factors), then for every constant $\alpha < \min\{\beta,1-\beta\}$, it holds that
	\[ \Pr_{x \in_R \cD_\beta} [x \in \sfGood ] \ge 1 - n e^{-\Omega(n)}. \]
To see this, observe that since $n$ is prime, then for any given nonzero shift $0 < i < n$, it holds that $(1,1+i,1+2i,\dots,1+ni)$ forms a permutation of $[n]$. Consider the string $y := (x_1,x_{1+i},x_{1+2i},\dots,x_{1+ni})$, formed by appropriately permuting $x \in \zo^n$. Then $\Delta(x, x \ll i)$ is equal to the number of positions $j \in [n]$ for which $y_j \neq y_{j+1}$ (taking $y_{n+1} := y_1$). For each $j \in [n]$, conditioned on the values of $y_1,\dots,y_j$, the value of $y_{j+1}$ is randomly sampled with bias $\beta$. In particular, for any such prefix, the probability that $y_{j+1} \neq y_j$ is at least $\min\{\beta,1-\beta\}$. The claim thus follows by a Chernoff bound, together with a union bound over choices of $i \in [n]$.
\end{remark}

\medskip
Our main result is a black-box construction of worst-case LPHS for the set of inputs $\sfGood$, from any LPHS for random inputs, with small overhead (poly-logarithmic in the query complexity).

At a high level, the construction reduces to the case of random inputs, by (1) considering an intermediate input value $x' \in (\zo^\len)^n$ over a {\em larger} alphabet $\zo^\len$, whose symbols are formed from a random tiling of $\len \in \omega(\log n)$ bits of $x$, and then (2) converting $x'$ to a new input $y \in \zo^n$ using an $n$-wise independent hash $\gamma: \zo^\len \to \zo$. This approach is reminiscent of the generic alphabet-enlarging procedure used in Property~\ref{lem:smaller-alphabet}
 (and other prior works) via shingling, except with randomly chosen tile windows to accommodate worst-case inputs. As we will show, the \good ness of the original input $x$ will ensure with high probability over the choice of tiling that all tile-symbols of $x'$ are distinct elements of $\zo^\len$. Then by $n$-wise independence of $\gamma$, the resulting input $y \in \zo^n$ will be {\em uniform}. Note that if we begin with an LPHS which makes only $d$ queries to the (random) input, then we require only $d$-wise independence of $\gamma$.

\begin{proposition} [Worst-case cyclic LPHS for \good~inputs]  \label{prop:worst-case}
Assume there exists an $(n,d,\delta)$-cyclic LPHS for random inputs. Then, for every $0 < \alpha(n) \le 1$ and $\len \ge \omega(\log n)$, there exists a
$(n,d',\delta')$-cyclic LPHS for worst-case inputs in $\sfGood$, with $d' = O(d \cdot \len)$ and $\delta' = \delta+ n^2 (1-\alpha)^{\len}$.
\end{proposition}

Plugging in the LPHS construction from Theorem~\ref{thm:1dupper} (quoted from~\cite{DKK18}), together with $\len \in \omega(\log n) \cap \log^{O(1)}(n)$ yields the following corollary.
\begin{corollary}
\label{cor:worst}
For every constant $0 < \alpha \le 1$, there exists a
worst-case \emph{cyclic} $(n,d,\delta)$-LPHS for $\sfGood$ with $d = \tilde O(\sqrt n)$ and $\delta = O(1/n)$.
\end{corollary}
%

We now proceed to prove Proposition~\ref{prop:worst-case}.
\medskip

\begin{proof}
Let $\len = \len(n)$.
Consider the following hash function family $\calH = \{h_{S,\gamma}\}$, indexed by a subset $S \subset [n]$ of size $\len$, and a hash function $\gamma : \zo^\len \to \zo$ from an $n$-wise independent hash family. Sampling a hash function $h_{S,\gamma}$ from $\calH$ will consist of randomly selecting $S \subset [n]$ and sampling $\gamma$ from the $n$-wise independent hash family.

The worst-case LPHS will use as a black box an underlying $(n,d,\delta)$-LPHS $h^* : \zo^n \to \Z$ on {\em random} inputs.

\medskip \noindent
Evaluation of $h_{S,\gamma}(x)$, for $x \in \zo$:
\begin{enumerate}
\item For $i \in [n]$, denote $i+S = \{ i+s \mod n ~|~ s \in S \}$ and $x_{i+S} = (x_j)_{j \in i+S} \in \zo^\len$.
\item Execute the algorithm for average-case LPHS $h^*$. For each index $i \in [n]$ that $h^*$ wishes to query, perform the following:
  \begin{enumerate}
  \item Query $\len$ indices of $x$, corresponding to the set $(i+S) \subset [n]$. Denote the corresponding bit string by $x_{i+S} \in \zo^\len$.
  \item Let $y_i := \gamma(x_{i+S}) \in \zo$. Submit $y_i$ to $h^*$, as the answer to query $i \in [n]$.
  \end{enumerate}
  \item Let $z \in \Z$ denote the output of $h^*$. Output $z$.
\end{enumerate}

Note that the query complexity of $h_{S,\gamma}$ is precisely $\len$ times the query complexity of $h^*$. We now analyze the correctness of the resulting worst-case LPHS.

\begin{claim} \label{claim:S-tiles}
Let $x \in \sfGood$. With overwhelming probability over a random choice of subset $S \subset [n]$ of size $\len(n) \le n$, all $S$-tiles $x_{i+S}$ of $x$ are distinct.
Namely, 
	\[ \Pr_{h_{S,\gamma} \in_R \calH} [ \exists i\neq i' \in [n], x_{i+S} = x_{i'+S} ]  \le n^2 \cdot (1-\alpha)^{\len}. \]    
\end{claim}
\begin{proof}[of Claim~\ref{claim:S-tiles}]
Fix $i < i' \in [n]$. By \good ness of the input $x$, it holds that $\Delta(x \ll  i, x \ll i') = \Delta(x, x \ll (i'-i)) \ge \alpha n$; that is, there exists a subset $T \subseteq [n]$ of size $|T| \ge \alpha n$ for which $x_{i+j} \neq x_{i'+j}$ for every $j \in T$. Over the choice of $S$, $\Pr_S [x_{i+S} = x_{i'+S}] \le \Pr_S[S \cap T = \emptyset] \le (1-\alpha)^\len$.
 The claim thus holds by a union bound over pairs $i,i' \in [n]$.
\end{proof}

Conditioned on distinctness of all $S$-tiles $x_{i+S}$ queried by the algorithm $h^*$, then by $n$-wise independence of the hash family $\gamma$, it holds that the computed output bits $y_i := \gamma(x_{i+S})$ will be independently random. That is, conditioned on the above event, the hash function $h_{S,\gamma}$ will err with probability identical to that of $h^*$ on a random input. The proposition follows.

\end{proof}

\subsection{Non-Cyclic LPHS for Worst-Case Inputs}
We next present and achieve a notion of \emph{non-cyclic} LPHS for worst-case inputs, up to a maximum shift $R$. In this case, the alphabet-tiling procedure must be adjusted, as symbols outside the shift window will be lost. Instead, we consider a modified approach, in which the tiles for each index $i$ are chosen as a random subset from within a $W$-size window beginning at index $i$, for $W<n$. This makes for a slightly more complex ``goodness'' condition for worst-case input strings.


\begin{definition}[Worst-Case LPHS]
Let $X \subseteq \zo^n$ be a set of inputs. A family $\calH$ of hash functions $h: \zo^n \to \Z_n$ is a (non-cyclic) {\em $(n,d,\delta)$-LPHS for worst-case inputs in $X$}, up to shift bound $R$, if each $h$ makes $d$ (adaptive) queries (of the form $x[i]$), and for every $x \in X$, and every $1 \le r \le R$ it holds that
	\[ \Pr_{h \in_R \calH} [ h(x) \neq h( x \lll r) + r ] \le \delta. \]	
\end{definition}



We formalize the desired input non-regularity requirement via {\em \goodw ness}, parameterized by a ``window size'' $W \in [n]$ and difference parameter $0 <\alpha \le 1$, such that
each of the length-$W$ substrings $x$ differ pairwise in at least $\alpha$ fraction of their symbols.
In what follows, we denote by $x_{i+[W]}$ the $W$-substring of $x$ beginning at index $i$ (which may have length less than $W$ if $(i+W) < n$), and by $\Delta(x',y')$ the Hamming distance of strings $x',y' \in \zo^\ell$ (where $1 \le \ell \le W$).

\begin{definition}[\goodw~inputs]  \label{def:good-input}
An input $x \in \zo^n$ is said to be {\em \goodw} if for
every $i < i' \in \{1,\dots,n-W/2\}$ it holds that $\Delta(x_{i+[W]},x_{i'+[W]}) \ge \alpha |x_{i'+[W]}|$.\footnote{Note for $i \le n-W$, the length $|x_{i+[W]}| = W$, but for $(n-W) < i' \le (n-W/2)$, then $W/2 \le |x_{i'+[W]}| = (n-i') < W$. That is, we consider also substring windows that ``hang over'' the edge of the string $x$ up to $W/2$.}
We denote the set of all \goodw~inputs in $\zo^n$ by $\sfGoodw$.
\end{definition}

\begin{remark}[Biased random inputs]
As with the cyclic case, we similarly have that biased per-bit random inputs satisfy \goodw ness with high probability for sufficiently large $W$ and $\alpha$ that is a function of the bias.

Namely, consider the same distribution $\cD_\beta$ over $\zo^n$ where each bit $x_i$ is selected as biased i.i.d.\ bits, equalling 1 with (constant) probability $0 < \beta < 1$. Then for every constant $\alpha < \min\{\beta,1-\beta\}$, it holds that
	\[ \Pr_{x \in_R \cD_\beta} [x \in \sfGoodw ] \ge 1 - n^2e^{-\Omega(W)}. \]
(Note here that $n$ need not be prime.)
The argument here follows similarly to the cyclic case, except without the complication of cyclic wraparound.

Consider a fixed pair $i < j \in [n-W/2]$. We analyze $\Delta(x_{i+[W]},x_{j+[W]}) = | \{\ell \in [W] : x_{i+\ell} \neq x_{j+\ell} \}|$.
We may again reorder the elements of $x$, this time into a collection of $i' :=j-i$ sequences, the corresponding shifted cosets of $(x_i,x_{i+i'}x_{i+2i'},\cdots)$ starting with the corresponding $(j-i)$th value $x_\ell$, $\ell \in \{i,\dots,j-1\}$, and containing all $i'$-multiple instances, up to $j+W$.
For any such sequence, and any fixed choice of values in the prefix of the sequence, the probability that the following term is equal to the previous is at least $\min\{\beta,1-\beta\}$. Altogether, all indices in the range $\{i,\dots,i+W\} \cup \{j,\dots,j+W\}$ thus contribute a fresh term, aside from the initial $(j-i)$ that served as start points of the sequences. The claim thus follows by a Chernoff bound, together with union bound over pairs $i,j$. Note that for $W \in \omega(\log n)$, the resulting error probability is negligible.
\end{remark}




As with the cyclic case, we demonstrate an analogous black-box construction of worst-case \emph{non-cyclic} LPHS for the set of inputs $\sfGoodw$, from any \emph{non-cyclic} LPHS for random inputs, with small overhead. We state the following proposition in terms of an arbitrary window-size parameter $W$, and then discuss relevant settings afterward.



\begin{proposition} [Worst-Case LPHS For \Goodw~Inputs]  \label{prop:worst-case2}
Assume there exists an $(n,d,\delta)$-non-cyclic LPHS with shift bound $R$ for \emph{random} inputs. Then, for every $0 < \alpha(n) \le 1$, $1 <W(n) \le n$, and $\len \ge \omega(\log n)$, there exists an
$(n',d',\delta')$-non-cyclic LPHS with shift bound $R$ for \emph{worst-case} inputs in $\sfGoodwp$, with $n'=n+W$, $d' = O(d \cdot \len)$ and $\delta' = \delta+ e^{-\Omega(b)}$. 
\end{proposition}

\begin{proof} 
Let $\len = \len(n)$, $W=W(n)$.
Consider the hash function family $\calH = \{h_{S,\gamma}\}$ as defined within the proof of Proposition~\ref{prop:worst-case}, indexed by a subset $S \subset [W]$ of size $\len$, and a hash function $\gamma : \zo^\len \to \zo$ from an $n$-wise independent hash family. Here, we have the following two differences from the cyclic case:
  \begin{itemize}
  \item Sampling a hash function $h_{S,\gamma}$ from $\calH$ will consist of randomly sampling $\gamma$ from the $n$-wise independent hash family (as before), but now sampling the subset $S \subset [W]$ instead of $[n]$.
  \item The worst-case LPHS will use as a black box an underlying $(n,d,\delta)$-\emph{non-cyclic} LPHS $h^* : \zo^n \to \Z$ on random inputs (as opposed to cyclic LPHS).
  \end{itemize}


\noindent
Note that (as before) the query complexity of $h_{S,\gamma}$ is precisely $\len$ times the query complexity of $h^*$. We now analyze the correctness of the resulting worst-case LPHS.

First, note that for any $x \in \sfGoodwp$ and $1 \le r \le R$, then with high probability over $y \in_R \zo^r$, it holds that the concatenated string $z := x || y \in \sfGoodwpp$. Indeed, consider $\Delta(z_{i+[W]},z_{j+[W]})$ for various choices of $i,j \in [n'+r-W/2]$. From the definition of $\sfGoodwp$, the required relative distance $\alpha/2$ holds for pairs with $i,j \in [n'-W/2]$. For any $i \in [n'+r-W/2]$ and $j \in \{(n'-W/2+1),\dots,(n'+r-W/2)\}$, then the string $z_j$ is composed of at least $W/2$ \emph{uniform} bits (coming from $y$); the required relative distance $\alpha/2$ thus holds except with probability upper bounded by $e^{-\Omega(W)}$, by the (biased) random inputs Remark above.

The proposition then follows from the following claim, which collectively implies that our transformation reduces (with small error) directly to the underlying average-case LPHS: (1) First, that aside from probability $n^2(1-\alpha)^b$ over $h_{S,\gamma}$, queries into a ``good'' input in $\zo^{n+W}$ get mapped to queries into a \emph{uniform} input in $\zo^n$; (2) Second, that the transformation preserves the (non-cyclic) shift metric.

\begin{claim} For given $S,\gamma$, define  $f_{S,\gamma} : \zo^{n+W} \to \zo^n$ via $f_{S,\gamma}(x) = (\gamma(x_{i+S}))_{i \in [n]}$. For every $x \in \sfGoodwp$, the following hold:
  \begin{enumerate}
  \item With high probability over a random choice of subset $S \subset [W]$ of size $\len(n) \le n$, all ``$S$-tiles'' $x_{i+S}$ of $x$ are distinct. Namely,
	\[ \Pr_{S \in_R {W \choose b}} [ \exists i\neq i' \in [n], x_{i+S} = x_{i'+S} ]  \le n^2 \cdot (1-\alpha)^{\len}. \] 	
As a consequence, then with probability $1-n^2(1-\alpha)^{\len}$ over the choice of $S$, it holds that
	$\left \{ f_{S,\gamma}(x) \right \}_\gamma \equiv U_{\zo^n}$,
where the distribution $\left \{ f_{S,\gamma}(x) \right \}_\gamma$ is over the choice of $\gamma$, and $U_{\zo^n}$ denotes the uniform distribution over $\zo^n$.
	
  \item The map $f_{S,\gamma}$ preserves the (non-cyclic) shift metric. That is, for every $1 \le r \le R$,
    \begin{itemize}
    \item Agreement window: For every $S,\gamma$, and every index $(1+r) \le i \le (n-r)$, it holds that $f_{S,\gamma}(x)_i = f_{S,\gamma}(x \lll r)_{i+r}$ with probability 1 (over randomness of $\lll$).
    \item Outside agreement window:
  	\[ \left \{  f_{S,\gamma}(x \lll r) \right \}_\gamma \equiv  \left \{ \Big(f_{S,\gamma}(x)_{[1+r,\dots,n]}, U_{\zo^r} \Big)\right \}_\gamma. \]
    \end{itemize}

  \end{enumerate}
\end{claim}

\begin{proof}
We address each part of the claim.
\begin{enumerate}
\item Fix $i < i' \in [n]$. By \goodw ness of the input $x$, it holds that $\Delta(x_{i+[W]},x_{i'+[W]}) \ge \alpha n$; that is, there exists a subset $T \subseteq [W]$ of size $|T| \ge \alpha n$ for which $x_{i+j} \neq x_{i'+j}$ for every $j \in T$. Over the choice of $S$, $\Pr_S [x_{i+S} = x_{i'+S}] \le \Pr_S[S \cap T = \emptyset] \le (1-\alpha)^\len$.
 The claim thus holds by a union bound over pairs $i,i' \in [n]$.

\item Note that $x \in \zo^{n+W}$ and $x \lll r$ satisfy the property $x_i = (x \lll r)_{i+r}$ for all $i \in [n+W-r]$. Since $S \subset [W]$, then in particular $x_{i+S} = (x \lll r)_{i+r+S}$ for every $i \in [n-r]$. This implies the desired equality within the agreement window.

Outside the agreement window, the required property holds directly by applying the Claim part (1) to the appended string $z := x || y$ for $y \in_R \zo^r$, which can be done since (as argued above), this string satisfies $z \in \sfGoodwpp$.


\end{enumerate}
\end{proof}




\end{proof}

Plugging in the non-cyclic LPHS construction from Theorem~\ref{thm:1dupper}, together with $\len \in \omega(\log n) \cap \log^{O(1)}(n)$ yields the following corollary.
\begin{corollary}
For every constant $0 < \alpha \le 1$ and $W \in \omega(\log n)$, there exists a
worst-case \emph{non-cyclic} $(n,d,\delta)$-LPHS for $\sfGoodw$ with $d = \tilde O(\sqrt n)$ and $\delta = O(1/n)$.
\end{corollary}
%

Let $d_S(x,y)$ denote non-cyclic shift distance, defined as the distance between $x$ and $y$ on the De Bruijn Graph. Then the above results imply the following (probabilistic, bounded-distance) isometric embedding of $d_S$ to the line.

\begin{corollary}
Let $0 < \alpha \le 1$, $R \in [n]$ be a shift distance bound, $W \in \omega(\log n)$, and $d<n^{1/2}$. Then, there exists a family of hash functions $\calH$ such that
  \begin{enumerate}
  \item Each $h \in \calH$ makes $d$ queries to the input;
  \item For all $x,y \in \sfGoodw$ with $d_S(x,y)=r \le R$, we have $\Pr[|h(x)-h(y)| \neq r] \le \tilde O(r/d^2)$.
  \end{enumerate}
\end{corollary}

%% file: Applications.tex
\section{Applications}
\label{sec:applications}

\newcommand{\xb}{x^{(b)}}
\newcommand{\xbk}{x^{(k,b)}}
\newcommand{\Xb}{X^{(b)}}
\newcommand{\Gw}{{\mathrm{G,w}}}

\newcommand{\cP}{\mathcal{P}}
\newcommand{\cI}{\mathcal{I}}
\newcommand{\HSS}{{\sf HSS}}
\newcommand{\pk}{{\sf pk}}
\newcommand{\ek}{{\sf ek}}
\newcommand{\ct}{{\sf ct}}
\newcommand{\fst}{^{(1)}}
\newcommand{\rhoth}{^{(\rho)}}
\newcommand{\jth}{^{(j)}}
\newcommand{\modout}{m}
\newcommand{\secpar}{\lambda}
\newcommand{\poly}{{\sf poly}}
\newcommand{\state}{{\sf state}}
\newcommand{\advA}{\mathcal{A}}
\newcommand{\share}{{\sf share}}

\newcommand{\ElGamal}{{\sf ElGamal}}
\newcommand{\Vector}{{\sf Vector}}
\newcommand{\Scalar}{{\sf Scalar}}
\newcommand{\bbH}{\mathbb{H}}
\newcommand{\bbG}{\mathbb{G}}
\newcommand{\Pair}{{\sf Pair}}
\newcommand{\kDDL}{{k\text{-}{\sf DDL}}}
\newcommand{\PRF}{{\sf PRF}}
\newcommand{\Convert}{{\sf Convert}}
\newcommand{\blink}{{{\sf blin}\text{-}k}}	
\newcommand{\blintwo}{{{\sf blin}\text{-}2}}


In this section we present several cryptographic and algorithmic applications that motivate the different LPHS flavors studied in this work. In all applications, there are two or more parties that have partially overlapping views of a large object, and the goal is to measure in {\em sublinear time} the relative misalignment between the views, with low failure probability and without direct interaction between the parties.
Before discussing the applications in detail, we give a taxonomy of the kinds of LPHS instances on which they depend.
\begin{itemize}
\item {\bf 1-dimensional vs.\ $k$-dimensional.} The first application, to packed homomorphic secret sharing, requires $k$-dimensional LPHS for $k\ge 2$. The other applications can apply in any dimension, but the 2-dimensional variant seems most useful in the context of natural use cases that involve 2-dimensional objects (such as digital images).
\item {\bf Cyclic vs.\ non-cyclic.} While some of the applications can also meaningful in the cyclic case, the non-cyclic one is needed to capture the ``partially overlapping views'' scenario.
\item {\bf Random vs.\ worst-case.} In the application to packed homomorphic secret sharing from a generic group, the object is a huge and locally random mathematical object. Hence the default random-input variant of LPHS suffices. In other applications, which may involve physical objects or digital documents, the worst-case variant is needed. While the latter must inevitably exclude objects that are close to being highly periodic, this is not an issue for most natural use cases.
\item {\bf Small shift vs.\ big shift.} When the misalignment is small, one could apply a MinHash-based $(d,\tilde O(1/d))$-LPHS and reduce its failure probability via repetition. This solution is not suitable for applications that depend on the simple metric structure of the LPHS output, and in any case leads to inferior communication rate (for the same failure probability and running time) compared to applying a single instance of a $(d,\tilde O(1/d^2))$-LPHS. In applications that require detecting an arbitrary misalignment in sublinear time, a $(d,\tilde O(1/d))$-LPHS  does not suffice at all since the failure probability is multiplied by the shift amount.
\end{itemize}

We proceed with the details of the applications, starting with cryptographic applications.

\subsection{Packed Homomorphic Secret Sharing}
\label{sec-applicationhss}

In this section, we present a cryptographic application of \emph{$k$-dimensional} LPHS for constructing a ``packed'' version of  homomorphic secret sharing (HSS) from cryptographically hard groups. Roughly speaking, $k$-packed HSS can share a vector of $k$ values at the same communication cost as sharing a single value. This improved communication comes at the expense of a computational overhead that we optimize using two distinct LPHS-based approaches.
Our main approach relies on $k$-dimensional LPHS. Since even for $k=2$ we do not have a provable optimal construction of $k$-dimensional LPHS, we also present an alternative approach that relies on 1-dimensional LPHS and performs better in some parameter regimes.




\medskip

The remainder of this section is composed as follows. In Section~\ref{sec:kDDL} we introduce a $k$-dimensional variant of the DDL problem and in Section~\ref{sec:klphs2kddl} we show how to realize it from $k$-dimensional non-cyclic LPHS. Then, in Section~\ref{sec:kHSS}, we present the application of $k$-dimensional DDL to reducing the communication complexity of group-based HSS, as well as an alternative construction based on 1-dimensional embedding.

%

\subsubsection{Multidimensional DDL}   \label{sec:kDDL}


In this section we describe a $k$-dimensional generalization of DDL ($k$D-DDL for short) that we will use as an intermediate step towards constructing packed HSS. Since DDL has already found applications beyond the context of HSS~\cite{DottlingGIMMO19,trapdoor2,trapdoor3}, one may expect the same for the $k$-dimensional variant.

While perhaps the most direct extrapolation of DDL to $k$ dimensions would correspond to a DDL challenge within $k$ \emph{independent} copies of the group (i.e., $\Z_n^k \hookrightarrow \bigoplus_{i=1}^k \G$), for our purposes, we will consider a notion where the $k$ dimensions are embedded within a {\em single} group $\G$ of size $N\gg n$. More explicitly, we will embed $\Z_n^k \hookrightarrow \G$ via  $(j_1,\dots,j_{k}) \mapsto \prod_{i=1}^{k} g_i^{j_i}$ for $k$ randomly selected generators $g_i$ of $\G$.
As this embedding is not injective, its use within a construction of $k$D-DDL from  $k$D-LPHS will introduce an extra collision error probability; however, when $N$ is much bigger than the number of queries $d$ made by the $k$D-LPHS algorithm (as will be the case for the HSS application), this collision error probability will be negligible.

As in the case of (1-dimemsional) DDL algorithms defined in Section~\ref{sec:DDL}, we consider generic algorithms, i.e.\ operating within the generic group model.
We present a version of this model that captures the $k$-dimensional case below.

\medskip

Recall that in the 1-dimensional case, the LPHS parameter $n$ was the same as the group size and the LPHS string $x^{(b)}$ represented the sequence of all group elements. Since this will no longer be the case in the $k$-dimensional variant, from here on we denote the group size by $N$, and the sequence of labels of group elements by $\Xb$. In the 1-dimensional case, we considered an experiment where a string $\Xb \in \Sigma_b^N$ (representing the group) and a value $v \in \Z_N$ are sampled uniformly; the generic DDL algorithm makes adaptive queries of the form $(\alpha,\beta) \in \Z_N \times \Z_N$, which are answered by $\Xb[\ell_v(\alpha,\beta)] \in \Sigma_b$, where $\ell_v(\alpha,\beta) := \alpha \cdot v + \beta \in \Z_N$, and outputs a value $\gamma \in \Z_N$. The goal of the 1D-DDL algorithm was that the outputs $\gamma,\gamma'$ on implicit inputs $v$ and $v'=v+1$ satisfied $\gamma - \gamma' = 1$.


In the $k$-dimensional experiment that we define, again, a string $\Xb \in \Sigma_b^{\Z_N}$, representing a single group of order $N$ with generator $g$, and a value $v \in \Z_N$.
are sampled uniformly. In addition, $k$ public ``basis elements'' $w_1,\dots,w_{k} \gets \Z_N$ are randomly selected, representing the (discrete logarithm of) $k$ group generators $g_i=g^{w_i}$. We assume $N$ to be prime, in which case the $g_i$ are distinct group generators with high probability.
%
Given this setup, the generic algorithm $A^{\mathrm{G}}$ makes adaptive queries of the form $(\alpha,(\beta_1,\ldots,\beta_k)) \in \Z_N \times \Z_N^k$, which are answered by $\Xb[ \alpha v + \sum_{i = 1}^{k} \beta_i w_i ]$, namely the string handle of the group element $(g^v)^\alpha\cdot \prod_{i=1}^{k} g_i^{\beta_i}$.

The output of the algorithm is an integer vector $\gamma \in \Z^k$.
Here we use the notation $A^\Gw(\Xb,v)$ to indicate that $A$ is a generic algorithm that does not have direct access to $\Xb$ and $v$, but can have full knowledge of the public basis $w=(w_1,\ldots,w_k)$.

The goal of the $k$D-DDL algorithm is to similarly detect a shift of 1 on the input in any one of the $k$ dimensions. In this case, such a shift in dimension $i\in[k]$ corresponds to an additive offset of $w_i$ in the discrete logarithm. Concretely, we would like that the outputs $\gamma,\gamma'$ on implicit inputs $v$ and $v'=v+w_i$ to satisfy $\gamma - \gamma' = e_i$, the $i$th unit vector in $\Z^k$.

Note that the error probability of a $k$D-DDL algorithm can come from multiple sources: its internal coins, the choice of the implicit input $v$, and the ``setup'' process of choosing group labels $\Xb$ and the basis $w$. While for the purpose of driving down the error via repetition it is  useful to separate the first source of error from the others (see more below), we will consider the probability space that combines all these random choices for simplicity.

\begin{definition}[$k$D-DDL]  \label{def:kDDL}
A generic $k$D-DDL algorithm $A$ is an $(N,b,d,\delta)$ $k$D-DDLA if
$A$ makes $d$ (adaptive) queries and for every $i \in [k]$,
	\[ \pr 
[A^{\Gw}(\Xb, v) - A^{{\Gw}}(\Xb, v + w_i) \neq e_i] \leq \delta, \]
where $A^\mathrm{G}$ denotes generic access as described above and the probability is over the choice of $\Xb$ from $\Sigma_b^{\Z_N}$ and $v, w_1,\ldots,w_k$ from $\Z_N$.
\end{definition}

\begin{remark}[Bigger $k$-dimensional shifts] Analogously to Lemma~\ref{lem:shift}, we can use a union bound to deduce
a bound on the error probability for bigger $k$-dimensional shifts. Concretely, for any $(N,b,d,\delta)$ $k$D-DDLA and $k$-dimensional shift $u=(u_1,\ldots,u_k)\in\Z^k$ we have
	\[ \pr
	\left[ A^{\Gw}(\Xb, v) -
		 A^{\Gw}\Big(\Xb, v + \sum_{i=1}^{k} u_iw_i \Big)
		 \neq u \right] \leq \delta\cdot  |u|_1, \]
where $|u|_1$ denotes the $\ell_1$-norm of $u$ and the probability space is as in Definition~\ref{def:kDDL}. In the context of the packed HSS application, this will imply an error that scales with the $\ell_1$-norm of the output vector.
\end{remark}
As DDL algorithms have a small but non-negligible error, it may be useful to drive the error probability down using independent repetitions, where the internal coins of the algorithm are picked independently but all other sources of randomness are fixed.  (This repetition comes at a price of additional communication, since each instance produces different outputs.) For this purpose, one can use a more refined version of Definition~\ref{def:kDDL} in which the probability space is only over the internal randomness of $A^{\Gw}$ and the probability bound $\delta$ should hold except with negligible probability over all other random choices. The construction we present next indeed satisfies this property.

\subsubsection{$k$D-DDL from $k$D-LPHS} \label{sec:klphs2kddl}

Our main solution to the $k$D-DDL problem is based on the {\em non-cyclic} variant of $k$-dimensional LPHS from Definition~\ref{def-lphs} ($k$D-LPHS for short) that we recall below.

Let $e_i$ denote the $i$th unit vector of length $k$. For a $k$-dimensional string $x=x^{(k,b)}$, we denote by $x \lll e_i$ a non-cyclic shift of $x$ by 1 in the $i$th dimension (e.g., for $k=2$ and $i=1$, this corresponds to chopping the top row and adding a random row on the bottom). In the following we will use $\Z_n$ to denote the set of integers $\{0,1,\ldots,n-1\}$ and will not use its group structure.

\begin{definition}[Non-cyclic $k$D-LPHS]
Let $h:\Sigma_b^{\Z_n^k} \to \Z^k$ be a function. We say that $h$ is a non-cyclic {\em $(n,b,d,\delta)$ $k$D-LPHS} if $h$ can be computed by making $d$ {\em adaptive} queries (of the form $x[\beta_1,\dots,\beta_k]$ for $\beta_i \in \Z_n$) to a $k$-dimensional input  $x=x^{(k,b)}\in\Sigma_b^{\Z_n^k}$ and for every $i \in [k]$:
	\[ \pr_{x \in_R \Sigma_b^{\Z_n^k}} \lbs h(x) \neq h(x \lll e_i) + e_i \rbs \leq \delta.\]
This naturally extends to a probabilistic $h\in_R\cal H$.
\end{definition}

Given any non-cyclic $k$D-LPHS, we can construct a $k$D-DDL algorithm with similar parameters. The transformation will not be perfect, but will induce collision probability error that for sufficiently large group sizes $N$ will be negligible.

\begin{theorem}[From non-cyclic $k$D-LPHS to $k$D-DDL] \label{thm-kddl}
Let $k,n$ be positive integers and $N$ be a prime such that $N\ge n$. There exists a black-box reduction that converts any $(n,b,d,\delta)$  {\em non-cyclic} $k$D-LPHS to an $(N,b,d,\delta+d^2/N)$ $k$D-DDL algorithm. 
The DDLA is {\em query-restricted} in the sense that the values $\alpha,\beta_1,\ldots,\beta_k$ for each generic query satisfy $\alpha=1$ and $\beta_i\in\Z_n$.
\end{theorem}
\begin{proof}
Given black-box access to an $(n,b,d,\delta)$ $k$-dimensional non-cyclic LPHS denoted by $h$, we construct the $(N,b,d,\delta+ d^2/N)$ $k$-dimensional DDLA, denoted by $A$, as follows.
The algorithm $A$ simulates $h$ by responding to each query $\beta=(\beta_1,\dots,\beta_k) \in \Z_n^k$ that $h$ makes into $x=x^{(k,b)}$ with the group symbol read by $A^\Gw[\Xb,v]$ on the corresponding generic group query $(\alpha=1,\beta)$. Once the $d$ queries are completed, $A$ outputs the output of $h$.

We now bound the error probability of $A$ over the choice of the generators basis $w=(w_1,\dots,w_{k}) \in_R \Z^k_N$, challenge value $v\in_R\Z_n$, and the randomness of $h$ and $\Xb$.
For $\beta\in\Z_n^k$ (viewed as a vector of integers), let $\langle w,\beta\rangle$ denote the mod-$N$ inner product of $w$ and $\beta$. Consider the executions of $A$ on $v$ and $v+w_i$. The key observation is that conditioned on the good event that neither of the two executions contains a colliding set of queries $\beta,\beta'$ such that $\langle w,\beta\rangle=\langle w,\beta'\rangle$,  the joint distribution of the symbols read by $A^\Gw[\Xb,v]$ and $A^\Gw[\Xb,v+w_i]$ is identical to that of the symbols read by $h(x)$ and $h(x \lll e_i)$. It follows that conditioned on this good event, the error probability of $A$ is bounded by $\delta$. It thus suffices to show that the probability of the bad collision event is bounded by $d^2/N$.

For any fixed pair of distinct query vectors $\beta,\beta'\in\Z_n^k$, we have
\[\pr_{w\in\Z_N^k}[\langle w,\beta\rangle=\langle w,\beta'\rangle]=1/N\]
(here we use the assumptions that $N$ is prime and $N\ge n$).
For each non-colliding query made by $A$, the value it receives back is a freshly sampled uniform value $\Xb[\ell] \in\Sigma_b$ (for some non-previously-queried index $\ell \in \Z_N$); in particular, this value is independent of the choice of $w$. Thus, taking a union a bound over the ${d \choose 2}$ pairs of distinct queries made in an execution of $A$, we get a collision probability bound of ${d \choose 2}/N$ in a single execution, and at most $2\cdot{d \choose 2}/N< d^2/N$ collision probability on either the execution on $v$ or on $v+w_i$.
\end{proof}

\medskip
For example, plugging in the results from Section~\ref{sec:2dim}, we get a provable 2D-DDL algorithm that makes $d$ queries (that can be implemented using $\tilde O(d)$ group multiplications) and has  $\delta = \tilde O(d^{-7/8} + d^2/N)$ error probability. Using the conjectured optimal algorithm, the $d^{-7/8}$ term can be replaced by $d^{-1}$.

The additive error term of $d^2/N$ can be further reduced by applying the worst-case notion of LPHS from Appendix~\ref{sec-worstcase} to get robustness against collisions. However, this term is already negligible for a typical choice of parameters. For instance, practical cryptographically hard groups have order $N>2^{256}$, and so for this term to be significant the running time $d$ should be close to $2^{128}$.


%

\subsubsection{From  $k$D-DDL to Packed HSS}  \label{sec:kHSS}

We now demonstrate applications of $k$D-DDL to packed HSS. More formally, we consider homomorphic secret sharing (HSS) with non-negligible error $\delta$, as introduced in~\cite{BGI16}:

\begin{definition}[$\delta$-Homomorphic Secret Sharing]
\label{def:HSS}
A (2-party) {\em $\delta$-Homomorphic Secret Sharing ($\delta$-HSS)} scheme for a class of programs $\cP$ over $\Z$ with input space $\cI\subseteq \Z$ consists of PPT algorithms $(\HSS.\HShare,$ $\HSS.\HEval)$ with the following syntax:
  \begin{itemize}
  \item $\HSS.\HShare(1^\secpar,x)$:   Given security parameter $1^\lambda$ and secret input value $\alpha\in\cI$, the sharing algorithm outputs secret shares $(\share_0,\share_1)$.
   \item $\HSS.\HEval(i, (\share\fst_i,\dots,\share\rhoth_i), P, \modout)$: Given party index $i \in \{0,1\}$, the $i$th secret share for $\rho$ inputs, program $P \in \cP$ with $\rho$ input values, and integer $\modout\geq 2$, homomorphic evaluation outputs $y_i \in \Z_\modout$, constituting party $i$'s share of an output $y\in \Z_\modout$.
  \end{itemize}
\end{definition}
\noindent
The algorithms $(\HSS.\HShare,\HSS.\HEval)$ must satisfy the expected homomorphic evaluation correctness $y_0+y_1 = P(\alpha\fst,\dots,\alpha\rhoth) \in \Z_\modout$ for any set of inputs and program $P \in \cP$, with probability at least $(1-\delta)$ over the execution of $\HSS.\HShare$.
In addition, $\HSS.\HShare$ should satisfy semantic security (i.e., given only the $i$th share of a sequence of inputs, a polynomial-time bounded adversary cannot distinguish which of two input sequences they were derived from).

\medskip

We will be interested in $\delta$-HSS for vector-scalar multiplications. That is, let $\cP_{\blink}$ denote the class of programs which perform bilinear operations over two variable types: \emph{vectors} over $[M]^k$, and \emph{scalars} over $[M]$. We will maintain notation of scalars denoted by Greek letters and vectors as lowercase roman letters: e.g., $\alpha \in [M], u = (u_1,\dots,u_k) \in [M]^k$.

More formally, we consider $\cP_{\blink}$ the class of programs which act on inputs each of type vector or scalar, and which makes polynomially many of the following operations. Note that all data types are additionally either ``Input'' or ``Non-Input''; to capture the bilinear limitation, multiplication is only allowed between Input scalars and vectors.
 \begin{itemize}
 \item {\sf Add Input (or Non-Input) Vectors}: $u'' \gets u + u'$ (or $z'' \gets z + z'$).
  \item {\sf Add (Input or Non-Input) Scalars}: $\alpha'' \gets \alpha + \alpha'$.
 \item {\sf Multiply Input Scalar and Input Vector}: $z \gets \alpha \cdot u$.
 \item {\sf Parse Non-Input Vector as Non-Input Scalars}: $(z_1,\dots,z_k) \gets z$
 \end{itemize}
For example, this class of programs includes many useful statistical computations, such as correlations.

\paragraph{Overview of existing group-based HSS.}
At a very high level, the group-based HSS constructions of~\cite{BGI16} and successors work as follows. Let $\G$ be a (cyclic) cryptographically hard group\footnote{For which finding discrete logarithms is computationally hard. More specifically, we require groups for which the ``Decisional Diffie-Hellman'' (DDH) assumption~\cite{DiffieH76} holds: i.e., $(g,g^a,g^b,g^{ab})$ is computationally indistinguishable from $(g,g^a,g^b,g^c)$ for random generator $g$ and random exponents $a,b,c$.} of prime order $N$, and randomly selected generator $g$. Let $c \in_R \Z_N$ be a random secret key for ElGamal encryption (as described below). Consider a message space $\cI \subset \Z^+$ (where correctness error will scale with the magnitude of the inputs and partial computation values). We will maintain these notations for the remainder of the section.


Secret data is encoded by the HSS in one of two types:
  \begin{itemize}
  \item Encryptions: $\alpha \in \Z$ encoded by an ElGamal ciphertext $[\alpha] := (g^r, g^{cr+\alpha}) \in \G^2$, for random $r \in_R \Z_N$. Both parties receive the ciphertext $[\alpha]$.
  \item Additive Shares: $\alpha' \in \Z$ encoded as two sets of additive secret shares $\langle \alpha' \rangle$ and $\langle c \alpha' \rangle$ over $\Z_N$, where $c \in \Z_N$ is the ElGamal secret key. (Notationally, $\langle \alpha \rangle$ denotes that each party receives a random share in $\Z_N$ subject to sum (over $\Z_N$) equaling $\alpha$.)
  \end{itemize}
Homomorphic evaluation takes place via a sequence of \emph{addition} steps, performable directly on data items encoded in the same type, and \emph{restricted multiplications}, in which a value in Encrypted type can be multiplied by a value in Additive Share type, as follows:
  \begin{itemize}
  \item Pairing: Via linear operation in the exponent using $\langle \alpha' \rangle$ and $\langle c\alpha' \rangle$ together with $[\alpha]$, the parties locally compute $g^{\langle \alpha\alpha' \rangle}$: that is, the parties hold group elements $g^\beta$ and $g^{\beta+\alpha\alpha'}$, for some exponent $\beta \in \Z_N$.

  This can be viewed as a form of ``distributed decryption,'' leveraging that decryption of a ciphertext $(g^\gamma,g^\zeta)$ via $(g^\zeta) \cdot (g^\gamma)^{-c}$ induces an operation $\zeta-c\gamma$ in the exponent space that is linear in $c$. Performing the operation on an identical ciphertext (i.e., fixed $\gamma, \zeta$) and additive shares of $c$ thus yields the desired result.

  \item Share Conversion: Each party executes the (1D) DDL on his resulting group element.
  \end{itemize}
If the DDL algorithm succeeds, then the parties result in additive shares $\langle \alpha\alpha' \rangle$ of the difference in the discrete logarithms of the two input elements $g^\beta, g^{\beta+\alpha\alpha'}$. In this case, namely, additive shares of the product $\alpha\alpha'$ over $\Z_N$.

Ultimately, the existing 1-dimensional HSS scheme performs multiplication of scalars $\zeta=\alpha\alpha' \in [M]$ and, given runtime $T$, succeeds except with error probability $\zeta/T^2$ (inherited from the 1D-DDL~\cite{DKK18}).


\medskip

We next proceed to describe three solution approaches for supporting HSS for bilinear functions $\cP_\blink$: (1) A non-packed baseline application of the HSS described above, where a $k$-dimensional plaintext vector $u$ is simply encoded as $k$ scalars; (2) Our new packed HSS solution \emph{from $k$D-DDL}, which packs a vector $u$ into a single ciphertext via a product of randomly selected generators, and leverages $k$-dimensional DDL; and (3) An alternative \emph{1-dimensional embedding} solution approach, which encodes $u$ into a scalar value and executes standard 1D-DDL.

Given an optimal $k$D-DDL (with error $d^{-2/k}$) then the $k$D-DDL packed HSS solution would dominate the 1-dimensional embedding approach; however, given the current gap, the two solutions are presently incomparable. In particular, the $k$D-DDL solution wins out when the payload magnitude and desired error probability are not fully known \emph{in advance}.

\paragraph{Baseline solution: Non-packed.}
For baseline comparison, we consider the existing solution for obtaining group-based HSS for $\cP_{\blink}$: \emph{Direct (non-packed) application} of~\cite{BGI16,BGI17,DKK18}, which encodes each component of a vector $u \in [M]^k$ as a separate element

In the direct application, each vector $u = (u_1,\dots,u_k) \in [M]^k$ is simply encoded as a collection of $k$ ElGamal ciphertexts, and scalar-vector multiplications $\alpha \cdot u$ is homomorphically evaluated via $k$ independent multiplications $\alpha \cdot v_i$, for $i =[k]$. The share size to encode a vector thus increases to $k$ ElGamal ciphertexts. Each vector-scalar multiplication corresponds to $k$ independent instances of a standard scalar-scalar multiplication; in particular, for output value $z = \alpha \cdot u$ and runtime $T$ for the share conversion procedure (namely, 1D-DDL), the error grows as $\sum_{i=1}^{k} (z_i/T^2)$.

\paragraph{New solution: From $k$D-DDL.}
In our new solution \emph{from $k$D-DDL}, we embed the vector $u$ into a single ElGamal ciphertext as a corresponding \emph{product of generators}, and run the $k$D-DDL algorithm to extract the corresponding additive shares of the product.

More concretely, let $\G$ be a cyclic DDH-hard group of prime order $N$, and let $g,g_i = g^{w_i}$, be $k+1$ randomly selected group generators. (Recall we assume $N$ to be prime, in which case $g^{w_i}$ for randomly chosen $w_i$ will be a generator with high probability.) Given a secret vector $u = (u_1,\dots,u_k) \in [M]^k$, we embed $u$ into a single ciphertext as follows. Recall a standard ElGamal ciphertext with secret key $c \in \Z_N$ encodes scalar plaintext $\alpha \in [M]$ as a pair $(g^r,g^{rc}\cdot g^\alpha) \in \G^2$ (where $g \in \G$ is a generator and $r \in_R \Z_N$ is encryption randomness). We now encode a \emph{vector} $u \in [M]^k$ into a pair of group elements using the generators $g_i$ by sampling random $r \in_R \Z_N$ and outputting:
	\[ \HSS.\HShare\Vector( (u_1,\dots,u_k); r) = \left( g^r, g^{rc} \cdot\prod_{i=1}^{k} g_i^{u_i} \right). \]
Scalar values $\alpha \in [M]$ will be encoded as Additive Secret Shares, as before.

Consider now homomorphic multiplication between an encoded \emph{vector} and scalar.
  \begin{itemize}
  \item Pairing: Perform the standard HSS pairing procedure between the additively shared $\langle \alpha \rangle$ and $\langle c\alpha \rangle$ together with the above ciphertext. This enables the parties to obtain group elements $g^\beta \in \G$ and $(g^\beta\cdot \prod_{i=1}^{k} g_i^{\alpha u_i})\in \G$ for some $\beta \in \Z_N$.

  \item Share Rerandomization: In order to rerandomize the exponent $\beta$, both parties multiply their local share by the \emph{same} random group element $g'$, computed as e.g.\ pseudorandom function of the unique instruction identifier. Note that this can be achieved with minimal additional share size (a single key to a pseudorandom function, included in each party's share) and computation. (Further, note that existence of pseudorandom functions is already implied by the DDH computational assumption.) We will thus roughly ignore this step in terms of analysis, and assume the value $\beta$ is distributed (pseudo-)uniformly, conditioned on the entire execution up to this point.

  \item Share Conversion: At this point, the parties will now attempt to extract shares of the exponent vector $(u_1,\dots,u_k)$ via execution of the $k$D-DDL algorithm.
  Let $A$ be a generic $(N,b,d,\delta)$-$k$D-DDLA, as per Definition~\ref{def:kDDL}. Recall for $x \in  \Sigma_b^N$ and $\beta \in \Z_N$ we denote by $A^\G(x, \beta)$ the execution of the algorithm $A$ with oracle access to the generic group represented by $x$ on the input challenge string $x[\beta] \in \Sigma_b$ representing the generic group string handle for the element $g^\beta$. (Note that the role of $\beta$ was notated by $v$ in the previous sections.) Then directly applying the $k$D-DDL property together with a union bound (see Remark below Definition~\ref{def:kDDL}), we directly have that for $\alpha \cdot u \in \Z_n^k$,
	\[ \pr_{ x \in_R\Sigma_b^N, \beta \in_R \Z_N} \left[ A^{\mathrm{G}}(x, \beta) -
		 A^{\mathrm{G}}\Big(x, \beta + \sum_{i=1}^{k} \alpha u_i \Big)
		 \neq \alpha u \right] \leq \delta  \alpha |u|_1, \]
where $|u|_1$ denotes the $\ell_1$ norm of $u$.

That is, aside from error probability bounded by $\delta \alpha |u|_1$, executing the algorithm $A$ with respect to the two values $g^\beta$ and $(g^\beta\cdot \prod_{i=1}^{k} g_i^{\alpha u_i})\in \G$ will result in precisely the desired additive output shares of the multiplied vector $\alpha u$.

  \end{itemize}
%


We now present a more detailed description of the Packed HSS construction from $k$D-DDL, based on the Decisional Diffie-Hellman assumption that underlies the construction.

\newcommand{\IG}{\mathcal{IG}}

\begin{definition}[DDH]
Let $\G=\{\mathbb{G}_\rho\}$ be a set of finite cyclic groups, where $|\mathbb{G}_\rho|=q$ and $\rho$ ranges over an infinite index set. We use multiplicative notation for the group operation and use $g \in \mathbb{G}_\rho$ to denote a generator of $\mathbb{G}_\rho$. Assume that there exists an algorithm running in polynomial time in $\log q$ that computes the group operation of $\mathbb{G}_\rho$. Assume further that there exists a PPT instance generator algorithm $\IG$ that on input $1^\lambda$ outputs an index $\rho$ which determines the group $\mathbb{G}_\rho$ and a generator $g \in \mathbb{G}_\rho$. We say that the Decisional Diffie-Hellman assumption (DDH) is satisfied on $\G$ if $\IG(1^\lambda)=(\rho,g)$ and for every non-uniform PPT algorithm ${\cal A}$ and every three random $a, b, c \in \{0,\ldots,q-1\}$ we have
$$|\mbox{Pr}[{\cal A}(\rho,g^a,g^b,g^{ab})=1]-\mbox{Pr}[{\cal A}(\rho,g^a,g^b,g^c)=1]| < \varepsilon(\lambda),$$
for a negligible function $\varepsilon$.
We will sometimes write $(\bbG,g,q) \gets \IG(1^\lambda)$.
\end{definition}



\begin{notation}  In this section: For $a \in \Z_q$, we denote by $\langle a \rangle$ additive secret shares $(a_0,a_1) \in \Z_q$. For selected generator $g_i$, we denote by $[a]_i$ the group element $(g_i)^a$. To maintain closer consistency to the notation of prior work, we will denote the (prime) group order by $q$ (as opposed to $N$).
\end{notation}

\begin{construction}[$k$-Packed HSS for Vector-Scalar Mult] \label{const:HSS}
Let $(\G,g,q) \in_R \IG(1^\lambda)$.
Let $g_1,\dots,g_{k}$ denote additional random generators of $\bbG$, selected as $g_i := g^{w_i}$ for random $w_i \in_R \Z_q$.

Let $A_\kDDL$ be a query-restricted $(N,b,d,\delta)$-$k$D-DDL algorithm, and $\PRF: \bbG \to \zo^b$ a pseudorandom function.
  \begin{itemize}
  \item $\HSS.\HShare(1^\secpar)$: Sample a random ElGamal secret key $c \in_R \Z_q$. To share each secret vector/scalar input, execute the corresponding algorithm:
    \begin{itemize}
    \item $\HSS.\HShare\Vector(\vec x)$: (``Encryption'') Given input vector $\vec x \in \zo^k$, output a {\em packed} ElGamal ciphertext,
  	$\ct = ([r]_1, [rc]_1 \cdot\prod_{i=1}^k [x_i]_i) \in \bbG \times \bbG$.
    \item $\HSS.\HShare\Scalar(y)$: (``Additive shares'') Given input scalar $y \in \zo$, output additive secret shares $\langle y \rangle,\langle cy \rangle$ over $\Z_q$. (Denote party $i$'s part of these values by $\share_i$.)
    \end{itemize}
  \item $\HSS.\HEval(i, ((\ct\fst,\dots,\ct\rhoth),(\share_i\fst,\dots,\share_i^{(\rho')})),P,r)$: Let $P$ denote a bilinear function
  	\[ P\left( (\vec x\fst,\dots,\vec x\rhoth), (y\fst,\dots,y^{(\rho')})\right) = \sum_{j \in [\rho],\ell \in [\rho']} \alpha_{j,\ell} \cdot y^{(\ell)} \cdot \vec x\jth \in \Z_\modout^k,\]
	 where $\alpha_{j,\ell} \in \Z_\modout$. Then homomorphic evaluation of $P$ takes place as in Algorithm~\ref{alg:HSS}.
	
	 \begin{algorithm}
	 \Begin{
	   Let bilinear function $P$ be defined by coefficients $\alpha_{j,\ell} \in \Z_\modout$
	
	   Initialize $z \in_R 1 \in \bbG$\;

  	   \For{$j=1$ to $\rho$, $\ell=1$ to $\rho'$}{
	
	     \If{$\alpha_{j,\ell} \neq 0$} {
	
			     $z \in_R z \cdot \Pair( \ct\jth, \share_i^{(\ell)})$\;
	     }
	   }
	
	   Output $A_\kDDL(z)$\;
	   }
	   \hrule
	   {\bf Subroutine} $\Pair(\ct,\share_i)$:
	
	   \Begin{
	
	   Parse $\ct = ([a], [b]) \in \bbG \times \bbG$ and $\share_i = (y_i, (cy)_i) \in \Z_q \times \Z_q$ \;
	
	   Output $[b]^{y_i} \cdot [a]^{-(cy)_i} \in \bbG$\;
	 }
	
	 \hrule
	 {\bf Subroutine} $\Convert(z)$:
	
	 \Begin{
	
	 Begin executing the DDL algorithm $A_\kDDL$
	
  	   \While{$A_\kDDL$ makes query $(1,\vec j) \in \Z_n \times \Z_n^k$} {
	
	   Compute $z_{\vec j} = z \cdot \prod_{i=1}^k [j_i]_i \in \bbG$\;
	
	   Respond to $A_\kDDL$ with value $\PRF(z_{\vec j}) \in \zo^b$\;
	
	   }
	
	   Output the value output by $A_\kDDL$\;
	 }
	 \caption{$\HSS.\HEval(i,(\ct\jth,\share_i^{(\ell)}),P,\modout)$ for bilinear functions, given $A_\kDDL$}
	 \label{alg:HSS}
	 \end{algorithm}
  \end{itemize}
\end{construction}


\paragraph{New solution: 1-dimensional embedding.}
Our second solution performs a \emph{1-dimensional embedding}, encoding a vector $u=(u_1,\dots,u_k)$ into a single scalar element $\sum_{i=1}^{k} u_i M^{i-1}$ in an enlarged payload space $[M^k]$ for chosen parameter $M$ (see below).


More concretely, in the 1-dimensional embedding approach, the existing (1-dimensional) HSS scheme is used in a black-box manner, embedding a vector value $u \in [M]^k$ as a single \emph{integer} $\tilde u := \sum_{i=1}^{k} M^{i-1} u_i \in [M^k]$ within a larger input space, where $M$ is a parameter chosen at the time of HSS encoding. (A larger choice of $M$ will ultimately result in smaller multiplication error, but will require greater runtime.) The corresponding share size of the vector $u$ is thus a \emph{single} ElGamal ciphertext. Homomorphic scalar-vector multiplication of $z = \alpha \cdot u$ will take place via two phases:
  \begin{itemize}
  \item Applying the 1D HSS multiplication procedure on the encoded secret \emph{integer} $\tilde u \in [M^k]$ (encoded via Encryption) together with secret scalar $\alpha \in [M]$ (encoded via Additive Shares), then with error probability $\tilde z/T^2$, one can homomorphically obtain additive secret shares of the \emph{integer} product $\tilde z := \alpha \tilde u = \sum_{i=1}^{k} M^{i-1} \alpha u_i$ over $\Z_N$.

  However, this is not the required output: HSS demands additive shares of the target vector $z = \alpha u$, over the corresponding \emph{vector space} $\Z_N^k$. This distinction is a crucial requirement for many HSS application settings, where additive shares are later combined or manipulated over the respective output space.\footnote{For example, HSS for program class $\cP$ yields succinct 2-server Private Information Retrieval (PIR)~\cite{ChorKGS98} for private database queries of related class $\cP'$~\cite{GilboaI14,BoyleGI15}, crucially depending on the \emph{additive} reconstruction of the HSS scheme over the output space of $\cP$.}
  \item To satisfy this requirement, the 1-dimensional embedding approach must thus add a step to revert shares of $\tilde z = \sum_{i=1}^{k} M^{i-1} z_i$ over $\Z_N$ to shares of $z = (z_1,\dots,z_{k})$ over $\Z_N^k$.

  Denote the original shares by $a,a' \in \Z_N$, and express as integers in base $M$: i.e., $a_0,a_1,\dots, a_m$ and $a'_0,a'_1,\dots,a'_m$ where $a_i,a'_i \in \{0,\dots,M-1\}$ and $m= \lceil \log_M(N) \rceil$. Conditioned on correct shares of $z$ over $\Z_N$, we have $a' = a + \sum_{i=1}^{k} M^{i-1} z_i$, which in turn implies each coordinate $a'_i = a_i + z_i$ as long as for each $i$ it holds that  $(a_i+z_i) < M$ (i.e., as long as there is no ``carry'' to the next power of $M$). Since $z_i \in [M]$, then by rerandomizing shares (i.e., adjusting both shares $a,a'$ by the same random additive offset in $\Z_N$), this bad event $(a_i +z_i) \ge M$ occurs only if $a_i \ge M-1-z_i $, i.e.\ with probability $z_i/M$.
  \end{itemize}
Combining the error of the HSS together with this error from share conversion (union bounding over the $k$ dimensions), the overall error of the 1-dimensional embedding approach becomes $(\sum M^{i-1} z_i/T^2) + (\sum z_i/M)$, where $z = \alpha u \in [M]^k$ and summations are over $i=1$ to $k$. 

\medskip

\begin{remark}[1D embedding parameter $M$]
Observe that the additive term $(\sum z_i/M)$ in the 1D embedding multiplication error expression does not decrease with runtime $T$. The choice of the parameter $M$ must thus be set sufficiently large to enable a desired error $(\sum z_i/M) < \epsilon$. However, $M$ must be selected at the time of HSS \emph{encoding}---to translate the plaintext vector $u$ to an integer $\sum M^{i-1} u_i$---at which point the eventual payload magnitude $\sum z_i$ and target error $\epsilon$ may not yet be known.

The 1D embedding solution thus requires one to \emph{predict} at encode time what the final payload magnitude and target error will eventually be. If the prediction mis-estimates the final expression $(\sum z_i/\epsilon)$ by a multiplicative factor $\xi$, then one of two things will take place. If $\xi<1$, i.e.\ $M$ was chosen too small, then the execution will have failed: error $\epsilon$ will be unachievable.  If $\xi>1$, i.e.\ $M$ was overestimated, then this will inflict a runtime overhead to appropriately shrink the second error term $(\sum M^{i-1}z_i/T^2)$. Specifically, to obtain a given error with this inflated $M$, one will now need to increase $T$ by a factor of $\xi^{(k-1)/2}$ to account for the extra leading $M^{k-1}$ term in the numerator.
\end{remark}




\paragraph{Comparison of approaches.}
Ultimately, the resulting parameters of the three approaches are summarized in the following theorem statement.

\begin{theorem}[Packed HSS from $k$D-DDL] 
Suppose there exists a query-restricted $(N,b,d,\delta)$-$k$-dimensional-DDL algorithm.
Then, based on the DDH assumption, there exists HSS for the class $\cP_{\blink}$ for vector-scalar operations with the following parameters. Multiplication error is given for vector-scalar multiplication $z = \alpha u$; summations are $\sum_{i=1}^{k}$.
  \begin{center}
  \begin{tabular}{l|c|cc|}
  			  				&\underline{Non-packed}	&\underline{1D Embedding $(M)$} 	&\underline{From $k$D-DDL}\\
  Share size, vector $u \in [M]^k$: 	& $k$ ElGamal CT		& 1 ElGamal CT 		& 1 ElGamal CT	\\
  Share size, scalar $\alpha \in [M]$:		& 2 $\Z_N$-elmts		& 2 $\Z_N$-elmts		& 2 $\Z_N$-elmts	\\
  Mult error, $z=\alpha u$, time $T$:		& $\sum z_i/T^2$ 		& $\sum (M^{i-1} z_i)/T^2 + (\sum z_i)/M$	& $\sum z_i \cdot \delta(T)$	\\
  \end{tabular}
  \end{center}
\end{theorem}

Note that the \emph{non-packed} application of~\cite{BGI16,BGI17,DKK18} has large vector share size.
As discussed above, if the desired final error probability and the magnitude of the final payload $z$ are known, then the error expression for the \emph{1D embedding solution} can be minimized to $\sim \sum z_i/T^{2/k}$, by setting the encoding parameter $M$ to $M=T^{2/k}$.

The \emph{solution from $k$D-DDL} does not require this a priori knowledge.
Plugging in the results from Section~\ref{sec:2dim}, for example, we have a provable 2D-DDL algorithm that makes $T$ queries (that can be implemented using $\tilde O(T)$ group multiplications) and results in 2-dimensional packed HSS with vector-scalar multiplication error $\sum z_i \delta(T) \sim \sum z_i T^{-7/8}$. Using the conjectured optimal algorithm, the $T^{-7/8}$ term can be replaced by $T^{-1}$.

\subsection{Location-Sensitive Encryption}
\label{sec-lse}

\input{LSE}

\subsection{Algorithmic Applications}
\label{sec-algorithmic}

In this section we discuss two representative algorithmic applications of  LPHS. These applications are generic in nature and could apply to any kind of LPH (see, e.g.,~\cite{IndykMRV97} for other examples). However, we have tried to identify the simplest settings and parameter regimes that benefit from the advantages of LPHS over alternative approaches.
In all of these applications the goal is to identify in {\em sublinear time}, and with low failure probability, either {\em small} or   {\em arbitrary} misalignments of two or more strings.

The first application only takes advantage of the short output length of the LPHS, whereas the second take advantage of the metric property of being ``locality preserving.'' Finally, while we describe the applications in the 1-dimensional, random-input case, they can naturally benefit from the $k$-dimensional, worst-case-input LPHS variants considered in this paper. In fact, 2-dimensional LPHS seems like the most useful variant in these contexts.

\subsubsection{Succinct sublinear-time sketching for shifts}
\label{sec-sketching}

Consider the following sketching scenario, described as a simultaneous messages (SM) communication complexity problem. Two parties $\cal A$ and $\cal B$, with shared randomness $\rho$, hold $n$-bit substrings $x_A$ and $x_B$ of a big string $X\in\{0,1\}^N$. The two inputs are within (non-cyclic) shift offset $s\in [-R,R]$ of each other, for some shift bound $R<n$.
Each party can send $c$ bits to Carol, where the message is computed by (adaptively) reading $d$ bits from the input. Carol should reconstruct $s$ from the two messages with error probability bound $\gamma$, assuming that $X$ is picked uniformly at random and independently of $\rho$. What are the achievable tradeoffs between the parameters?

\begin{claim}
\label{cl:alg1}
There exists a non-cyclic shift-finding SM protocol with $c=\log R+O(1)$ bits of communication, $d=O(\sqrt{n})$ input queries, and $\gamma=\tilde O(R/n)$ error probability.
\end{claim}
\begin{proof}
Let $h_\rho:\zo^n\to\Z$ be a (near-optimal) non-cyclic $(d,\delta)$-LPHS  with $d=O(\sqrt{n})$ and $\delta=\tilde O(1/n)$, as guaranteed by Theorem~\ref{thm:1dupper}. Each party sends to Carol the output of $h$ on its input, reduced modulo $2R+1$. Carol computes the difference $\delta$ between the two messages modulo $2R+1$, and outputs either $s=\delta$ if $0\le \delta\le R$ or $s=\delta-(2R+1)$ otherwise.
Using Lemma~\ref{lem:shift}, the error probability is at most $|s|\cdot\delta\le\tilde O(R/n)$ as required.
\end{proof}

In particular, if $R=\mathrm{polylog}(n)$, we get $\tilde O(1/n)$ error using $\log R+O(1)$ bits of communication (and with only $O(\sqrt{n})$ queries). Note that within the tight communication budget of $c=\log R+O(1)$, we cannot afford to amplify the success probability of a protocol based on the simple $(d,O(1/d))$-LPHS via repetition. Finally, the repetition-based approach cannot yield sublinear-time protocols with low error probability $\delta$ when $R$ is big, as required by the extension to unbounded shifts discussed below.

One can get a Las Vegas variant of the above sketching protocol, where Carol can {\em detect} whenever an error may occur (except with negligible probability), using the Las Vegas flavor of LPHS (see Definition~\ref{def-morelphs} and following remark).
Other LPHS variants can also be motivated in this setting. The assumption that $X$ is random can be relaxed to ``far from periodic'' by using the notion of worst-case LPHS from Appendix~\ref{sec-worstcase}. A 2D-LPHS can be used to capture a 2-dimensional terrain $X$.

\paragraph{\bf Sketching for unbounded shifts.} The above protocol can be extended to apply to an {\em arbitrary} shift amount, with $\mathrm{polylog}(n)$ communication, $d=\tilde{O}(n^{1/2})$ input queries, and $n^{-\omega(1)}$ failure probability, by running multiple instances of the protocol with $R=n/\mathrm{polylog}(n)$, where in each instance $\cal B$ shifts its inputs by a different multiple of $R$. Using a Las Vegas variant of LPHS, Carol can identify the correct instance. Alternatively, one can avoid using a Las Vegas variant and rely instead on Lemma~\ref{lem-smoothness} for identifying the distractors. Note that, unlike the case of small $R$ discussed above, here we cannot use at all the simple $(d,O(1/d))$-LPHS. Indeed, with $d=\tilde{O}(n^{1/2})$, the failure probability is too big to handle large shifts.

\subsubsection{Locality-sensitive hashing and near-neighbor data structures for shifts}
A near-neighbor data structure represents $m$ points in a metric space and enables efficient, e.g. sublinear time, near-neighbor queries on any point on the space. A near-neighbor query on point $x$ with distance $R$ returns a point $y$ in the data structure that is within  distance $R$ from $x$ or returns an indication of failure if no such point exists. Approximate near-neighbor queries relax the requirement so that with good  probability the output $y$ is within distance $cR$ from $x$ for some constant $c > 1$.

In this section we design  an approximate near-neighbor data structure for strings with distance measured by a shift metric, with applications to matching shifted pictures, or determining the location of an  agent in some terrain  given only a local view of its surroundings.

Intuitively, the term ``shift distance'' refers to the minimal shift amount required to obtain one string from the other. In the non-cyclic case this should correspond to directed distance in the De Bruijn graph, whereas in the cyclic case the distance is infinite if there is no such shift. Our notion of shift distance should not be confused with an alternative notion (cf.~\cite{andoni2013homomorphic}) referring to the smallest Hamming distance between one string and some cyclic shift of the other.

However, defining a shift distance correctly requires some care for non-cyclic shifts. Measuring the number of shifts required to obtain an $n$-bit string $y$ from an $n$-bit string $x$ in the directed  De Bruijn graph is not symmetric. Using the graph metric on the undirected De Bruijn graph leads to cases in which strings that should be distant in our proposed applications are near in the graph. In our proposed applications of this metric, the shifted strings are part of a larger ``universe''. We therefore adopt the following metric, defined over a sub-graph of the De Bruijn graph.

\begin{definition}[Shift metric]
Let $n,N$ be two integers $n \le N$ and let $C$ be a cycle of length $N$ in the undirected De Bruijn graph over strings of length $n$. The {\em shift metric} over $C$, denoted by $d_C$, is the standard graph metric on $V(C)$, the $N$ nodes in $C$.
\end{definition}
An equivalent way to view his definition is to regard $C$ as a circular string of length $N$ such that any substring of  length $n$ appears exactly once in $C$. The set of all $N$ substrings of length $n$ in $C$ is denoted $V(C)$. The distance $d_S(x,y)$ between two substrings $x,y \in V(C)$ is the minimum of two  values: the number of shifts on $C$ required to move from $x$  to $y$ and the number of shifts on $C$ required to move from $y$ to $x$.

A useful tool in the design of near-neighbor data structures is {\em Locality Sensitive Hashing} (LSH) which assigns to any two points $x,y$ the same value with high probability if they are close and two different values if they are distant. More precisely,

\begin{definition}[LSH]
Let $M$ be a set with a metric $d$. A family of hash functions $\cal{H}$ is a \emph{$(R,cR,p_1,p_2)$ Locality-Sensitive Hash (LSH)} for $(M,d)$ if for any two points $x,y \in M$
\begin{itemize}
\item If $d(x,y) \leq R$ then $\Pr_{h \in_R \cal{H}}[h(x)=h(y)] \ge p_1$.
\item If $d(x,y) \ge cR$ then $\Pr_{h \in_R \cal{H}}[h(x)=h(y)] \le p_2$.
\end{itemize}
\end{definition}

Constructing a near-neighbor data structure from an LSH (see \cite{andoni2006near} and references therein) uses a preprocessing step in which $L$ hash functions $h_1,\ldots,h_L$ are randomly chosen from $\cal{H}$ and $L$ hash tables are constructed. Each point $y$ is then placed in $L$ buckets $h_1(y),\ldots,h_L(y)$. Given a query point $x$ the distances between the points in all the buckets $h_1(x),\ldots,h_L(x)$ and $x$ are computed. If there exists some point $y$ in one of these buckets such that $d(x,y) \leq cR$ then $y$ is returned and otherwise failure is announced.

We construct an LSH family for the shift metric $d_C$ on a large cycle $C$ in the De Bruijn graph on $n$-bit strings. Similarly to the case of Location Sensitive Encryption the input string cannot be too regular. We begin by showing an LSH family for random input and then discuss how to achieve similar results to the case of worst-case inputs which are \Good~in the sense of Appendix~\ref{sec:worst}.



\begin{theorem}
\label{th:lsh}
Let $\Sigma$ be an alphabet, $|\Sigma| \ge n^3$, let $n,N$ be integers, $N>n$, and let $C$ be a random $N$-bit binary string. For any constant $c>1$ there exist a constant $a$ and a hash family $\cal{H}$ of functions $h:\Sigma^n \to \zo$ that is $(R,cR,p_1,p_2)$-LSH for the cyclic shift metric over $C$, with parameters $0 \le R \le \frac{n}{2ac}$, $p_1=1-\frac{R(2a+1)+1}{n}$, and $p_2=p_1-\frac{R((2c-2)a-c-1)}{n})$. In addition, any $h \in \cal{H}$ can be computed with $\sqrt{n}$ queries $x[i]$ to the input string $x$, and the same result holds for any alphabet $\Sigma$ with $a$ that is polylogarithmic  in $n$.
\end{theorem}

\begin{proof}
Let ${\cal H}'$ denote the LPHS family constructed in Theorem~\ref{thm:1dupper} that with $\sqrt{n}$ queries of an input string of length $n$ achieves error  probability $a/n$, such that $a$ is a constant for $|\Sigma| \ge n^3$, and is polylogarithmic in $n$ for general alphabet. Let $\cal{H}$ be the family of hash functions $h_u:V(C) \to \zo$, for $u=(h',z)$, $h' \in \cal{H}'$, $z \in \{0,\ldots,n-1\}$  defined by
$$
h_u(x)=
\begin{cases}
0 & (h'(x) \bmod n) < z \\
1 & \mbox{otherwise}
\end{cases}
$$

If $x,y \in V(C)$ and $d_C(x,y)=1$  then  $\mbox{Pr}_{h' \in \cal{H}'}[h'(x)=h'(y)+1] = 1-a/n$. Therefore, by union bound, if the distance between $x$ and $y$ is $r \le R$  then $\mbox{Pr}_{h' \in \cal{H}'}[h'(x)=h'(y)+r] = 1-ar/n$. If $h'(x)=h'(y)+r$ then $h_u(x) \neq h_u(y)$ if $z$ is chosen so that it is between $h'(x) \bmod n$ and $h'(y) \bmod n$. There are two possible cases for $h'(x)=h'(y)+r$: either $h'(x) \bmod n>h'(y) \bmod n$ over the integers or $h'(x)> n-r$. In  the first case we have that $\Pr_{h_u \in \cal{H}}[h_u(x)=h_u(y)]=1-r/n$. While for the second case it follows from the proof to Lemma \ref{lem-smoothness} (plugging in $m=n$ and $\delta=a/n$) that $\Pr_{x \in \Sigma^n}[(h'(x) \bmod n > n-r] \le 1/n+ ar/n$. It follows that if the shift distance between $x$ and $y$ is at most $R$  then
$$\mbox{Pr}_{h_u \in \cal{H}}[h_u(x)=h_u(y)]\ge \left(1-\frac{ar}{n}\right)\left(1-\frac{r+1+ar}{n}\right) \geq 1-\frac{R(2a+1)+1}{n}.$$

If $x,y \in V(C)$ and $d_C(x,y) \ge cR$ then there are two cases: $(h'(x) -h'(y)) \bmod n \le cR$ and $(h'(x) -h'(y)) \bmod n > cR$. We show an  upper bound for the probability of the first case, and use the fact that in the second case, the probability that $h_u$ assigns different values to $x$ and $y$ is at least $\frac{cR}{n}$. To bound the first case we divide into two sub-cases: $d_C(x,y) \ge n$ and $d_C(x,y)<n$. In the first sub-case, $y$ is a random string, independent of $x$. By the proof to Lemma \ref{lem-smoothness} and by union bound it holds that $\mbox{Pr}_{y \in \Sigma^n}[(h'(x) - h'(y)) \bmod n \le cR]\le \frac{1+acR)}{n}$. The probability that the second sub-case occurs is bounded by $\mbox{Pr}[(h'(x)-h'(y)) \bmod n \le cR] \le acR/n$, by the definition of LPHS and union bound. Therefore,

$$\Pr[h_u(x)=h_u(y)]	\le \frac{1+2acR}{n}+1-\frac{cR}{n}=1-\frac{cR(2a-1)+1}{n}.$$
\end{proof}

For worst-case results consider input that is \Good~for the whole cycle $C$ and windows of size $n$. That is, assume that for every every two strings $x,y \in V(C)$ it holds that a fraction $\alpha$ of the strings is different. If $\alpha$ is a constant then the result of Theorem \ref{th:lsh} holds with a degradation in the probabilities $p_1$ and $p_2$ of at most $O(a/n)$.

A data structure for approximate near-neighbor searches can be constructed directly from binary-LSH for the shift metric. In the interest of brevity we describe a slightly different construction based on previous work on near-neighbor data structures for the $L_1$ metric.

\begin{corollary}
Let $\Sigma$ be an alphabet, let $n,N$ be integers, $N>n$, and let $C \in \Sigma^N$ be \Good. There exists a $c$-approximate $R$-near-neighbor data structure for the shift metric on the strings in $C$, for a constant $c$, $R=o(n)$, with size $\tilde{O}(mn\log{|\Sigma|}+m^{1+1/c})$ and search time $\tilde{O}(m^{1/c})$.
\end{corollary}
\begin{proof}
Let the strings in the data structure be $x_1,\ldots,x_m$ and let ${\cal H}'$ denote the LPHS family constructed in Theorem~\ref{thm:1dupper}. Choose $k$ random functions $h_1,\ldots,h_k \in {\cal H}'$, $k=\omega(\log nm)$. For each $h_j$, construct the $c$-approximate near-neighbor data structure for the $L_1$ metric given by
 Andoni and Indyk in \cite{andoni2006efficient,andoni2005approximate} on inputs $h_j(x_1),\ldots,h_j(x_m)$. In each cell that stores $h_j(x_i)$ store a pointer to $x_i$.

Search for a near-neighbor to a string $y$ by running an independent search on each structure for $h_j$. The result is a list of possible neighbors $x_{i_1},\ldots,x_{i_t}$, such that for each $x_{i_\ell}$ for at least one $h_j$ it holds that the $|h_j(x_{i_\ell})-h_j(y)| \le cR$. To test whether $h_j$ correctly  measures the shift distance between  $y$ and $x_{i_\ell}$ choose at random $\omega(\log n)$ locations in $y$ and check that the shifted locations in $x_{i_\ell}$ are identical.

The near-neighbor data structure for the $L_1$ metric is of size $\tilde{O}(m^{1+1/c})$, which together with the $m$ strings of length $n$ that the structure must store implies the size of the structure in the statement of the corollary. The search time in the $L_1$ structure is dominated by computing the distance of $O(m^{1/c})$ points to the query point. In the current structure that requires $\tilde{O}(m^{1/c})$ with the $\tilde{O}$ notation used for polylogarithmic factors in $nm$.

Failing to find a near neighbor can have one of two causes. The first is failure in the $L_1$ search  structure, which has negligible probability \cite{andoni2006efficient,andoni2005approximate}. The second is that the shift distance between $y$ and some $x_i$ is small, but the difference $|h_j(y)-h_j(x_i)|$ is large. The probability of that event for a specific $h_j$ is $\tilde{O}(1/n)$ by Theorem~\ref{thm:1dupper} and the probability that it will occur for all functions $h_j$ is negligible. Reporting a distant string as a neighbor can occur only if $\omega(\log n)$ locations in $y$ are identical to the same locations plus a small shift (less than $R$) in $x_i$, which has negligible probability for substrings of an \Good~$C$.
\end{proof}


%% file: LSE.tex

\newcommand{\Gen}{{\sf Gen}}
\newcommand{\Enc}{{\sf Enc}}
\newcommand{\Dec}{{\sf Dec}}
\newcommand{\Embed}{{\sf Embed}}
\newcommand{\pp}{{\sf pp}}
\newcommand{\calA}{\mathcal{A}}
\newcommand{\calM}{\mathcal{M}}
\newcommand{\dist}{{\sf dist}}
\newcommand{\negl}{{\sf negl}}
\newcommand{\GOOD}{{\sf Good}_n}
\newcommand{\Ext}{{\sf Ext}}
\newcommand{\logsq}{\log^2 n}
\newcommand{\wc}{{\sf wc}}



In this section, we present an application of LPHS techniques to a form of {\em location-sensitive encryption} (LSE). At a high level, an LSE scheme enables a user to encrypt messages with respect to his location (e.g., within a virtual world), as captured by a substring $x \in \zo^n$ representing the user's view within a global environment (modeled as a much larger binary string). The LSE scheme should support two properties: (1) that any other user within close proximity to the location of encryption can decrypt, and (2) that any user who is far from this location does not learn any information about the hidden plaintext. We focus on the case of a 1-dimensional such string for simplicity, although an analogous approach can be take for 2 dimensions. Proximity in this setting corresponds to shifted view strings $x \in \zo^n$, $(x \lll i) \in \zo^n$. 

Importantly, we seek LSE schemes whose efficiency requirements are {\em sublinear} in the (potentially large) view size $n$. In particular, this rules out approaches based on existing constructions of ``fuzzy extractors''~\cite{DORS} (loosely, extractors robust to small input noise) with respect to edit distance. While such object would imply LSE, the constructions operate via a combination of tiling and hashing for set similarity, which require linear time.

However, the constructions in this section take inspiration from~\cite{DORS}, and can informally be viewed as constructing a {\em sublinear-time} fuzzy extractor for shift distance.

We remark that the above-described goal of location-sensitive encryption is inherently different from ``position-based encryption'' as in~\cite{ChandranGMO14}, where parties rely on geographic location at the time of communication and leverage the speed of light to ensure a single recipient.

\paragraph{Conditions on global environment.}
In order to provide the desired LSE guarantees, it must of course hold that the global environment string is neither too regular nor predictable. 
 
Regularity: 
We follow the terminology as introduced in the worst-case LPHS section, in Section~\ref{sec-worstcase}.  
Recall that $\sfGoodw$, parameterized by ``window size'' $W \in [n]$ and difference parameter $0 < \alpha \le 1$, denotes the set of strings $x$ in $\zo^n$ such that each of the length-$W$ substrings of $x$ differ pairwise in at least $\alpha$ fraction of their symbols (see Definition~\ref{def:good-input}). For purposes of concreteness and streamlined notation, we will focus on inputs in $\sfGoodw$ for specific fixed parameters:

\begin{notation}
Within this section, we will denote $\sfGoodw$ for $\alpha = 0.1$ and $W = 0.1 n$ by the abbreviated notation $\GOOD$. 
\end{notation}

Predictability: In order to provide the desired secrecy property for parties whose view strings are not overlapping the encryptor's, it must necessarily be that un-viewed portions of the global environment string are unpredictable. A natural approach would be to place some entropy requirement on the global string; however, leveraging such worst-case entropy while maintaining sublinear complexity would necessitate heavy requirements on the amount of such entropy available. We instead consider the following notion of {\em local unpredictability}, which in particular holds for sources with constant-fraction entropy, but further holds for most distributions whose entropy is only polylogarithmic, including many occurring naturally within applications. Roughly, a source $X$ over $\zo^n$ is locally unpredictable if one cannot predict the symbols $x_{i+S}$ of a sample $x \in_R X$ at any $i$-shift of random challenge index set $S \subset [n]$, except with negligible probability.

\begin{definition}[Local Unpredictability] \label{def:localunpred}
We will say a distribution $X$ on $\zo^n$ is {\em $(m,\epsilon)$-locally unpredictable} with respect to window size $W$ if for every algorithm $\calA$, the probability of $\calA$ winning in the following challenge is $\epsilon(n)$.
  \begin{enumerate}
  \item The challenger samples $x \in_R X$ and a random $m(n)$-size subset of coordinates $S \in_R {[W] \choose m}$. It sends $S$ to $\calA$. 
  \item The algorithm $\calA$ must output a pair $(i,x'_S) \in [n] \times \zo^{m}$. It wins if $x'_{i+S} = x_{i+S}$, where $i+S := \{ i + s \mod n : s \in S \}$.
  \end{enumerate}
For purposes of this section, if we say $X$ is simply ``locally unpredictable,'' this will implicitly refer to a convenient case where $m(n) = \log^4 n$ and $\epsilon(n) = 2^{-\log^{3} n}$, for window size $W=0.1n$.
\end{definition}




For a string $x \in \zo^n$ and $\ell \in [n]$, recall the notation $x \lll \ell$ denotes a randomized process which samples a random suffix $x'_2 \in_R \zo^\ell$ and outputs the $n$-bit string formed by the last $(n-\ell)$ bits of $x$ appended with the $\ell$ bits $x'_2$. We will use the notation $x' \in x \lll \ell$ to denote that $\Pr[ x' = x \lll \ell] > 0$; equivalently, $x' \in {\sf Supp}(x \lll \ell)$.


\medskip

We now define the Locality-Sensitive Encryption notion that is the focus of this section.

\begin{definition}[Location-Sensitive Encryption]
An $\alpha$-{\em location-sensitive encryption} (LSE) scheme 
for message space $\calM$ is a pair of PPT algorithms $(\Enc,\Dec)$ with the following syntax:
  \begin{itemize}
  \item $\Enc(m,x)$ takes as input message $m \in \calM$ and location string $x \in \zo^n$, and outputs a ciphertext $c$.
  \item $\Dec(c,x')$ takes as input ciphertext $c$ and location string $x' \in \zo^n$, and outputs a plaintext value $m' \in \calM$.
  \end{itemize}
The scheme is said to be a {\em sublinear LSE} if $\Gen$ and $\Enc$ make oracle access into the bits of their respective inputs $x,x'$, and the number of such queries made is $o(n)$.

An LSE scheme must satisfy the following correctness and security properties:
  \begin{itemize}
  
  \item {\bf Nearby decryption.}  For every message $m \in \calM$, location string $x \in \GOOD$, and $x' \in x \lll i$ for some $i \le \alpha n$, it holds that
  	\[ \Pr_{c \in_R \Enc(m,x)} \left[  \Dec(c,x') = m \right] = 1 - \negl(n). \]
	

  \item {\bf Far-distance security (from local unpredictability).} 
  Let $X$ be any distribution on $\GOOD$ which is 
  locally unpredictable (Definition~\ref{def:localunpred}). Then it holds for $X$ and for any $m,m' \in \calM$ that encryptions of $m$ and $m'$ are computationally indistinguishable: 
  	\[ \{ \Enc(m,X) \} \overset{c}{\cong} \{ \Enc(m',X) \}. \]

  
  \end{itemize}

\end{definition}

\subsubsection{Building Sublinear LSE}

In what follows, we demonstrate how to achieve a {\em sublinear} location-sensitive encryption scheme, taking inspiration from notions of fuzzy extractors and biometric embeddings of Dodis {\em et al.}~\cite{DORS}.
The high-level idea will be to build a tool comparable to a fuzzy extractor for {\em shifts}. 
This is done by constructing a form of metric embedding from shift into Hamming distance, which then enables us to directly appeal to fuzzy extractor results for Hamming metric.\footnote{We remark that a sublinear method for embedding edit distance into Hamming distance was shown in a recent independent work~\cite{KociumakaS20}, also using a random walk technique.} The construction of such embedding is the focus of this subsection. 

We begin by constructing a weaker tool---a Binary LPHS---and then amplify. Loosely, a Binary LPHS family maps inputs to a single bit, such that close inputs with respect to shift distance are mapped to the same bit with good probability, whereas inputs with sufficient unpredictability hash unpredictably. 

The exposition is organized as follows. We first present the definition of shift-to-Hamming embedding. We present and construct a notion of Binary LPHS. We then provide an amplification procedure which attains a shift-to-Hamming embedding given access to Binary LPHS. Finally, we demonstrate how to pair this embedding together with techniques from~\cite{DORS} to reach the desired LSE primitive.

For $x,x' \in \zo^n$, we denote Hamming distance of $x$ and $x'$ by $\Delta(x,x')$.

\begin{definition}[Shift-to-Hamming Embedding] \label{def:StH}
We define a $(\alpha,\alpha')$-randomized {\em shift-to-Hamming embedding} with output length $\ell=\ell(n)$ as a pair of polynomial-time algorithms $(\Gen,\Embed)$ with the following syntax. 
  \begin{itemize}
  \item $\Gen(1^n)$ is a randomized procedure that takes input length $1^n$ in unary and outputs an index value $v$. 
  \item $\Embed(v,x)$ is a deterministic procedure that takes as input an index $v$ and input string $x \in \zo^n$, and outputs a value $y \in \zo^\ell$. 
  
  We will consider sublinear embeddings, wherein $\Embed$ makes oracle access into the bits of $x$, and the number of such queries made is $o(n)$.
  \end{itemize}
The algorithms satisfy the following properties.
  \begin{itemize}
  \item {\bf Preserves closeness.} For any $x \in \GOOD$ and $x' \in \zo^n$ for which $x' \in x \lll i$ for some $i \le \alpha n$, it holds that 
  	\[ \Pr_{v \in_R \Gen(1^n)} [\Delta\big(\Embed(v,x),\Embed(v,x')\big) \le \alpha' \ell] \ge 1 - \negl(n). \]

  \item {\bf Preserves unpredictability.} Let $X$ be any distribution on $\GOOD$ which is locally unpredictable (Definition~\ref{def:localunpred}). Then for every algorithm $\calA$, the probability of $\calA$ winning in the following challenge is negligible in $n$. 
	\begin{enumerate}
	\item The challenger samples $v \in_R \Gen(1^n)$ and $x \in_R X$. It sends $v$ to $\calA$.
	\item The algorithm $\calA$ must output $y' \in \zo^\ell$. It wins if $y' = \Embed(v,x)$.
	\end{enumerate}
	
  
  \end{itemize}
\end{definition}

Note that the ``preserves unpredictability'' property guarantees that the output string (on a locally unpredictable output) has super-logarithmic minimum entropy. Once such embedding is reached, this can be combined with constructions of (computational) fuzzy extractors for Hamming distance, to yield the desired LSE encryption scheme. 
While this entropy level is not sufficient for the existence of information theoretically secure fuzzy extractors for Hamming distance~\cite{FullerRS20}, computational constructions exist based on plausible computational hardness assumptions (see e.g.~\cite{FullerMR20,HerderRDYD17,CanettiFPRS21,AlamelouBCCFGS18,WenLH18}). The resulting LSE security inherits the corresponding agreement/entropy parameters and computational assumption of the underlying fuzzy extractor.

We will achieve the desired shift-to-Hamming embedding by means of the following tool, a form of Binary LPHS.\footnote{Note that this notion is highly related but not equivalent to the Binary LSH primitive defined and constructed in Theorem~\ref{th:lsh}.} As described earlier, a Binary LPHS family maps inputs to a single bit, such that close inputs map to the same bit with good probability, whereas inputs with sufficient unpredictability hash in an unpredictable manner.

\begin{definition}[Binary LPHS] 		\label{def:BinLPHS}
A family of hash functions $\calH = \{ h: \zo^* \to \zo\}$ is said to be a $(\alpha,\beta)$-{\em binary LPHS} if the following properties hold. 
  \begin{itemize}
  \item {\bf Close inputs agree.} For any $x \in \GOOD$ and $x' \in \zo^n$ for which $x' \in x \lll i$ for some $i \le \alpha n$, it holds 
  	\[ \Pr_{h \in_R \calH} [ h(x) = h(x') ] \ge 1-\beta.  \]
  \item {\bf Output unpredictability for locally unpredictable inputs.} 
  Let $X$ be any distribution on $\GOOD$ which is locally unpredictable (Definition~\ref{def:localunpred}). Then for every algorithm $\calA$, the probability of $\calA$ winning in the following challenge is bounded by $1/2 + \negl(n)$. 
	\begin{enumerate}
	\item The challenger samples $h \in_R \calH$ and $x \in_R X$. It sends $h$ to $\calA$.
	\item The algorithm $\calA$ must output $y \in \zo$. It wins if $y = h(x)$.
	\end{enumerate}  
  \end{itemize}
\end{definition}

We next proceed to build the above notion of Binary LPHS. The approach borrows techniques from worst-case LPHS: first reducing to the random-input case via a ``random subset tiling'' of the input $x$ (to obtain a related string over a larger alphabet with all distinct symbols) followed by a $t$-wise independent hash applied to individual symbols. As a second step (similar to worst-case LPHS) we can now apply an average-case LPHS on the resulting string $y$. However, instead of simply outputting the resulting LPHS-output integer $z$ (or compressed version), which has limited unpredictability, we instead use this value to {\em select a symbol} from the intermediate randomized string $y$, and extract out a single bit output.

Loosely, given inputs $x,x'$ differing by small shift $i$, the resulting strings $y,y'$ will differ by shift of $i$, and the LPHS will (aside from some error) provide outputs $z$ and $z' = z-i$ differing by $i$. Thus the two outputs will select the {\em same} symbol, $y_z = y'_{z'}$. On the other hand, for any $x$ with sufficient unpredictability, then no single symbol of the corresponding string $y$ will be predictable, thus the resulting bit will be close to uniform.

\newcommand{\LSH}{{\sf LSH}}

\begin{proposition}[Constructing Binary LPHS]
There exist $0 < \alpha < \beta < 1$ for which there exists $(\alpha,\beta)$-binary LPHS making $\tilde O(\sqrt{n})$ queries to the input $x$.
\end{proposition}
\begin{proof}
We provide the desired binary LPHS construction. Consider the following tools (each with $\tilde O(\sqrt{n})$ query complexity into the input $x$, when relevant): 
  \begin{itemize}
  \item A hash function family $\calH = \{h_{S,\gamma}\}$, indexed by a subset $S \subset [n]$ of size $\logsq$, and a hash function $\gamma : \zo^{\logsq} \to \zo^{\logsq}$ from an $n$-wise independent hash family. Sampling a hash function $h_{S,\gamma}$ from $\calH$ will consist of randomly selecting $S \subset [n]$ and sampling $\gamma$ from the $n$-wise independent hash family. The output of $h_{S,\gamma}(x)$ is the string $(y_i)_{i \in [n]} \in (\zo^{\logsq})^n$ for which $y_i = \gamma(x_{i+S})$.
  \item An $(n,d,\delta)$-non-cyclic {\em average-case} LPHS $h^*$ with shift bound $\alpha n$, for random inputs in $(\zo^{\logsq})^n$, for $d \in \tilde O(\sqrt{n})$ and $\delta \in \tilde\Theta(1/n)$.
  \item A seeded extractor $\Ext: \zo^{\logsq} \times \zo^{\logsq}\to \zo$ for entropy sources $X$ over $\zo^{\logsq}$ with $H_\infty(X) \ge \log^{3/2} n$. Namely, for any such $X$, it holds $(U,\Ext(U,X)) \overset{s}{\cong} (U,U_{\zo})$.
  \end{itemize}
  
  We observe that the following combinations of these tools are proved to satisfy various properties in other sections of this work:
    \begin{itemize} 
    \item $h_{S,\gamma}: \zo^n \to (\zo^{\logsq})^n$ for randomly sampled $S,\gamma$ is proved in Appendix~\ref{sec-worstcase} to convert a worst-case input $x \in \GOOD$ to a {\em random} input $h_{S,\gamma}(x) \in (\zo^{\logsq})^n$, while preserving shift distance between close input pairs.
    \item The composition $h_\wc := h^* \circ h_{S,\gamma}: \zo^n \to \Z$ for randomly sampled $S,\gamma$ is proved in Appendix~\ref{sec-worstcase} to be a {\em worst-case LPHS} for inputs $x \in \GOOD$. 
    \end{itemize}
  
  We construct the desired Binary LPHS hash function family $\calH$ as follows.
    \begin{itemize}
    \item {\bf Sampling.} A hash function $h \in_R \calH$ is sampled from the family by sampling $S,\gamma$ as above, and sampling a random extractor seed $r \in_R \zo^{\logsq}$. The hash function is indexed by $(S,\gamma,r)$
    \item {\bf Evaluating.} To evaluate $h = h_{(S,\gamma,r)}(x)$: 
      \begin{enumerate}
      \item Compute $y = h_{S,\gamma}(x) \in (\zo^{\logsq})^n$.
      \item Compute $z = h^*(y) \in [n]$. (This corresponds to $z=h_\wc(x)$.)
      \item Output $\Ext(r,y_z) \in \zo$.
      \end{enumerate}
    \end{itemize}
    
We now prove that $\calH$ satisfies the necessary properties.

{\em Close inputs agree.} Let $x \in \GOOD$ and $x' \in \zo^n$ for which $x' \in x \lll i$ for $i \le \alpha n$. Consider the probability space defined by sampling $(S,\gamma,r)$ via $h \in_R \calH$. Denote $y=h_{S,\gamma}(x),y'=h_{S,\gamma}(x'), z=h^*(y)$, and $z' = h^*(y')$. Define the event $E$ to occur when the following two conditions hold:
	\begin{enumerate}
	\item $y' \in y \lll i$: that is, the tiling of $x$ and $x'$ properly preserves their shift.
	\item $z = z' + i$: that is, correctness holds for the worst-case LPHS $h_\wc$ applied to $x,x'$.
	\end{enumerate}
It is proved in Section~\ref{sec-worstcase} (as part of worst-case LPHS Proposition~\ref{prop:worst-case}), together with a union bound on $i$, that event $E$ fails to occur with probability no greater than $i\delta$.
Conditioned on event $E$, then (since $z$ is sufficiently small) it holds that $y'_{z'} = (y \lll i)_{z-i} = y_z$; in particular, $\Ext(r,y_z) = \Ext(r,y'_{z'})$.
%

{\em Output unpredictability for locally unpredictable inputs.}  Let $X$ be a locally unpredictable distribution on $\zo^n$. Consider the probability space defined by sampling $(S,\gamma,r)$ via $h \in_R \calH$ and $x \in_R X$. Consider a hash $h(x)$-predictor algorithm $\calA$ who is given both the hash index $(S,\gamma,r)$ (as is the case in the security game), as well as the value $z = h_\wc(x) \in [n]$ as additional leakage on $x$. 
By the local unpredictability of $X$, given just $S$, no algorithm can predict $x_{i+S}$ for any shift $i \in [n]$ with better than negligible probability $\epsilon = 2^{-\log^3 n}$. The values $\gamma,r$ are independent and thus do not affect this probability. The leakage $z=h_\wc(x) \in [n]$ provides at most $\log n$ bits of information and hence cannot improve the prediction ability beyond a factor of $2^{\log n}$. In particular, $\calA$ cannot predict $y_z = x_{z+S}$ with probability better than negligible $2^{\log^2 n}$; that is, $H_\infty(y_z | S,\gamma,r,z) \in \Omega(\log^2 n)$. Thus, it holds that $((S,\gamma,r),z,\Ext(r,y_z)) \overset{s}{\cong} ((S,\gamma,r),z,U_{\{0,1\}})$. The claim follows.

  
\end{proof}


We next demonstrate how to attain a randomized shift-to-Hamming embedding, making use of the Binary LPHS. Simply, the embedding map will be generated by sampling polylogarithmically many independent Binary LPHS hash function descriptions, and the embedding of an input $x \in \GOOD$ is performed by evaluating each hash function on $x$. At a high level, for any two inputs $x,x'$ close in shift distance, a large fraction of the Binary LPHS evaluations will agree, resulting in close Hamming distance of the respective outputs; in contrast, any input distribution $X$ with local unpredictability will introduce entropy into several of the hash output values. 

The formal unpredictability argument is slightly more subtle, as the hash functions (while independent) are each applied to the {\em same} input sample. Thus hash outputs $h_1(x),...,h_j(x)$ can be viewed as side information on $x$ that can aid in predicting $h_{j+1}(x)$. However, as we prove, this leaked information still leaves sufficient local unpredictability in the sample $x$, thus allowing us to appeal to the existing entropy argument.


\begin{proposition}
Let $\calH$ be a $(\alpha,\beta)$-binary LPHS, as per Definition~\ref{def:BinLPHS}. Then the following procedure is a $(\alpha,\beta')$ random shift-to-Hamming embedding, as per Definition~\ref{def:StH}, for any $\beta' > \beta$, with output length $\ell(n) = \log^2 n$.
  \begin{itemize}
  \item $\Gen(1^n)$: Sample $\ell = \log^2 n$ independent hash functions $h_1,\dots,h_\ell \in_R \calH$. \\Output $v = (h_1,\dots,h_\ell)$.
  \item $\Embed(v,x)$: Parse $v = (h_1,\dots,h_\ell)$. Output $(h_1(x),\dots,h_\ell(x)) \in \zo^\ell$.
  \end{itemize}
\end{proposition}
\begin{proof}
Consider the necessary properties.

{\em Preserves closeness}: 
  Let $x \in \GOOD$ and $x' \in x \lll i$ for $i \le \alpha n$. Then it holds that 
	\begin{align*}
  	\Pr_{v \in_R \Gen(1^n)} [\Delta(\Embed(v,x),\Embed(v,x')) \le \beta' \ell]  
		&=  \Pr_{h_j \in_R \calH, j \in [\ell]} \left[ \sum_{j \in [\ell]} \big| h_j(x) - h_j(x') \big| \le \beta' \ell \right]	
	\end{align*}
Now, recall for any $x \in \GOOD$ that each $h_j$ {\em independently} satisfies $h_j(x)=h_j(x')$ with probability at least $1-\beta$. Each $|h_j(x)-h_j(x')|$ can then be analyzed as an independent Bernoulli boolean variable equal to 1 with probability no greater than $\beta$. The probability expression above is thus bounded by $e^{-\Omega(\ell^2)}$ by a Chernoff bound. Since $\ell=\logsq$ this is negligible in $n$, as desired.
	

 {\em Preserves unpredictability}: Let $X$ be a locally unpredictable distribution over $\zo^n$. We wish to show that no algorithm $\calA$ can predict the output $\Embed(v,x)$ (for random $v \in_R \Gen(1^n), x \in_R X$) given the hash index $v$, except with negligible advantage. Fix a strategy $\calA$ for $X$.
Let
	\[ p_j :=  \Pr_{\substack{h_1,\dots,h_{j} \in_R \calH, \\x \in_R X}} \left[ \calA_\ell \big(v,h_1(x),\dots,h_{j-1}(x) \big) = h_j(x) \right] \]
  Then 
    \begin{align*} 
      \Pr_{\substack{v \in_R \Gen(1^n), \\x \in_R X}} &\left[ \calA(v) =  \Embed(v,x) \right]
      	= \Pr_{\substack{h_j \in_R \calH, \\x \in_R X}} \left[ \forall j \in [\ell], \calA_j(v) = h_j(x) \right] \le \prod_{j \in [\ell]} p_j.
    \end{align*}
Consider a single $p_j$. We know the input distribution $X$ is locally unpredictable, where no algorithm can win the game of Definition~\ref{def:localunpred} with probability better than $2^{-\log^3 n}$. In the expression for $p_j$, an algorithm is given leakage on $x$, in the form of the previous hash function evaluations. In the worst case, there are $\ell = \logsq$ such bits of leakage on $x$. However, such leakage can at best improve the local unpredictability success probability to $(2^{-\log^3 n})(2^{\logsq}) < 2^{-\Omega(\log^3 n)}$. That is, even the distribution of $X$ conditioned on the leakage satisfies local unpredictability, the output unpredictability property of the Binary LPHS will apply, implying that $p_j$ is bounded above by $1/2 + \negl(n)$. Thus, the desired probability above is bounded by $\prod_{j \in [\ell]} (1/2 + \negl) \in 2^{-\Omega(\logsq)}$, which is negligible, as required.
    

\end{proof}


%% file: DKK-Results.tex

\section{LPHS Results Based on Iterative Random Walks~\cite{DKK18}}
\label{sec:dkk}
\label{app:dkk}

In the rest of this section, we summarize the IRW algorithm used to derive Theorem~\ref{thm:1dupper}.
Note that it is sufficient to derive the theorem for $b \ge 3 \log n$, while the theorem for $b < 3 \log n$
follows from Lemma~\ref{lem:smaller-alphabet}.

\subsection{The Basic LPHS}
\label{sec:basic}


We begin by describing the min-based LPHS~\cite{BGI17} in Algorithm~\ref{alg:basic} and refer to it as the basic LPHS.
It scans the $d$ values $x[0],x[1],\ldots,x[d-1]$ and chooses the index $i_{min}$ for which $x[i]$ is minimal.
The output of the LPHS is $\mathrm{Basic}_d(x)= i_{min}$.

The motivation behind the algorithm is apparent: $\mathrm{Basic}_d(x)$ and $\mathrm{Basic}_d(x \ll 1)$ scan the two lists $x[0],x[1],\ldots,x[d-1]$ and $x[1],x[2],\ldots,x[d]$, respectively. These lists have $d-1$ common symbols, and with high probability, the minimal value among $x[0],x[1],\ldots,x[d-1],x[d]$ is unique and obtained on a common symbol, implying
$\mathrm{Basic}_d(x) = \mathrm{Basic}_d(x \ll 1) + 1,$ as desired.

\paragraph{Error probability.} The following lemma calculates the error probability of the basic DDL algorithm as a function of $d$.
\begin{lemma}\label{lem:basic_fail_prob}
	The error probability of the basic LPHS is
	\[
		\pr_{x}[\mathrm{Basic}_d(x) - \mathrm{Basic}_d(x \ll 1) \neq 1] \leq \frac{2}{1+d} + 1/n.
	\]
\end{lemma}

\begin{proof}
Executions $\mathrm{Basic}_d(x)$ and $\mathrm{Basic}_d(x \ll 1)$ scan the two lists $x[0],x[1],\ldots,x[d-1]$ and $x[1],x[2],\ldots,x[d]$, respectively. If the minimal value is obtained on a symbol $x_{min}= x[i_{min}]$ which is queried by both, then we have $\mathrm{Basic}_d(x) = i_{min}$ and $\mathrm{Basic}_d(x) =i_{min}-1$, implying that $\mathrm{Basic}_d(x) - \mathrm{Basic}_d(x \ll 1) = 1$ and the executions are successful. Similarly, an error occurs when the minimal value on the symbols $x[0],\ldots,x[d]$ is obtained on a symbol computed only by one party, namely on one of the $2$ symbols $x[0]$ or $x[d]$. Since the symbols are uniform (and since $x[0],\ldots,x[d]$ are distinct, except with probability at most $1/n$), this occurs with probability $2/(1+d)$ and the lemma follows.
\end{proof}

An important quantity that plays a role in the more advanced LPHS constructions is the output difference of the executions in case they fail to synchronize on the same symbol $x[i_{min}]$ (i.e., their output difference is not $1$). Clearly, we have $\left| \mathrm{Basic}_d(x) - \mathrm{Basic}_d(x \ll 1) - 1 \right| \leq d+1$.

\begin{algorithm}
\Begin{
$i \leftarrow 0$, $min \leftarrow \infty$;

\While{$i < d$}{
    $t  \leftarrow x[i]$;

    \If{$t < min$}{

        $i_{min}  \leftarrow i$, $min \leftarrow t$;
    }

    $i \leftarrow i + 1$;
}
Output $i_{min}$;
}
\caption{$\mathrm{Basic}_d(x)$}
\label{alg:basic}
\end{algorithm}

\subsection{The Random Walk LPHS}
\label{sec:randomw}

The random walk LPHS is useful when the shift is bounded by $R$, which is (much) bigger than 1 (see Definition~\ref{def-morelphs}). We think of two LPHS executions as two parties $A$ and $B$ that perform (pseudo) random walks on the symbols of the string $x \in \Sigma_b^n$, starting from $x[0]$ and $x[r]$, respectively. For a parameter $L$, the step length of each party is uniformly distributed in $[1,L-1]$, and is determined by a shared random function $\psi_{L-1}: \Sigma_b \rightarrow [1,L-1]$, as $\psi_{L-1}(x[i])$.


Algorithm~\ref{alg:random} describes the random walk LPHS, parameterized by $(L,d)$ which determine the maximal step length and the number of steps, respectively. The algorithm closely resembles Pollard's ``kangaroo'' random walk algorithm for solving the discrete logarithm in an interval problem using limited memory~\cite{Pollard78}.
\begin{lemma}[\cite{DKK18}, Lemma 6 (adapted)]
\label{lem:big_shift}
For any positive integer $r \leq R$,
\[
\pr_{x}[\mathrm{RW}_{L,d}(x) - \mathrm{RW}_{L,d}(x \ll r) \neq r] = O\left(\frac{r/L + L}{d}\right).
\]
In particular, a choice of $L = \sqrt{R}$ gives
\[
\pr_{x}[\mathrm{RW}_{L,d}(x) - \mathrm{RW}_{L,d}(x \ll r) \neq r] = O\left(\frac{\sqrt{R}}{d}\right).
\]
\end{lemma}
\begin{proof}[sketch]
Similarly to the proof of Lemma~\ref{lem:basic_fail_prob}, if the minimal value in both walks is obtained on the same symbol $x[j_{min}]$ (queried by both), then $\mathrm{RW}_{L,d}(x) - \mathrm{RW}_{L,d}(x \ll r) = r$. Otherwise, we say that the parties err, and our goal is to upper bound the error probability.

Suppose that $A$ lands on a symbol queried by $B$'s walk after performing $M$ steps. From this stage, the walks coincide on the remaining $r-M$ steps. Since the symbols are uniform, then the probability that the minimal value in both executions is not obtained on the same symbol and the parties err is $O(M/d)$. In the following, we argue that the expected value of $M$ is $O(r/L + L)$. A formal proof then shows that $M$ is tightly concentrated around its expectation, which is sufficient to upper bound the error probability by $O(M/d) = O(\frac{r/L + L}{d})$, as required.

In order to estimate the expectation of $M$, partition $A$'s walk into two stages, where the first stage ends when $A$'s walk reaches (or goes beyond) $B$'s starting point. Since the initial distance between the parties is $r$ and step size of $A$ is uniform in $[1,L-1]$, the expected number of steps in the first stage is $O(r/L)$. In the second stage, due to the uniformity of $\psi_{L-1}$, each step of $A$ has probability of at least $1/L$ to land on a symbol queried by $B$'s walk. Therefore, the expected number of steps until this event occurs is $O(L)$.
\end{proof}

We further note that the parties travel a distance of $O(L \cdot d)$. In case the parties err, then their final distance is at most $O(R + L \cdot d)$, which evaluates to $O(R + \sqrt{R} \cdot d)$ for $L = \sqrt{R}$, and $O(\sqrt{R} \cdot d)$ when $d = \Omega(\sqrt{R})$.

\begin{algorithm}
\Begin{
$j \leftarrow 0$, $i \leftarrow 0$, $min \leftarrow \infty$;

\While{$i < d$}{
    $t  \leftarrow x[j]$;

    \If{$t < min$}{

        $j_{min}  \leftarrow j$;

        $min \leftarrow t$;

    }

    $j \leftarrow j + \psi_{L-1}(x[j])$; {\label{step:compute}}

    $i \leftarrow i + 1$;
}
Output $j_{min}$;
}
\caption{$\mathrm{RW}_{L,d}(x)$}
\label{alg:random}
\end{algorithm}

\subsection{The Iterated Random Walk LPHS}
\label{sec:iterated}

The starting point of the Iterated Random Walk (IRW) LPHS is Algorithm~\ref{alg:basic}. It makes $d$ queries and fails with probability of roughly $2/d$ as noted above. Let us assume that we run this algorithm with only $d/2$ queries, which increases the error probability by a factor of 2 to about $4/d$. On the other hand, we still have a budget of $d/2$ queries and we can exploit them to reduce the error probability.

It would be instructive to think about the executions $\mathrm{Basic}_d(x)$ and $\mathrm{Basic}_d(x \ll 1)$ as being performed by two parties $A$ and $B$ (respectively). After the first $d/2$ queries, we say that $A$ (or $B$) is placed at index $i$ if $x[i]$ is the minimal value in its computed set of size $d/2$. Assume that $A$ and $B$ fail to synchronize on the same index after the first $d/2$ queries (which occurs with probability of roughly $4/d$). Then, as noted above, they are placed at symbols which are at distance of at most $d/2 + 1$, i.e., if $A$ is placed at $x[i]$ and $B$ is placed at $x[j]$, then $|i-j| \leq d/2 + 1$.

Next, the parties then use the random walk of Algorithm~\ref{alg:random} to try and synchronize.\footnote{For simplicity, in Algorithm~\ref{alg:random} $j$ is initialized to 0 rather than to a number that depends on previous queries. This is dealt with in the IRW algorithm by appropriately shifting the input $x$ itself.}
Note that since both $A$ and $B$ use the same algorithm (which is deterministic given the shared randomness), they remain synchronized if they already are at the beginning of the walks.

The random walk algorithm is applied using $d/2$ queries and step length of $\sqrt{d}$.
If the parties are not initially synchronized, according to Lemma~\ref{lem:big_shift}, the error probability of the random walk is $O(\sqrt{d}/d) = O(d^{-1/2})$ and the total error probability is $O(d^{-1} \cdot d^{-1/2} = d^{-3/2})$. In this case the distance between the parties is $O(\sqrt{d} \cdot d = d^{3/2})$.

The success probability can be further amplified by reserving an additional number of $O(d)$ queries to be used in another random walk. This is made possible by shortening the first two random walks, without affecting the failure probability significantly. Hence, assume that the parties fail to synchronize after the random walk (which occurs with probability of $O(d^{-3/2}$)) and that we still have enough available queries for another random walk with $O(d)$ steps. As in the previous random walk, the step length is about the square root of the initial distance of the parties, namely $\sqrt{d^{3/2}} = d^{3/4}$. Applying Lemma~\ref{lem:big_shift} with these parameters gives an error probability of $O(d^{3/4}/d) = d^{-1/4}$ for the random walk, and a total error probability of $O(d^{-3/2} \cdot d^{-1/4}) = O(d^{-7/4})$.

We continue executing random walk iterations with a carefully chosen step length (distributing a budget of $O(d)$ queries among them). After $k$ random walk iterations, the error probability is reduced to about $d^{-2 + 2^{-k}}$ (and the expected distance between the parties is roughly $d^{2 - 2^{-k}}$). Choosing $k \approx \log \log d$ gives an error probability of $\tilde{O}(d^{-2 + 1/\log d}) = \tilde{O}(d^{-2})$. Additional optimizations allow to reduce the error probability to $O(d^{-2})$. The IRW LPHS is presented in Algorithms~\ref{alg:random} and~\ref{alg:iterated}.

\subsubsection{Details of the Iterated Random Walk LPHS}

Algorithm~\ref{alg:iterated} describes the full protocol which is composed of application of the basic LPHS (using $d_0 < d$ queries, reserving queries for the subsequent random walks), and then $K$ additional random walks, where the $k$'th random walk is parameterized by $(L_i,d_i)$ which determine its maximal step length and number of steps. Between each two iterations in Step~\ref{step:shift}, both parties are moved forward by a large (deterministic) number of steps, in order to guarantee independence between the iterations. We are free to choose the parameters $K,\{L_i,d_i\}$, as long as $\sum_{i=0}^K d_i=d$ is satisfied.
By fine tuning the parameters which distribute the number of queries among the iterations and select the step length of each random walk, the following theorem is derived.
\begin{theorem}\cite{DKK18}, Theorem 2 (adapted) \label{thm:mainDDK}
There exists a parameter set $PS = (K,d_0,\{(L_i,d_i)_{i=1}^{K}\})$, where $d = \sum_{i=0}^K d_i$ for which
\begin{align*}
	\pr_{x}[\mathrm{IRW}_{PS}(x) - \mathrm{IRW}_{PS}(x \ll 1) \neq 1] \leq 2^{10.2+o(1)}/d^{2}.
\end{align*}
\end{theorem}
This theorem immediately implies Theorem~\ref{thm:1dupper}.

\begin{remark}[Cyclic vs. non-cyclic]
The random walk makes queries within an interval of size bounded by $O(d^2)$, hence if $n = \Omega(d^2)$ is large enough, the LPHS gives both a cyclic and non-cyclic LPHS with the same parameters.
\end{remark}

\begin{algorithm}
\Begin{

$j \leftarrow \mathrm{Basic}_{d_0}(x)$;

$p \leftarrow 0$;

$k \leftarrow 1$;

\While{$k \leq K$}{

    $J \leftarrow \sum_{i<k}d_{i} L_{i}$; {\label{step:shift}}

    $j \leftarrow j + J$;

    $p \leftarrow p + j$;

	$x \leftarrow x \ll j$;
	
    $j \leftarrow \mathrm{RW}_{L_i,d_i}(x)$;

	$i \leftarrow i+1$;
}
Output $p + j$;
}

\caption{$\mathrm{IRW}_{K,d_0,\{(L_i,d_i)_{i=1}^{K}\}}(x)$}
\label{alg:iterated}
\end{algorithm}


\subsubsection{The Cyclic Random Walk LPHS}
\label{sec:crandomw}

\begin{lemma}
\label{lem:crandomw}
There exists a cyclic $(n,b,d,\delta)$-LPHS with $d = \tilde{O}(n^{1/2})$ and $\delta = n^{-\omega(1)}$.
\end{lemma}
\begin{proof}[sketch]
We assume that $b = (\log n)^a$ for some $a > 1$ (e.g., $b = \log ^2 n$) such that the $n$ symbols of $x$ are all distinct with probability $1 - n^{-\omega(1)}$. If $b$ is smaller, then apply Lemma~\ref{lem:bigger-alphabet} and obtain a bigger alphabet.

The main idea is to apply Algorithm~\ref{alg:random} with parameters $L = \sqrt{n}$ and $d = \sqrt{n}$ repeatedly $m = (\log n)^a$ times and output the shift from the final $j_{min}$ location.

Once again, we think of two invocations of the LPHS as two parties as performing random walks.
We make sure that if the parties are synchronized at the beginning of application $i$ of Algorithm~\ref{alg:random}, then they remain synchronized. On the other hand, if they are not synchronized, they will agree on the same symbol with some constant probability $p$, independently of all other applications. Therefore, the total error probability is $p^m = n^{-\omega(1)}$.

In order to make the $m$ applications of the algorithm independent, in application $i \in \{1,2,\ldots,m\}$, replace symbol $j$ of $x$ with $\phi_i(x[j])$, where each $\phi_i: \Sigma_b \rightarrow \Sigma_b$ is an independent (shared) random permutation.
Furthermore, after application $i$, jump ($\bmod$ $n$) by $\rho_i(x[j_{min}])$ (and start the next application from this location), where $j_{min}$ is the output of the algorithm, and $\rho_i: \Sigma_b \rightarrow \mathbb{Z}_n$ is an independent (shared) random function.

Note that for every $0 < c < 1$, the probability that the parties start application $i$ at distance at most $c \cdot n$ is at least $1/c$ (we define the distance as the minimal cyclic distance between the parties). For a sufficiently small $c$, Algorithm~\ref{alg:random} (applied with $L = \sqrt{n}$ and $d = \sqrt{n}$) succeeds with constant probability.
\end{proof}

\subsubsection{Las Vegas LPHS for Big Shifts}
\label{sec:LasVegas}

\begin{lemma}\label{lem-lv}
For $R = O(d)$, $n = \Omega(d^2)$, there exists an LPHS $h$ that makes $O(d)$ queries on average such that
\[
\pr_{x}[\exists \, r \in [R+1] \,:\, h(x) - h(x \ll r) \neq r] = \tilde{O}\left(\frac{R}{d^2}\right),
\]
and for every $x$ such that $\exists \, r \in [R+1] \,:\, h(x) - h(x \ll r) \neq r$, we have $h(x) = \bot$.
\end{lemma}
The construction gives both a cyclic and non-cyclic LPHS with the same parameters.
We note that the lemma can be strengthened by proving a concentration inequality around the expected number of queries of $h$.

\begin{proof}
We assume that $b \geq 3 \log n$ (in case $b < 3 \log n$, the construction can be extended using Lemma~\ref{lem:bigger-alphabet} with logarithmic loss in the error probability).

The LPHS $h$ works in two stages. It first runs Algorithm~\ref{alg:basic} $R+1$ times:
$\mathrm{Basic}_d(x), \ldots,$ $\mathrm{Basic}_d(x \ll R)$. Note that these algorithms require a total of $R +d+ 1$ queries to $x[0],\ldots,x[R + d]$. Let $S = \{\mathrm{Basic}_d(x \ll r) \, : \, r \in [R+1]\}$. In the second stage, $h$ runs the IRW algorithm with $d$ queries on each input in $\{x \ll i \, | \, i \in S\}$, and if all runs agree on the same symbol $x[j]$, then it outputs $j$. Otherwise, it outputs $\bot$.

The bound $O(R/d^2)$ on the probability of $h$ outputting $\bot$ follows from the IRW analysis and Lemma~\ref{lem:shift}.
The number of queries of $h$ is $R + d + 1 + |S| \cdot d$, where $|S|$ is the size of $S$.
Consider a pair of executions $\mathrm{Basic}_d(x \ll i), \mathrm{Basic}_d(x \ll (i+1))$ for $i \in [R]$. According to the analysis of Algorithm~\ref{alg:basic}, $\mathrm{Basic}_d(x \ll (i+1))$ contributes a new element to $S$ (that is different from $\mathrm{Basic}_d(x \ll i)$) with probability $O(1/d)$. Hence the expected size of $S$ is $O(R/d) = O(1)$.
\end{proof}